% Generated by export-arxiv.ps1. Edit content.tex or the yj-* stack, then rerun.
% The yj-* class/package code below is materialized from the current local sources.
% Materialized from yj-paper.cls
\documentclass[12pt, reqno]{amsart}
\makeatletter
\ExplSyntaxOn

% Fix \part: center (amsart uses \raggedright), Roman numerals, bold number
\ExplSyntaxOff

\def\part{\@startsection{part}{0}%
  \z@{\linespacing\@plus\linespacing}{.5\linespacing}%
  {\normalfont\bfseries\centering}}
\def\yj@part{part}
\def\@seccntformat#1{%
  \expandafter\ifx\csname yj@#1\endcsname\yj@part
    \csname the#1\endcsname\@secnumpunct
  \else
    \protect\textup{\protect\@secnumfont
      \csname the#1\endcsname
      \protect\@secnumpunct
    }%
  \fi
}
\ExplSyntaxOn

% --- Begin yj-bib.sty ---
% Materialized from yj-bib.sty
\bool_new:N \g__yjbib_chinese_bool

\keys_define:nn { yj-bib }
  {
    chinese .bool_gset:N = \g__yjbib_chinese_bool,
    chinese .default:n   = true,
  }

\usepackage[
  backend=biber,
  backref=true,
  bibstyle=alphabetic,
  citestyle=alphabetic,
  isbn=false,
  doi=false,
  url=false,
  giveninits=true,
  maxnames=4,
  minnames=4,
  maxalphanames=4,
  minalphanames=4,
  date=year,
  labeldate=year
]{biblatex}
\usepackage{biblatex-shortfields}
% Load the AMS modifier after biblatex-shortfields so the policy layer selects
% short fields while the visual style retains its series punctuation, for
% example "GTM, no. 52". Keeping this as a package also lets export-arxiv.ps1
% inline the local formatting without shipping a separate .bbx file.

% --- Begin yj-amsalpha.sty ---
% This file may be distributed and/or modified under the conditions of the
% LaTeX Project Public License, version 1.3c or later.
% Materialized from yj-amsalpha.sty
\ExecuteBibliographyOptions{
  abbreviate=true,
  doi=false,
  eprint=true,
  isbn=false,
  maxbibnames=999,
  related=false,
  url=false
}

\renewcommand*{\subtitlepunct}{\addcolon\space}
\setlength{\biblabelsep}{\labelsep}

\DeclareFieldFormat{yjams:titlecase}{\MakeSentenceCase*{#1}}
\DeclareFieldFormat{journaltitle}{#1\isdot}
\DeclareFieldFormat{booktitle}{#1}
\DeclareFieldFormat[article,periodical]{volume}{\mkbibbold{#1}}
\DeclareFieldFormat[article,periodical]{pages}{#1}

% Let the URL formatter supply breakpoints within long arXiv identifiers and
% retain a natural breakpoint before the optional subject classification. Use
% an explicit URI action because \href's file heuristic misclassifies legacy
% identifiers such as hep-th/0012136 when the linked text spans multiple lines.
\DeclareFieldFormat{eprint:arxiv}{%
  arXiv\addcolon\space
  \ifhyperref
    {\hyper@linkurl{%
       \nolinkurl{#1}%
       \iffieldundef{eprintclass}
         {}
         {\addspace\texttt{\mkbibbrackets{\thefield{eprintclass}}}}}%
       {https://arxiv.org/\abx@arxivpath/#1}}
    {\nolinkurl{#1}%
     \iffieldundef{eprintclass}
       {}
       {\addspace\texttt{\mkbibbrackets{\thefield{eprintclass}}}}}}
\DeclareFieldAlias{eprint:arXiv}{eprint:arxiv}

\DeclareFieldFormat
  [article,periodical,book,mvbook,inbook,bookinbook,suppbook,
   collection,mvcollection,incollection,suppcollection,
   proceedings,mvproceedings,inproceedings,
   reference,mvreference,inreference]
  {number}{\bibstring{number}\addnbspace #1}

\DeclareFieldFormat
  [article,inbook,incollection,inproceedings,patent,thesis,unpublished]
  {title}{\mkbibemph{#1}}

\renewbibmacro*{in:}{}

\renewbibmacro*{title}{%
  \ifboolexpr{
    test {\iffieldundef{title}}
    and
    test {\iffieldundef{subtitle}}
  }
    {}
    {\printtext[title]{%
       \printfield[yjams:titlecase]{title}%
       \setunit{\subtitlepunct}%
       \printfield[yjams:titlecase]{subtitle}}%
     \newunit}%
  \printfield{titleaddon}}

\renewbibmacro*{maintitle}{%
  \ifboolexpr{
    test {\iffieldundef{maintitle}}
    and
    test {\iffieldundef{mainsubtitle}}
  }
    {}
    {\printtext[maintitle]{%
       \printfield[yjams:titlecase]{maintitle}%
       \setunit{\subtitlepunct}%
       \printfield[yjams:titlecase]{mainsubtitle}}%
     \newunit}%
  \printfield{maintitleaddon}}

% BibLaTeX's current data model distinguishes a work's volume from its
% number within a publication series. Prefer the abbreviated series title
% when supplied and never reinterpret volume as a series number.
\renewbibmacro*{series+number}{%
  \iffieldundef{shortseries}
    {\printfield[noformat]{series}}
    {\printfield[noformat]{shortseries}}%
  \setunit*{\addcomma\space}%
  \printfield{number}%
  \newunit}

\renewbibmacro*{maintitle+title}{%
  \iffieldsequal{maintitle}{title}
    {\clearfield{maintitle}%
     \clearfield{mainsubtitle}%
     \clearfield{maintitleaddon}}
    {\iffieldundef{maintitle}
       {}
       {\usebibmacro{maintitle}%
        \newunit\newblock
        \iffieldundef{volume}
          {}
          {\printfield{volume}%
           \printfield{part}%
           \setunit{\subtitlepunct}}}}%
  \usebibmacro{title}%
  \newunit}

\renewbibmacro*{maintitle+booktitle}{%
  \iffieldundef{maintitle}
    {}
    {\usebibmacro{maintitle}%
     \newunit\newblock
     \iffieldundef{volume}
       {}
       {\printfield{volume}%
        \printfield{part}%
        \setunit{\subtitlepunct}}}%
  \usebibmacro{booktitle}%
  \newunit}

\renewbibmacro*{journal+issuetitle}{%
  \usebibmacro{journal}%
  \setunit*{\addspace}%
  \printfield{volume}%
  \setunit{\addspace}%
  \iffieldundef{year}
    {}
    {\printtext[parens]{\printfield{year}}}%
  \newunit
  \printfield{number}%
  \setunit{\bibeidpunct}%
  \printfield{eid}%
  \newunit}

\renewbibmacro*{publisher+location+date}{%
  \printlist{publisher}%
  \setunit*{\addcomma\space}%
  \printlist{location}%
  \setunit*{\addcomma\space}%
  \usebibmacro{date}%
  \newunit}

\newbibmacro*{yjams:container}{%
  \usebibmacro{maintitle+booktitle}%
  \ifboolexpr{
    test {\iffieldundef{venue}}
    and
    test {\iffieldundef{eventyear}}
  }
    {}
    {\setunit{\addspace}%
     \printtext[parens]{%
       \printfield{venue}%
       \setunit*{\addcomma\space}%
       \printeventdate}}%
  \ifnameundef{editor}
    {}
    {\setunit{\addspace}%
     \printtext[parens]{%
       \printnames{editor}%
       \setunit{\printdelim{editortypedelim}}%
       \usebibmacro{editorstrg}}}%
  \newunit}

\DeclareBibliographyDriver{book}{%
  \usebibmacro{bibindex}%
  \usebibmacro{begentry}%
  \usebibmacro{author/editor+others/translator+others}%
  \setunit{\printdelim{nametitledelim}}\newblock
  \usebibmacro{maintitle+title}%
  \newunit
  \printlist{language}%
  \newunit\newblock
  \usebibmacro{byauthor}%
  \newunit\newblock
  \usebibmacro{byeditor+others}%
  \newunit\newblock
  \printfield{edition}%
  \newunit
  \iffieldundef{maintitle}
    {\printfield{volume}%
     \printfield{part}}
    {}%
  \newunit
  \printfield{volumes}%
  \newunit\newblock
  \usebibmacro{series+number}%
  \newunit\newblock
  \printfield{note}%
  \newunit\newblock
  \usebibmacro{publisher+location+date}%
  \newunit\newblock
  \usebibmacro{chapter+pages}%
  \newunit\newblock
  \usebibmacro{doi+eprint+url}%
  \newunit\newblock
  \usebibmacro{addendum+pubstate}%
  \setunit{\bibpagerefpunct}\newblock
  \usebibmacro{pageref}%
  \usebibmacro{finentry}}

\DeclareBibliographyDriver{incollection}{%
  \usebibmacro{bibindex}%
  \usebibmacro{begentry}%
  \usebibmacro{author/translator+others}%
  \setunit{\printdelim{nametitledelim}}\newblock
  \usebibmacro{title}%
  \setunit{\addperiod\space}\newblock
  \usebibmacro{yjams:container}%
  \newunit
  \iffieldundef{maintitle}
    {\printfield{volume}%
     \printfield{part}}
    {}%
  \newunit\newblock
  \usebibmacro{series+number}%
  \newunit\newblock
  \usebibmacro{publisher+location+date}%
  \newunit\newblock
  \printfield{edition}%
  \newunit
  \printfield{note}%
  \newunit\newblock
  \usebibmacro{chapter+pages}%
  \newunit\newblock
  \usebibmacro{doi+eprint+url}%
  \newunit\newblock
  \usebibmacro{addendum+pubstate}%
  \setunit{\bibpagerefpunct}\newblock
  \usebibmacro{pageref}%
  \usebibmacro{finentry}}

\DeclareBibliographyDriver{inproceedings}{%
  \usebibmacro{bibindex}%
  \usebibmacro{begentry}%
  \usebibmacro{author/translator+others}%
  \setunit{\printdelim{nametitledelim}}\newblock
  \usebibmacro{title}%
  \setunit{\addperiod\space}\newblock
  \usebibmacro{yjams:container}%
  \newunit
  \iffieldundef{maintitle}
    {\printfield{volume}%
     \printfield{part}}
    {}%
  \newunit\newblock
  \usebibmacro{series+number}%
  \newunit\newblock
  \printlist{organization}%
  \newunit
  \usebibmacro{publisher+location+date}%
  \newunit\newblock
  \printfield{note}%
  \newunit\newblock
  \usebibmacro{chapter+pages}%
  \newunit\newblock
  \usebibmacro{doi+eprint+url}%
  \newunit\newblock
  \usebibmacro{addendum+pubstate}%
  \setunit{\bibpagerefpunct}\newblock
  \usebibmacro{pageref}%
  \usebibmacro{finentry}}

\DeclareBibliographyAlias{conference}{inproceedings}
% --- End yj-amsalpha.sty ---

\addbibresource{references.bib}

% URLs, DOIs, and eprint identifiers are data-owned strings. BibLaTeX already
% permits breaks at punctuation; also permit a modestly penalized break after
% digits, with micro-stretch available at that semantic breakpoint.
\setcounter{biburlnumpenalty}{100}
\biburlnumskip=0mu plus 1mu\relax

% AMS and the standard article class relax line breaking locally in the
% bibliography because its data-owned prose cannot be rephrased. Use the same
% fully justified fallback only inside BibLaTeX's bibliography group; do not
% change warning thresholds or apply it to the manuscript body.
\AtBeginBibliography{\sloppy}

% Separate compact footnote-sized back-reference page numbers with a period
% while retaining the hyperlinks supplied by BibLaTeX's pageref formatter.
\renewcommand*{\bibpagerefpunct}{\addperiod\space}
\renewbibmacro*{pageref}{%
  \iflistundef{pageref}
    {}
    {\printtext{%
       {\footnotesize
        \printlist[pageref][-\value{listtotal}]{pageref}}}%
     \nopunct}}

% Force citation labels upright regardless of surrounding font shape
% (e.g. inside italic theorem bodies from \theoremstyle{plain}).
\renewcommand{\mkbibbrackets}[1]{\textup{[#1]}}

% Force uppercase initials after hyphens in given names.
% Biber preserves source case (e.g., Wu-yen -> W.-y.), but initials
% should be capitalized (W.-Y.).

% --- Chinese bibliography headings ---
\bool_if:NT \g__yjbib_chinese_bool
  {
    \AddToHook{begindocument}
      {
        \defbibheading{bibliography}[\textbf{參考資料}]
          { \section*{#1} }
        \defbibheading{bibintoc}[\textbf{參考資料}]
          { \section*{#1} \addcontentsline{toc}{section}{#1} }
      }
  }
% --- End yj-bib.sty ---

% --- Begin yj-preamble.sty ---
% Materialized from yj-preamble.sty
% Classes using yj-bib must load it before this bundle because yj-refs
% transitively loads hyperref.
\PassOptionsToPackage{dvipsnames}{xcolor}

% --- Begin yj-graphics.sty ---
% Materialized from yj-graphics.sty
\usepackage{pgfplots}
\pgfplotsset{compat=1.17}
\usepackage{quiver}
\ExplSyntaxOff
\usetikzlibrary{babel,calc,positioning,arrows.meta,intersections}
\ExplSyntaxOn
% --- End yj-graphics.sty ---

% --- Begin yj-math.sty ---
% Materialized from yj-math.sty
\usepackage{amsmath}
\allowdisplaybreaks
\usepackage{mathtools}
% Use eqnlines only through its explicit equations environment.  Keep the
% established amsmath/mathtools environments, matrices, and \eqref unchanged.

\usepackage{amssymb}
% Font ownership invariant: do not load math-font packages here.
% yj-fonts loads newtxtext before newtxmath so math operators match the text font.
\usepackage{stmaryrd}
% stmaryrd does not ship bold fonts; keep bold math from requesting U/stmry/b/n.
\SetSymbolFont{stmry}{bold}{U}{stmry}{m}{n}
\usepackage{dsfont}

% --- Operators ---

\DeclareMathOperator{\rmRe}{Re}
\DeclareMathOperator{\rmIm}{Im}

\clist_map_inline:nn
{
  supp, id, Id, rank, End, Aut, im, sgn, Hom, Ob, Pow, nullity, Gal, Irr,
  Coh, Per,
  Alt, Tr, GL, Sym, Spec, coker, SO, Gr, SU, Rad, Map, Res, Ind, Cl, FS,
  Bij, SL, Frac, Li, Ai, Log, ind, sech, csch, disc, Sp, arsinh, artanh,
  lcm, Ann, Ker, Sub, dist, Int,
  colim, vol, Coker, tr,
  dom, ran, cl, sign, Ad, DT
}
{ \exp_args:Nc \DeclareMathOperator {#1} {#1} }

% --- Delimiters ---
% newtxmath's extensible vertical bar can be taller than the surrounding
% delimiters at some automatically selected sizes.  For starred paired
% delimiters, use the font's native bar while it is large enough.  Only when
% TeX requests a larger delimiter do we draw a centered rule of that requested
% extent.  The rule is smashed so it does not affect the outer delimiter
% calculation.
\cs_new_nopar:Npn \__yj_math_middle: { \middle }
\cs_new_protected:Npn \__yj_middle_bar:nn #1#2
{
  \cs_if_eq:NNTF \delimsize \__yj_math_middle:
    { \mathpalette \__yj_auto_middle_bar:nn { {#1}{#2} } }
    { \mathrel { \delimsize\vert } }
}
\cs_new_protected:Npn \__yj_auto_middle_bar:nn #1#2
  { \__yj_auto_middle_bar_aux:nnn {#1} #2 }
\cs_new_protected:Npn \__yj_auto_middle_bar_aux:nnn #1#2#3
{
  \hbox_set:Nn \l_tmpa_box { $\m@th#1#2#3$ }
  \hbox_set:Nn \l_tmpb_box { $\m@th#1\vcenter{\hbox{}}$ }
  \dim_set:Nn \l_tmpa_dim
    { \box_ht:N \l_tmpa_box - \box_ht:N \l_tmpb_box }
  \dim_set:Nn \l_tmpb_dim
    { \box_dp:N \l_tmpa_box + \box_ht:N \l_tmpb_box }
  \dim_compare:nNnT { \l_tmpa_dim } < { \l_tmpb_dim }
    { \dim_set_eq:NN \l_tmpa_dim \l_tmpb_dim }
  \dim_set:Nn \l_tmpb_dim
    { \l_tmpa_dim * \delimiterfactor / 500 }
  \dim_set:Nn \l_tmpa_dim
    { 2 \l_tmpa_dim - \delimitershortfall }
  \dim_compare:nNnT { \l_tmpb_dim } < { \l_tmpa_dim }
    { \dim_set_eq:NN \l_tmpb_dim \l_tmpa_dim }
  \hbox_set:Nn \l_tmpa_box { $\m@th#1\vert$ }
  \dim_set:Nn \l_tmpa_dim
    { \box_ht:N \l_tmpa_box + \box_dp:N \l_tmpa_box }
  \dim_compare:nNnTF { \l_tmpb_dim } > { \l_tmpa_dim }
    {
      \mathrel
        {
          \smash
            {
              \vcenter
                {
                  \hbox
                    {
                      \vrule
                        width \fontdimen 8 \textfont 3
                        height \l_tmpb_dim
                    }
                }
            }
        }
    }
    { \mathrel { \vert } }
}
\DeclarePairedDelimiterX\setc[2]{\{}{\}}{
  \,#1 \__yj_middle_bar:nn {#1}{#2} #2\,
}
\DeclarePairedDelimiterX\vargs[2]{(}{)}{
  #1 \__yj_middle_bar:nn {#1}{#2} #2
}
\DeclarePairedDelimiter{\abs}{\lvert}{\rvert}
\DeclarePairedDelimiter{\norm}{\lVert}{\rVert}

\DeclarePairedDelimiterX{\angvargs}[2]{\langle}{\rangle}{
  #1 \__yj_middle_bar:nn {#1}{#2} #2
}

% The `c' in \cdbr means "composed": it builds [\![...\]\!] from
% two ordinary square brackets, unlike \dbr which uses \lBrack/\rBrack.
\NewDocumentCommand{\cdbr}{m}{\mathopen{[\![}#1\mathclose{]\!]}}

% --- Utilities ---
\usepackage{centernot}
\NewDocumentCommand{\ndiv}{O{}}{\mathrel{\centernot{#1\vert}}}

\usepackage{mleftright}

\cs_new_protected:Nn \__yj_math_cycle:nn
{
  \seq_set_split:Nnn \l_tmpa_seq { , } {#2}
  (
    \int_zero:N \l_tmpa_int
    \seq_map_inline:Nn \l_tmpa_seq
    {
      \int_incr:N \l_tmpa_int
      \int_compare:nNnF { \l_tmpa_int } = { 1 } {#1}
      ##1
    }
  )
}

\NewDocumentCommand{\cycle}{O{\;} m}
{ \__yj_math_cycle:nn {#1} {#2} }
% --- End yj-math.sty ---

% --- Begin yj-fonts.sty ---
% Materialized from yj-fonts.sty
% --- Text and math fonts ---
% Keep newtxtext before newtxmath. Loading newtxmath earlier makes math
% operators fall back to Latin Modern instead of the newtx/Termes text family.
\usepackage{newtxtext}
\defaultfontfeatures{}

\usepackage[varvw]{newtxmath}
\usepackage{anyfontsize}

% --- Symbol overrides (CM large operators) ---
\DeclareSymbolFont{largesymbolsCM}{OMX}{cmex}{m}{n}

\let\sum\relax
\DeclareMathSymbol{\sum}{\mathop}{largesymbolsCM}{"50}

\let\prod\relax
\DeclareMathSymbol{\prod}{\mathop}{largesymbolsCM}{"51}

\let\coprod\relax
\DeclareMathSymbol{\coprod}{\mathop}{largesymbolsCM}{"60}

\let\rtimes\relax
\DeclareMathSymbol{\rtimes}{\mathbin}{AMSb}{"6F}

% --- Math alphabets ---
\usepackage[cal=cm, bb=ams]{mathalfa}
\DeclareMathAlphabet{\eulerScr}{U}{eus}{m}{n}
% --- End yj-fonts.sty ---

% Standard table environments and column rules are document-type independent.
\usepackage{array}
\usepackage{booktabs}
\usepackage{longtable}
\usepackage{tabularx}
% Hyperref installs its float-anchor patch only when float is already loaded.
\usepackage{float}

% --- Begin yj-refs.sty ---
% Materialized from yj-refs.sty
\IfPackageLoadedF{hyperref}
  { \PassOptionsToPackage{pdfusetitle}{hyperref} }
\usepackage{xr-hyper}
\usepackage{hyperref}
\tl_const:Nn \c__yjrefs_shade_tl { 85 }
\hypersetup{
  linkcolor  = YellowOrange!\c__yjrefs_shade_tl!black,
  citecolor  = Aquamarine!\c__yjrefs_shade_tl!black,
  urlcolor   = violet!\c__yjrefs_shade_tl!black,
  colorlinks = true,
}

% Classes with nonstandard title/author storage register their metadata sources
% here. Explicit pdftitle/pdfauthor values remain authoritative.
\tl_new:N \g__yjrefs_pdf_title_source_tl
\tl_new:N \g__yjrefs_pdf_author_source_tl
\tl_new:N \l__yjrefs_pdf_title_tl
\tl_new:N \l__yjrefs_pdf_author_tl
\bool_new:N \g__yjrefs_custom_pdf_metadata_bool

\cs_new_protected:Npn \yjrefs_set_pdf_metadata_sources:nn #1#2
  {
    \tl_gset:Nn \g__yjrefs_pdf_title_source_tl { #1 }
    \tl_gset:Nn \g__yjrefs_pdf_author_source_tl { #2 }
    \bool_gset_true:N \g__yjrefs_custom_pdf_metadata_bool
  }

\pdfstringdefDisableCommands
  {
    \def\\{~}
    \let\newline\\
    \def\and{;~}
    \def\thanks#1{}
    \def\orcid#1{}
    \let\cjk\@empty
  }

\cs_new_protected:Npn \__yjrefs_apply_pdf_metadata:
  {
    \bool_if:NT \g__yjrefs_custom_pdf_metadata_bool
      {
        % Materialize the registered source macros before \maketitle can clear
        % them. Bare math shifts are presentation-only and can be dropped from
        % otherwise plain metadata such as ``$K$-type''.
        \exp_args:NNV \tl_set:No
          \l__yjrefs_pdf_title_tl \g__yjrefs_pdf_title_source_tl
        \exp_args:NNV \tl_set:No
          \l__yjrefs_pdf_author_tl \g__yjrefs_pdf_author_source_tl
        \tl_remove_all:Nn \l__yjrefs_pdf_title_tl { $ }
        \tl_remove_all:Nn \l__yjrefs_pdf_author_tl { $ }
        \tl_if_empty:NT \@pdftitle
          {
            \hypersetup
              { pdftitle = { \tl_use:N \l__yjrefs_pdf_title_tl } }
          }
        \tl_if_empty:NT \@pdfauthor
          {
            \hypersetup
              { pdfauthor = { \tl_use:N \l__yjrefs_pdf_author_tl } }
          }
      }
  }

\AddToHook{begindocument/before}{\__yjrefs_apply_pdf_metadata:}

\usepackage{zref-clever}
\zcsetup{cap, nameinlink=tsingle, noabbrev}
\usepackage{bookmark}
\usepackage{xurl}
% --- End yj-refs.sty ---

% --- Begin yj-theorems.sty ---
% Materialized from yj-theorems.sty
\bool_new:N \g__yjthm_chinese_bool
\bool_new:N \g__yjthm_noparent_bool

\keys_define:nn { yj-theorems }
  {
    chinese  .bool_gset:N = \g__yjthm_chinese_bool,
    chinese  .default:n   = true,
    noparent .bool_gset:N = \g__yjthm_noparent_bool,
    noparent .default:n   = true,
  }

\usepackage{zref-clever}
\usepackage{keytheorems}

\bool_if:NTF \g__yjthm_chinese_bool
  {
    \theoremstyle{plain}
    \bool_if:NTF \g__yjthm_noparent_bool
      {
        \newkeytheorem{theorem}[
          name=定理,
          refname={定理,定理}, Refname={定理,定理}
        ]
      }
      {
        \newkeytheorem{theorem}[
          name=定理, parent=section,
          refname={定理,定理}, Refname={定理,定理}
        ]
      }
    \newkeytheorem{corollary}[
      name=推論, sibling=theorem,
      refname={推論,推論}, Refname={推論,推論}
    ]
    \newkeytheorem{conjecture}[
      name=猜想, sibling=theorem,
      refname={猜想,猜想}, Refname={猜想,猜想}
    ]
    \newkeytheorem{hypothesis}[
      name=假設, sibling=theorem,
      refname={假設,假設}, Refname={假設,假設}
    ]
    \newkeytheorem{lemma}[
      name=引理, sibling=theorem,
      refname={引理,引理}, Refname={引理,引理}
    ]
    \newkeytheorem{property}[
      name=性質, sibling=theorem,
      refname={性質,性質}, Refname={性質,性質}
    ]
    \newkeytheorem{proposition}[
      name=命題, sibling=theorem,
      refname={命題,命題}, Refname={命題,命題}
    ]
    \newkeytheorem{proposition*}[name=命題, numbered=false]
    \newkeytheorem{claim}[name=斷言, numbered=false]
    \newkeytheorem{fact}[
      name=事實, sibling=theorem,
      refname={事實,事實}, Refname={事實,事實}
    ]

    \theoremstyle{definition}
    \newkeytheorem{exercise}[
      name=習題, sibling=theorem,
      refname={習題,習題}, Refname={習題,習題}
    ]
    \newkeytheorem{exercise*}[name=習題, numbered=false]
    \newkeytheorem{example}[
      name=範例, sibling=theorem,
      refname={範例,範例}, Refname={範例,範例}
    ]
    \newkeytheorem{example*}[name=範例, numbered=false]
    \newkeytheorem{definition}[
      name=定義, sibling=theorem,
      refname={定義,定義}, Refname={定義,定義}
    ]
    \newkeytheorem{problem}[
      name=問題, sibling=theorem,
      refname={問題,問題}, Refname={問題,問題}
    ]

    \theoremstyle{remark}
    \newkeytheorem{remark}[
      name=備註, sibling=theorem,
      refname={備註,備註}, Refname={備註,備註}
    ]
    \newkeytheorem{remark*}[name=備註, numbered=false]
    \newkeytheorem{theorem*}[name=定理, numbered=false]
    \newkeytheorem{fact*}[name=事實, numbered=false]
    \newkeytheorem{definition*}[name=定義, numbered=false]
    \newkeytheorem{lemma*}[name=引理, numbered=false]
  }
  {
    \theoremstyle{plain}
    \bool_if:NTF \g__yjthm_noparent_bool
      { \newkeytheorem{theorem} }
      { \newkeytheorem{theorem}[parent=section] }
    \newkeytheorem{corollary}[sibling=theorem]
    \newkeytheorem{conjecture}[
      sibling=theorem,
      refname={conjecture,conjectures},
      Refname={Conjecture,Conjectures}
    ]
    \newkeytheorem{hypothesis}[
      sibling=theorem,
      refname={hypothesis,hypotheses},
      Refname={Hypothesis,Hypotheses}
    ]
    \newkeytheorem{lemma}[sibling=theorem]
    \newkeytheorem{property}[
      sibling=theorem,
      refname={property,properties},
      Refname={Property,Properties}
    ]
    \newkeytheorem{proposition}[sibling=theorem]
    \newkeytheorem{proposition*}[name=Proposition, numbered=false]
    \newkeytheorem{claim}[name=Claim, numbered=false]
    \newkeytheorem{fact}[
      sibling=theorem,
      refname={fact,facts},
      Refname={Fact,Facts}
    ]

    \theoremstyle{definition}
    \newkeytheorem{exercise}[
      sibling=theorem,
      refname={exercise,exercises},
      Refname={Exercise,Exercises}
    ]
    \newkeytheorem{exercise*}[name=Exercise, numbered=false]
    \newkeytheorem{example}[sibling=theorem]
    \newkeytheorem{example*}[name=Example, numbered=false]
    \newkeytheorem{definition}[sibling=theorem]
    \newkeytheorem{problem}[
      sibling=theorem,
      refname={problem,problems},
      Refname={Problem,Problems}
    ]

    \theoremstyle{remark}
    \newkeytheorem{remark}[sibling=theorem]
    \newkeytheorem{remark*}[name=Remark, numbered=false]
    \newkeytheorem{theorem*}[name=Theorem, numbered=false]
    \newkeytheorem{fact*}[name=Fact, numbered=false]
    \newkeytheorem{definition*}[name=Definition, numbered=false]
    \newkeytheorem{lemma*}[name=Lemma, numbered=false]
  }

% --- Custom instance key: number ---
\keys_define:nn { keytheorems/thmuse }
  {
    number .code:n =
      { \setcounter{theorem}{ \int_eval:n { #1 - 1 } } },
  }

% --- Shared counters ---
\cs_set_eq:NN \c@equation \c@theorem
\cs_set_eq:NN \c@figure \c@theorem

% --- Chinese cross-reference type names (zref-clever) ---
\bool_if:NT \g__yjthm_chinese_bool
  {
    % Section: "第 X 節" format — refbounds puts 節 after the number
    \zcRefTypeSetup{section}{
      Name-sg = {第}, name-sg = {第},
      Name-pl = {第}, name-pl = {第},
      refbounds            = {,,{},},
      +refbounds-rb        = {,,{},},
      +refbounds-re        = {,,{},},
      refbounds-first-sg   = {,,{節},},
      +refbounds-last      = {,,{節},},
      refbounds-last-re    = {,,{節},},
    }
    \zcRefTypeSetup{subsection}{
      Name-sg = 小節, name-sg = 小節,
      Name-pl = 小節, name-pl = 小節,
    }
    \zcRefTypeSetup{equation}{
      Name-sg = 式, name-sg = 式,
      Name-pl = 式, name-pl = 式,
    }
    \zcRefTypeSetup{figure}{
      Name-sg = 圖, name-sg = 圖,
      Name-pl = 圖, name-pl = 圖,
    }
    \zcRefTypeSetup{table}{
      Name-sg = 表, name-sg = 表,
      Name-pl = 表, name-pl = 表,
    }
  }

% --- Chinese proof name ---
\bool_if:NT \g__yjthm_chinese_bool
  {
    \AddToHook{begindocument}
      {
        \cs_if_exist:NT \proofname
          { \renewcommand*{\proofname}{\textbf{證明：}\nopunct} }
      }
  }

% --- Solution environment ---
\bool_if:NTF \g__yjthm_chinese_bool
  {
    \NewDocumentEnvironment{solution}{O{\textbf{解答：}\nopunct}}
      { \begin{proof}[#1] }
      { \end{proof} }
  }
  {
    \NewDocumentEnvironment{solution}{O{Solution}}
      { \begin{proof}[#1] }
      { \end{proof} }
  }
% --- End yj-theorems.sty ---

% --- General formatting ---
\usepackage{comment}

\usepackage{enumitem}
\setlist[enumerate,1]{
  font=\normalfont,
  label=\textup{(\roman*)},
  listparindent=\parindent
}
% --- End yj-preamble.sty ---

% Keep body text and references fully justified while improving paragraph fit
% through glyph expansion and margin protrusion rather than loose global
% tolerance settings or ragged-right bibliography text.
\usepackage[protrusion=true]{microtype}

% Permit standard English hyphenation with two letters retained on either side.
% This adds genuine word-break opportunities without changing paragraph
% tolerance or warning thresholds.
\AddToHook{begindocument/end}{%
  \lefthyphenmin=2\relax
  \righthyphenmin=2\relax
}

\usepackage[margin=1in]{geometry}

% Retain amsart's zero natural paragraph separation while restoring the
% standard stretch needed to balance flush-bottom pages at paragraph boundaries.
\setlength{\parskip}{0pt plus 1pt}

% BibLaTeX lists replace paragraph glue with their own item spacing. Retain
% compact entries while allowing flush-bottom bibliography pages to balance
% when amsart keeps the final author-address block together.
\setlength{\bibitemsep}{0pt plus .25\baselineskip}
% amsart's address block starts with a no-break request, but a list-ending
% breakpoint can already precede it. Keep the bibliography's final entry
% attached to the author addresses when they are present.
\AtBeginBibliography{
  \ifx\addresses\@empty\else
    \@endparpenalty\@M
  \fi
}

% Shared author blocks may include bare ORCID identifiers for journal wrappers.
% Plain amsart drafts do not print ORCID metadata.
\ProvideDocumentCommand{\orcid}{m}{}

% Hyperref's generic pdfusetitle support sees the optional short title and only
% the last amsart author. Use the full title and accumulated author list instead.
\yjrefs_set_pdf_metadata_sources:nn { \@title } { \authors }
\ExplSyntaxOff
\makeatother
% Style wrapper. Change the class or add journal/arXiv-specific setup here;
% content.tex remains the single source for title, authors, abstract, body, and bibliography.
\title[Crosscap resurgence and Riemann--Hilbert problems]{Crosscap resurgence and
Riemann--Hilbert problems for orientifolded conifolds}

\author[Wu-yen Chuang]{Wu-yen Chuang\orcid{0000-0003-3230-3252}}
\address{Department of Mathematics and TIMS, National Taiwan University, Taipei, Taiwan}
\email{wychuang@gmail.com}

\author[Yi-Jing Tseng]{Yi-Jing Tseng\orcid{0009-0008-8349-6772}}
\address{Department of Mathematics, National Taiwan University, Taipei, Taiwan}
\email{yjtseng.math@gmail.com}


\begin{document}

\begin{abstract}
  We study the crosscap contribution to the large-$N$ $\SO/\Sp$ free energy of the orientifolded conifold by comparing its formal expansion with Faddeev's quantum dilogarithm. This determines its Borel transform, sectorial sums, Stokes jumps, and accumulation limits; in the convergent chamber, these limits combine with resolved-conifold summation to recover the full orientifold free energy. We also construct a convergent ray-finite uncoupled rank-five rational BPS structure supported on the finite crosscap rays, and show that its Stokes tails violate both unrestricted half-plane normalization and two-sided polynomial growth at the non-active accumulation rays. Restricting normalization to closed subsectors and dropping the growth condition at those rays yields a unique solution expressed in terms of normalized double-sine functions. Finally, finite differences of the scalar Stokes transition recover the two BPS character multipliers and determine the supported active rational invariants.
\end{abstract}

\maketitle

\section{Introduction}

Perturbative topological string theory organizes the free energy by the genus expansion
\begin{equation}\label{perturbative-free-energy}
  \mathcal{F}_{\text{pert}}(\lambda) = \sum_{g=0}^{\infty}
  \mathcal{F}_g\lambda^{2g-2},
\end{equation}
where $g$ is the genus and $\lambda$ the string coupling. Its coefficients are expected to grow factorially, and hence \eqref{perturbative-free-energy} generally has zero radius of convergence. Resurgence replaces this divergent expansion by directional Borel--Laplace sums whose Stokes jumps encode nonperturbative information \cite{ecalle-1981-8c402eea}. For the resolved conifold, this viewpoint relates Borel singularities to nonperturbative corrections and large-order growth \cite{pasquetti-2010-c984dedf}; it also governs the factorial growth of Gromov--Witten invariants \cite{cousosantamaria-2017-b93f05df} and gives explicit sectorial sums with a Donaldson--Thomas interpretation \cite{alim-2023-de7fb1f0}.

Donaldson--Thomas theory provides a nonlinear Riemann--Hilbert framework for these sectorial functions, with gamma-product solutions in special cases \cite{bridgeland-2019-6bdcc39c}. For the resolved conifold, Barnes double and triple sine functions solve the problem, and the associated $\tau$-function has the topological string partition function as its asymptotic expansion \cite{bridgeland-2020-7a4fdc8b}. The same relation between Borel sums and Donaldson--Thomas Riemann--Hilbert problems extends to elliptic Calabi--Yau threefolds \cite{bridgeland-2025-2eeb54b4}.

We study the orientifolded resolved conifold. Its large-$N$ $\SO/\Sp$ Chern--Simons descriptions on $S^3$ are dual to topological strings on the corresponding orientifold \cite{sinha-2000-c6cf9be1}, and the free energy splits into the resolved-conifold Gromov--Witten potential and an odd-projected crosscap contribution. Unoriented resurgence has also been studied through the real topological string and its holomorphic-anomaly equations \cite{krefl-2009-d1c56d3a,walcher-2009-238c25e9}; this approach yields all-orders trans-series and identifies disk invariants as Stokes constants \cite{marino-2026-ea927d3b}. Our analysis instead starts from the exact crosscap product and determines its sectorial sums, accumulation directions, and torus-valued Riemann--Hilbert problem.

We first analyze the unprojected crosscap product. After a change of variables, Faddeev's quantum dilogarithm gives the positive-real Borel sum of its formal expansion \cite{garoufalidis-2021-a73cf607}; the common rank-one shift equation by itself leaves a periodic ambiguity. A vertical Laplace integral fixes the $q$-Pochhammer branch once the dilogarithm principal term is restored. This comparison yields the Borel transforms, sectorial sums, finite Stokes jumps, and accumulation limits of both the unprojected and odd-projected series. Combining the resulting formulas with the resolved-conifold summation of \cite{alim-2023-de7fb1f0} gives the full orientifold free energy.

We next construct a convergent ray-finite uncoupled rank-five rational BPS structure whose finite active rays have two non-active accumulation directions. \zcref{crosscap-full-rh-obstruction,crosscap-full-rh-growth-obstruction} show that its Stokes tails violate both normalization and polynomial growth at these directions. Consequently, Bridgeland's full Riemann--Hilbert problem \cite{bridgeland-2020-7a4fdc8b} has no solution with the chosen constant term. \zcref{crosscap-rh-subsector-normalization} replaces half-plane normalization there by closed-subsector normalization and omits growth only on the two limiting rays. \zcref{crosscap-rh-solution,crosscap-accumulation-character-asymptotics} give the unique solution in ratios of normalized double-sine functions and its exact central-ray asymptotic.

Finite differences of the scalar Stokes transition reproduce the two character multipliers for the mixed pairing, in parallel with \cite{alim-2023-de7fb1f0}. \zcref{crosscap-bps-reconstruction} shows that, within the prescribed two-coordinate ansatz, they force a rank-five lattice, recover the active span and mixed pairing up to a nonzero integral scale, and determine the supported rational invariants. Modified triple-sine functions give a local differential $\tau$-function, and \zcref{crosscap-tau-free-energy-relation} relates it asymptotically to the crosscap free energy. Finally, \zcref{resurgent-crosscap-local-coordinate-data,resurgent-crosscap-half-turn-obstruction} exhibit raw coordinates that satisfy every finite-ray equation and tend to $1$ at the origin but fail both half-turn identities. Thus local Stokes data and small-parameter normalization do not by themselves imply global sectorial compatibility.

The additive free-energy decomposition factorizes the scalar $\SO/\Sp$ transitions, but a torus-valued factorization also needs a mixed ordinary--crosscap pairing. Ordinary resolved-conifold point classes make the accumulation rays active, so the isolated crosscap obstruction does not decide the full orientifold problem. The orientihole charge picture \cite{denef-2010-ba429955} and Hall-module wall crossing for self-dual quiver representations \cite{young-2015-9be91531,bu-2025-1a4b688b} point instead toward a category with duality that realizes the rank-five data.

The paper is organized as follows. \zcref{sosp-primitive-qdilog} relates the $\SO/\Sp$ crosscap free energy to Faddeev's quantum dilogarithm. \zcref{resurgence-analysis} determines the Borel transforms, Stokes data, and accumulation limits. \zcref{crosscap-bps-rh} constructs the BPS structure, isolates the accumulation-ray obstructions, and solves the modified Riemann--Hilbert problem. \zcref{crosscap-tau-resurgence} constructs the differential $\tau$-function and compares its character multipliers with finite differences of the resurgent partition function.

\section{\texorpdfstring{$\SO/\Sp$}{SO/Sp} Chern--Simons theory and Faddeev's
quantum dilogarithm}\label{sosp-primitive-qdilog}

We first isolate the primitive crosscap factor in the large-$N$ $\SO/\Sp$ free energy. We then compare its shift equation and formal expansion with Faddeev's quantum dilogarithm \cite{faddeev-1995-65fc7c50}.

\subsection{Large-\texorpdfstring{$N$}{N} free energy and the crosscap term}

The large-$N$ resolved-conifold dual of Chern--Simons theory on $S^3$ admits an orientifold quotient describing the corresponding $\SO/\Sp$ theories \cite{gopakumar-2001-d6b8f4c3,sinha-2000-c6cf9be1}. On the resolved side, the exceptional $\mathbb P^1$ is replaced by $\mathbb{RP}^2$, while the crosscap sign distinguishes the two cases \cite[Section~3]{sinha-2000-c6cf9be1}. Set $\epsilon=+1$ for $\SO$ and $\epsilon=-1$ for $\Sp$.

Let $T$ be the Kähler parameter in the Chern--Simons normalization. After identifying the gauge-theory and closed-string variables, the $\SO/\Sp$ free energy takes the equivalent form \cite[eq.~(3.2)]{bouchard-2004-e19a9bf3}
\begin{equation}\label{sosp-free-energy-cs}
  \mathcal F_{\SO/\Sp}(g_s,T)
  =\frac{1}{2}\sum_{d=1}^{\infty}
  \frac{1}{d}\frac{e^{-dT}}
  {(q^{d/2}-q^{-d/2})^2}
  -\epsilon\sum_{d\textup{ odd}}
  \frac{1}{d}\frac{e^{-dT/2}}{q^{d/2}-q^{-d/2}}.
\end{equation}
Here $q=e^{g_s}$, where $g_s=2\pi i/(k+y)$ and $y=N-2$ for $\SO(N)$ or $y=N+1$ for $\Sp(N)$. The corresponding 't Hooft parameter is $T=(N-\epsilon)g_s$ \cite[eqs.~(4.38)--(4.39)]{sinha-2000-c6cf9be1}.

For the resurgent analysis, introduce the additive coordinate $t$ and coupling $\lambda$ by
\[
  e^{-T/2}=e^{\pi i t},
  \qquad g_s=i\lambda,
  \qquad Q=e^{2\pi i t}.
\]
Thus $t$ is an additive lift of $Q^{1/2}$, and translation by $2\mathbb Z$ only relabels the multicover and Stokes indices. Since $2\sinh(dg_s/2)=2i\sin(d\lambda/2)$, \eqref{sosp-free-energy-cs} becomes
\begin{equation}\label{sosp-free-energy-large-n}
  \mathcal F_{\SO/\Sp}(\lambda,t)
  =-\frac{1}{2}\mathcal F_{\mathrm{GV}}(\lambda,t)
  +\epsilon i\sum_{k\textup{ odd}}
  \frac{1}{k}\frac{e^{\pi i kt}}{2\sin(\lambda k/2)},
\end{equation}
where
\[
  \mathcal F_{\mathrm{GV}}(\lambda,t)
  \coloneqq\sum_{k=1}^{\infty}
  \frac{1}{k}\frac{e^{2\pi i kt}}{(2\sin(\lambda k/2))^2}.
\]
For a primitive unoriented sector, the fixed-locus condition in the M-theory lift retains the modes of weight $q^{n+1/2}x$, with $n\geq0$ \cite[eq.~(3.12)]{aganagic-2012-386ffcfc}. Here $x$ combines the charge fugacity with the fixed intrinsic-spin weight, and the resulting local factor is
\begin{equation}\label{fock-psi}
  \Psi(x)=\prod_{n\geq0}\left(1-q^{n+1/2}x\right).
\end{equation}
Reindexing yields the rank-one shift identity
\begin{equation}\label{psi-shift}
  \Psi(q^{-1/2}x)=(1-x)\Psi(q^{1/2}x).
\end{equation}
Assume first that $\rmIm(t)>0$ and $\rmIm(\lambda)>0$, and set $x=Q^{1/2}=e^{\pi it}$ and $q=e^{i\lambda}$. Then \eqref{fock-psi} selects the $q$-Pochhammer branch
\begin{equation}\label{primitive-qpochhammer-branch}
  F^+(\lambda,t)
  \coloneqq
  i\log\Psi(Q^{1/2})
  =
  i\log\mleft(Q^{1/2}e^{i\lambda/2};e^{i\lambda}\mright)_\infty.
\end{equation}
The product converges in this chamber, and its logarithm expands as
\[
  F^+(\lambda,t)
  =
  -i\sum_{k=1}^{\infty}
  \frac{Q^{k/2}e^{ik\lambda/2}}{k(1-e^{ik\lambda})}.
\]
The identity $(2\sin(k\lambda/2))^{-1}=-ie^{ik\lambda/2}/(1-e^{ik\lambda})$ therefore gives
\[
  F^+(\lambda,t)=\sum_{k=1}^{\infty}
  \frac{1}{k}\frac{e^{\pi ikt}}{2\sin(\lambda k/2)}.
\]
Reflection across the origin in the $\lambda$-plane gives the second branch
\[
  F^-(\lambda,t)\coloneqq
  -F^+(-\lambda,t)
  =
  -i\log\mleft(Q^{1/2}e^{-i\lambda/2};e^{-i\lambda}\mright)_\infty .
\]
This branch is analytic for $\rmIm(\lambda)<0$ and has the same sectorial expansion. Hence, on the respective half-planes,
\begin{equation}\label{primitive-sectorial-series}
  F^\pm(\lambda,t)=\sum_{k=1}^{\infty}
  \frac{1}{k}\frac{e^{\pi ikt}}{2\sin(\lambda k/2)}.
\end{equation}
Projecting \eqref{primitive-sectorial-series} to odd positive integers gives
\begin{equation}\label{primitive-odd-projection}
  F^\pm(\lambda,t)-\frac12F^\pm(2\lambda,2t)
  =
  \sum_{k\textup{ odd}}
  \frac{1}{k}\frac{e^{\pi ikt}}{2\sin(\lambda k/2)}.
\end{equation}
Multiplication by $\epsilon i$ recovers the crosscap term in \eqref{sosp-free-energy-large-n}.

Expanding each summand of \eqref{primitive-sectorial-series} at $\lambda=0$ gives the formal Laurent series
\begin{equation}\label{primitive-formal-series}
  \widetilde F(\lambda,t)\coloneqq
  \frac{1}{\lambda}\Li_2(Q^{1/2})+
  \sum_{g=1}^{\infty}\lambda^{2g-1}
  \frac{(-1)^{g-1}(1-2^{1-2g})B_{2g}}{(2g)!}
  \Li_{2-2g}(Q^{1/2}).
\end{equation}
The expansion is only coefficientwise: the $k$-th summand has radius $\abs{\lambda}<2\pi/k$, which tends to zero with $k$. Nevertheless, the correction term in \eqref{primitive-formal-series} is Gevrey-$1$, and its Borel--Laplace sums recover $F^\pm$, as shown in \zcref{primitive-borel-transform} and \eqref{primitive-sectorial-sum}.

\subsection{Comparison with the quantum
dilogarithm}

We adapt the Borel-summation formulas for Faddeev's quantum dilogarithm from \cite{garoufalidis-2021-a73cf607} to the primitive crosscap factor. This requires translating the variables and Borel--Laplace conventions, together with two normalization corrections recorded below.

We use the Gevrey-$1$ and resurgent conventions of \cite[Sections~4 and~18]{sauzin-2014-eb425ca6}, with formal Borel transform $\mathcal B[w^{n+1}](\zeta)=\zeta^n/n!$. Terms of degree less than $1$ remain outside the transform.
\begin{definition}
  \label{sectorial-sum-laurent-series}
  Let
  \[
    \widetilde\varphi(w)=\widetilde\varphi_{<1}(w)
    +\widetilde\varphi_{\geq1}(w),\qquad
    \widetilde\varphi_{<1}(w)\in\mathbb C[w^{-1}],\quad
    \widetilde\varphi_{\geq1}(w)\in w\mathbb C\cdbr{w} .
  \]
  If $\widetilde\varphi_{\geq1}$ is Borel summable in a non-Stokes direction $\rho$, define
  \[
    \mathcal S_\rho\widetilde\varphi(w)
    \coloneqq\widetilde\varphi_{<1}(w)
    +\int_\rho
    e^{-\zeta/w}\mathcal B[\widetilde\varphi_{\geq1}](\zeta)d\zeta .
  \]
\end{definition}

Set $x=-e^z$ and $q=e^{2\pi i\tau}$. Then \eqref{psi-shift} becomes
$f_\tau(z-i\pi\tau)=(1+e^z)f_\tau(z+i\pi\tau)$. Although the boundary condition $\lim_{\rmRe(z)\to-\infty}f(z)=1$ corresponds to $\Psi(0)=1$, it does not remove the freedom to multiply by a $2\pi i\tau$-periodic function. Borel summation selects a branch $f_\tau$, whose sum along $i\mathbb R_{>0}$ recovers \eqref{primitive-qpochhammer-branch}.

The required shift, asymptotic, and Borel formulas are \cite[eqs.~(1.1)--(1.3), (1.19)--(1.23), Theorems~1.1 and~1.3, and Lemma~2.1]{garoufalidis-2021-a73cf607}. For $\abs{\rmIm(z)}<\pi$, set
\[
  \widehat\phi_\tau(z)=\sum_{n=1}^{\infty}
  (2\pi i)^{2n-1}\frac{B_{2n}(1/2)}{(2n)!}
  \partial_z^{2n}\Li_2(-e^z)\tau^{2n-1}.
\]
The selected solution satisfies
\[
  \log f_\tau(z)\sim
  \frac{1}{2\pi i\tau}\Li_2(-e^z)+\widehat\phi_\tau(z),
  \qquad \tau\to0,
\]
and, for $\abs{\xi}<\pi-\abs{\rmIm(z)}$,
\begin{equation}\label{gk-borel-transform}
  \mathcal B[\widehat\phi_{(-)}(z)](\xi)
  =
  \frac{1}{2\pi i}\sum_{n=1}^{\infty}
  \frac{(-1)^n}{n^2}
  \left(
    \frac{1}{1+e^{\xi/n-z}}
    +
    \frac{1}{1+e^{-\xi/n-z}}
  \right).
\end{equation}
The series on the right is expanded at $\xi=0$ and extends meromorphically, with possible poles at $n\xi_m(z)$, where
\[
  \xi_m(z)=z+(2m+1)\pi i,
  \qquad
  n\in\mathbb Z\setminus\{0\},
  \quad m\in\mathbb Z.
\]
At a fixed point $\xi_*$, all coincident labels contribute to the principal part:
\begin{equation}\label{gk-residue}
  \mathcal B[\widehat\phi_{(-)}(z)](\xi)=
  \frac{1}{2\pi i}
  \sum_{n\xi_m(z)=\xi_*}\frac{(-1)^{n+1}}{n}
  \frac{1}{\xi-n\xi_m(z)}
  +\textup{holomorphic}.
\end{equation}
The sum can vanish, making the apparent pole removable. For instance, when $z=0$, the two terms in parentheses in \eqref{gk-borel-transform} add to $1$. Direct expansion of that formula gives the sign in \eqref{gk-residue}, opposite to the sign printed in \cite[eq.~(1.10)]{garoufalidis-2021-a73cf607}.

If a ray $\rho\subset\mathbb C^\ast$ avoids every line $\mathbb R\xi_m(z)$, define
\begin{equation}\label{gk-sectorial-sum}
  \mathcal S_\rho\widehat\phi_\tau(z)
  =
  \int_\rho e^{-\xi/\tau}
  \mathcal B[\widehat\phi_{(-)}(z)](\xi)d\xi.
\end{equation}
The sum is constant under deformation of $\rho$ within a component of the complement and is holomorphic whenever $\rmRe(\xi/\tau)>0$ along the ray.

Suppose that $\rmIm(\tau)>0$, $\rmRe(z)<0$, and $\abs{\rmIm(z)}<\pi$, and put $q=e^{2\pi i\tau}$. Then
\begin{equation}
  \log f_\tau(z)=
  \log(-q^{1/2}e^z;q)_\infty
  -
  \log(-e^{-\pi i/\tau}e^{z/\tau};
  e^{-2\pi i/\tau})_\infty.
\end{equation}
The positive-imaginary sectorial sum is therefore
\begin{equation}\label{gk-vertical-limit}
  \mathcal S_{i\mathbb R_{>0}}\widehat\phi_\tau(z)
  =
  \log(-q^{1/2}e^z;q)_\infty
  -\frac{1}{2\pi i\tau}\Li_2(-e^z).
\end{equation}
The dilogarithm term is forced by the definition of $\widehat\phi$ and is absent from \cite[eq.~(1.24)]{garoufalidis-2021-a73cf607}. For a direct verification, put $a=-e^z$, so $\abs a<1$, and expand both denominators of \eqref{gk-borel-transform} on $i\mathbb R_{>0}$. The resulting double series is bounded by a constant multiple of $\sum_{n=1}^\infty n^{-2}\sum_{r=1}^\infty\abs a^r$, and the Laplace kernel decays because $\rmRe(i/\tau)>0$. Termwise integration, followed by the partial-fraction expansion of $\csc(\pi u)$, yields
\begin{align*}
  \mathcal S_{i\mathbb R_{>0}}\widehat\phi_\tau(z)
  &=-\frac{\tau}{\pi i}\sum_{r=1}^\infty a^r
  \sum_{n=1}^\infty\frac{(-1)^n}{n^2-r^2\tau^2}\\
  &=\sum_{r=1}^\infty a^r\mleft(
  \frac{1}{2ir\sin(\pi r\tau)}-\frac{1}{2\pi ir^2\tau}\mright)\\
  &=\log(aq^{1/2};q)_\infty-\frac{\Li_2(a)}{2\pi i\tau}.
\end{align*}
For $\tau>0$, $b^2=\tau$, and $\abs{\rmIm(z)}<\pi$, the positive-real sectorial sum gives Faddeev's quantum dilogarithm:
\begin{equation}\label{gk-positive-faddeev}
  \log\Phi_b\mleft(\frac{z}{2\pi b}\mright)
  =
  \frac{1}{2\pi i\tau}\Li_2(-e^z)
  +
  \mathcal S_{\mathbb R_{>0}}\widehat\phi_\tau(z).
\end{equation}

The required change of variables is
\[
  \tau=\frac{\lambda}{2\pi},\qquad
  z(t)=\pi i(t-1),\qquad
  q=e^{2\pi i\tau}=e^{i\lambda}.
\]
Since $-e^{z(t)}=e^{\pi it}=Q^{1/2}$, logarithmic representatives differing from $z(t)$ by $2\pi i$ give the same formal coefficients. The displayed representative lies in the strip $\abs{\rmIm(z)}<\pi$ for $0<\rmRe(t)<2$; elsewhere, the resulting identities continue analytically in $t$.

After setting $\lambda=2\pi\tau$, denote the correction term in \eqref{primitive-formal-series} by $\Phi(\tau,t)$:
\[
  \widetilde F(\lambda,t)=\frac{1}{\lambda}\Li_2(Q^{1/2})
  +\Phi\mleft(\frac{\lambda}{2\pi},t\mright).
\]
Using $\partial_z^m\Li_2(-e^z)=\Li_{2-m}(-e^z)$ and $B_{2n}(1/2)=-(1-2^{1-2n})B_{2n}$ gives
\begin{equation}\label{primitive-phihat-relation}
  \Phi(\tau,t)=i\widehat\phi_\tau(z(t)).
\end{equation}
Consequently, the asymptotic expansion of $\log f_\tau(z)$ gives $\widetilde F(2\pi\tau,t)\sim i\log f_\tau(z(t))$.

\section{Borel sums, Stokes jumps, and accumulation limits}\label{resurgence-analysis}

We first determine the crosscap Borel transform, Stokes jumps, and accumulation limits. We then combine them with resolved-conifold summation to recover the full orientifold free energy.

\subsection{Resurgence of the crosscap part}

For a ray $\rho\subset\mathbb C^\ast$, define its centered Borel--Laplace half-plane by
\[
  \mathbb H_\rho
  =
  \setc*{w\in\mathbb{C}^\ast}{\rmRe(\zeta/w)>0},
  \qquad \zeta\in\rho .
\]
The definition is independent of the choice of $\zeta\in\rho$.

\begin{definition}
  Let $\varphi(z)\in z\mathbb C\cdbr{z}$ be resurgent and let $\rho$ be a Stokes ray. For sufficiently small $\varepsilon>0$, rotate $\rho$ counterclockwise and clockwise through $\varepsilon$ to obtain $\rho_+$ and $\rho_-$. Whenever both lateral sums exist, define
  \[
    \disc_\rho \varphi(w)\coloneqq
    \mathcal S_{\rho_+}\varphi(w)-\mathcal S_{\rho_-}\varphi(w).
  \]
\end{definition}

Throughout this section, $Q^{1/2}=e^{\pi it}$ for the chosen additive lift with $\rmIm(t)>0$. Sectorial sums are understood in the sense of \zcref{sectorial-sum-laurent-series}. Each displayed contour or product fixes the initial branch of every logarithm and fractional power, which is then continued within the stated non-Stokes sector.

For $z\in\mathbb C$ and $\omega_1,\omega_2\in\mathbb C^\ast$ with $\omega_1/\omega_2\notin\mathbb R_{<0}$, set \cite[Section~4.1]{bridgeland-2020-7a4fdc8b}
\[
  F_2\vargs*{z}{\omega_1, \omega_2}\coloneqq\exp \mleft(-\frac{\pi
  i}{2} B_{2,2}\vargs*{z}{\omega_1, \omega_2}\mright)\sin_2\vargs*{z}{\omega_1, \omega_2},
\]
where $\sin_2\vargs{z}{\omega_1,\omega_2}$ is the double sine function and
\[
  B_{2,2}\vargs*{z}{\omega_1, \omega_2}= \frac{z^2}{\omega_1
  \omega_2}-\left(\frac{1}{\omega_1}+\frac{1}{\omega_2}\right)
  z+\frac{1}{6}\left(\frac{\omega_2}{\omega_1}+\frac{\omega_1}{\omega_2}\right)+\frac{1}{2}
\]
is the multiple Bernoulli polynomial.

\begin{proposition}[\cite[Proposition~4.1(i)]{bridgeland-2020-7a4fdc8b}]\label{f2-bridgeland-properties}
  Suppose that $\omega_1/\omega_2\notin\mathbb R_{<0}$. Then $F_2\vargs{z}{\omega_1,\omega_2}$ is single-valued and meromorphic in $z\in\mathbb C$ and $\omega_1,\omega_2\in\mathbb C^\ast$. Its zero and pole sets are, respectively,
  \[
    \setc*{-m\omega_1-n\omega_2}{m,n\in\mathbb Z_{\geq0}},
    \qquad
    \setc*{(m+1)\omega_1+(n+1)\omega_2}{m,n\in\mathbb Z_{\geq0}}.
  \]
  It is holomorphic and nonvanishing away from these sets.
\end{proposition}

If $\rmRe(\omega_i)>0$ for $i=1,2$ and $0<\rmRe(z)<\rmRe(\omega_1+\omega_2)$, the integral representation \cite[Proposition~4.1(v)]{bridgeland-2020-7a4fdc8b} is
\begin{equation}\label{f2-integral-representation}
  F_2\vargs*{z}{\omega_1,\omega_2}
  =
  \exp\mleft(
    \int_{\mathbb R+i0}
    \frac{e^{zs}}{\left(e^{\omega_1s}-1\right)\left(e^{\omega_2s}-1\right)}
    \frac{ds}{s}
  \mright).
\end{equation}
Here $\mathbb R+i0$ is the real axis with a small upper-half-plane detour around the origin.

\begin{lemma}
  Let $b>0$, set $\tau=b^2$, and suppose that $z\in\mathbb C$ satisfies $\abs{\rmIm(z)}<\pi$. Define \cite[eq.~(1.5)]{garoufalidis-2021-a73cf607}
  \[
    \Phi_b(u)=
    \exp\mleft(
      \int_{\mathbb R+i0}
      \frac{e^{-2ixu}}{4\sinh(bx)\sinh(b^{-1}x)}
      \frac{dx}{x}
    \mright).
  \]
  Then
  \begin{equation}\label{f2-phi-bridge-formula}
    \Phi_b\mleft(\frac{z}{2\pi b}\mright)
    =
    F_2\vargs*{\frac{1+\tau}{2}+\frac{z}{2\pi i}}{\tau,1}.
  \end{equation}
  In particular, for $z(t)=\pi i(t-1)$,
  \begin{equation}\label{f2-phi-bridge-specialization}
    \Phi_b\mleft(\frac{z(t)}{2\pi b}\mright)
    =
    F_2\vargs*{\frac{t+\tau}{2}}{\tau,1}.
  \end{equation}
\end{lemma}

\begin{proof}
  Put $a=(1+\tau)/2+z/(2\pi i)$. The hypotheses imply $0<\rmRe(a)<\tau+1$, so the branch fixed by \eqref{f2-integral-representation} satisfies
  \[
    \log F_2\vargs*{a}{\tau,1}
    =
    \int_{\mathbb R+i0}
    \frac{e^{as}}{(e^{\tau s}-1)(e^s-1)}
    \frac{ds}{s}.
  \]
  Substituting $s=2x/b$ gives
  \[
    \frac{e^{as}}{(e^{\tau s}-1)(e^s-1)}
    =
    \frac{e^{-ixz/(\pi b)}}{4\sinh(bx)\sinh(b^{-1}x)},
  \]
  together with $ds/s=dx/x$. The resulting integral is the logarithm of $\Phi_b(z/(2\pi b))$, which proves \eqref{f2-phi-bridge-formula}; substituting $z=z(t)$ gives \eqref{f2-phi-bridge-specialization}.
\end{proof}

For the formal series in \eqref{primitive-formal-series}, write
\[
  \widetilde F(\lambda,t)
  =\frac{1}{\lambda}\Li_2(Q^{1/2})+\Phi(\check\lambda,t),
  \qquad
  \check\lambda=\frac{\lambda}{2\pi},
  \qquad
  Q=e^{2\pi i t},
\]
\begin{proposition}\label{primitive-borel-transform}
  Suppose that $t\in\mathbb C^\ast$ and $\abs{\rmRe(t)}<1$, and let $G(\xi,t)\coloneqq\mathcal B[\Phi(-,t)](\xi)$. This series converges for $\abs{\xi}<\pi\abs t$ and extends meromorphically in $\xi$ as
  \[
    G(\xi,t)=\frac{1}{2\pi}\sum_{m=1}^{\infty}\frac{(-1)^m}{m^2}\left(\frac{1}{1-e^{-\pi i t+\xi/m}}+\frac{1}{1-e^{-\pi i t-\xi/m}}\right).
  \]
  Its possible singularities are $\eta=\pi i(t+2k)m$, where $k\in\mathbb Z$ and $m\in\mathbb Z\setminus\{0\}$. At such a point,
  \[
    G(\xi,t)
    =\frac{1}{2\pi}
    \mleft(
      \sum_{\pi i(t+2k)m=\eta}\frac{(-1)^{m+1}}{m}
    \mright)
    \frac{1}{\xi-\eta}
    +\textup{holomorphic}.
  \]
  The sum is finite; hence $\eta$ is a simple pole precisely when its coefficient is nonzero, and is otherwise removable.
\end{proposition}
\begin{proof}
  Put $\tau=\check\lambda$ and $z=z(t)=\pi i(t-1)$. On the strip $\abs{\rmIm(z)}<\pi$, \eqref{primitive-phihat-relation} gives
  \[
    \Phi(\tau,t)=i\widehat\phi_\tau(z),
    \qquad
    G(\xi,t)=i\mathcal B[\widehat\phi_{(-)}(z)](\xi).
  \]
  Assume first that $0<\rmRe(t)<1$, so $z$ lies in the required strip. Since $-e^z=Q^{1/2}$ and $e^{-z}=-e^{-\pi it}$, \eqref{gk-borel-transform} gives the stated series. Moreover, $\xi_r(z)=\pi i(t+2r)$, so its denominators vanish exactly at the displayed candidates.

  By \eqref{gk-residue}, the label $(k,m)$ contributes $(-1)^{m+1}/(2\pi m)$ to the residue of $G$. For fixed $t$, the inequality $\abs{t+2k}>\abs t$ for $k\neq0$ makes the candidate set locally finite, and summing coincident contributions gives the principal part. The nearest candidates $\pm\pi it$ carry only the labels $(0,\pm1)$ and have nonzero residues; all others have larger modulus. Thus the convergence radius is exactly $\pi\abs t$. Analytic continuation in $t$ away from the candidate locus proves the remaining cases with $\abs{\rmRe(t)}<1$.
\end{proof}

For $j\in\mathbb Z$, set $\ell_j(t)=\mathbb R_{>0}\cdot\pi i(t+j)$ and let $\Sigma_j(t)$ be the convex open sector bounded by $\ell_{j-1}(t)$ and $\ell_j(t)$. Also set $\ell_\infty=\mathbb R_{>0}\cdot i$, and suppress $t$ when it is fixed. Since $\xi_r(z)=\pi i(t+2r)$, the excluded lines in \eqref{gk-sectorial-sum} pass through the even rays $\pm\ell_{2r}$ and accumulate at $\pm\ell_\infty$. For any ray $\rho$ avoiding these directions and any $\lambda\in\mathbb H_\rho$, define
\begin{equation}\label{primitive-sectorial-sum}
  \mathcal S_\rho\widetilde F(\lambda,t)
  \coloneqq\frac{1}{\lambda}\Li_2(Q^{1/2})+\int_\rho e^{-\xi/\check\lambda}G(\xi,t)d\xi.
\end{equation}

\begin{proposition}\label{primitive-positive-real-sum}
  Suppose that $0<\rmRe(t)<2$, $\rmIm(t)>0$, $\rmRe(\check\lambda)>0$, and $\rmRe(t)<\rmRe(\check\lambda+2)$. Then
  \[
    \mathcal S_{\mathbb R_{>0}}\widetilde F(\lambda,t)
    =i\log F_2\vargs*{\frac{t}{2}+\frac{\check\lambda}{2}}{\check\lambda,1}.
  \]
\end{proposition}
\begin{proof}
  Begin with $\tau>0$. The hypotheses imply
  \[
    \rmRe(\tau)>0,
    \qquad
    0<\rmRe\mleft(\frac{t+\tau}{2}\mright)<\rmRe(1+\tau),
  \]
  and therefore \eqref{f2-integral-representation} applies to $F_2\vargs{(t+\tau)/2}{\tau,1}$. \zcref{f2-phi-bridge-specialization,gk-positive-faddeev} prove the claim on the positive real $\tau$-interval.

  Now fix any target point in the stated domain. Join it to the positive-real $\tau$-interval by a path along which $(t+\tau)/2$ avoids the zero--pole divisor of $F_2$; the locally finite intersections can be bypassed by arbitrarily small detours. The endpoint is itself off the divisor: an equality $(t+\tau)/2=a\tau+b$ would give $t=(2a-1)\tau+2b$. If $a,b\leq0$, then $\rmRe(t)<0$, whereas if $a,b>0$, then $\rmRe(t)\geq\rmRe(\tau)+2$. Both contradict the hypotheses. Analytic continuation along the chosen path proves the identity.
\end{proof}

\begin{proposition}\label{primitive-stokes-jump}
  Suppose that $\rmIm(t)>0$. Then, for $\varepsilon\in\{+1,-1\}$,
  \[
    \disc_{\varepsilon\ell_{2k}}\widetilde F(\lambda,t)
    =-\varepsilon i\log\mleft(1+e^{-\varepsilon\pi i(t+2k)/\check\lambda}\mright).
  \]
\end{proposition}
\begin{proof}
  Summing the residues in \zcref{primitive-borel-transform} gives
  \[
    \disc_{\ell_{2k}}\widetilde F
    =-2\pi i\sum_{m=1}^{\infty}\frac{(-1)^{m+1}}{2\pi m}
    e^{-\pi im(t+2k)/\tau}
    =-i\log\mleft(1+e^{-\pi i(t+2k)/\tau}\mright).
  \]
  The opposite ray similarly gives $i\log(1+e^{\pi i(t+2k)/\tau})$. Since $\rmIm(t)>0$, increasing $j$ orders the rays $\ell_j$ clockwise. The positive and negative lateral directions at $\ell_{2k}$ therefore lie in $\Sigma_{2k}$ and $\Sigma_{2k+1}$.
\end{proof}

\begin{proposition}\label{primitive-accumulation-limit}
  Suppose that $0<\rmRe(t)<2$, $\rmIm(t)>0$, and $\rmIm(\check\lambda)>0$. For each $k$, let $\rho_k$ lie in the convex open sector between the consecutive primitive Stokes rays $\ell_{2k-2}$ and $\ell_{2k}$. For all sufficiently large $k$, the sectorial sums are defined and
  \[
    \lim_{k\to\infty}\mathcal S_{\rho_k}\widetilde F(\lambda,t)
    =i\log\mleft(e^{\pi i(t+\check\lambda)};e^{2\pi i\check\lambda}\mright)_\infty
    =\sum_{r=1}^{\infty}\frac{e^{\pi irt}}{2r\sin(\pi r\check\lambda)}.
  \]
  The product logarithm is its convergent power-series branch, and the sine series converges normally on compact subsets of the stated domain.
\end{proposition}
\begin{proof}
  Put $\tau=\check\lambda$. Then $z=\pi i(t-1)$ has negative real part, and \zcref{gk-vertical-limit,primitive-phihat-relation} give
  \[
    \frac{\Li_2(Q^{1/2})}{2\pi\tau}
    +i\mathcal S_{i\mathbb R_{>0}}\widehat\phi_\tau(z)
    =i\log\mleft(e^{\pi i(t+\tau)};e^{2\pi i\tau}\mright)_\infty.
  \]
  Expanding the product logarithm and using $1-e^{2\pi ir\tau}=-2ie^{\pi ir\tau}\sin(\pi r\tau)$ gives the sine series, with normal convergence on the stated domain. The principal term in \eqref{primitive-sectorial-sum} cancels the subtraction in \eqref{gk-vertical-limit}. On each compact subset of $\rmIm(\tau)>0$, all sufficiently remote jumps between $\rho_k$ and $i\mathbb R_{>0}$ are bounded by $Ce^{-cj}$. Their normally convergent sum tends to zero with $k$, so the sectorial sums converge locally uniformly to the vertical sum, independently of the rays $\rho_k$.
\end{proof}
Applying the odd projection \eqref{primitive-odd-projection} to the formal expansion gives
\begin{align*}
  \widetilde F_\mathrm{cc}(\lambda,t)
  &\coloneqq i\left(\widetilde F(\lambda,t)-\frac12\widetilde F(2\lambda,2t)\right)\\*
  &=\frac{i}{\lambda}\Li_2(Q^{1/2})-\frac{i}{4\lambda}\Li_2(Q)+i\Phi(\check\lambda,t)-\frac{i}{2}\Phi(2\check\lambda,2t),
\end{align*}
where $\check\lambda=\lambda/(2\pi)$ and $Q=e^{2\pi i t}$.

\begin{theorem}
  Suppose $t\in\mathbb C^\ast$ satisfies $\abs{\rmRe(t)}<1/2$. Write
  \begin{align*}
    G_\mathrm{cc}(\xi,t)
    &\coloneqq i\mathcal B[\Phi(-,t)](\xi)-\frac{i}{2}\mathcal B[\Phi(2(-),2t)](\xi)\\*
    &=iG(\xi,t)-iG(2\xi,2t).
  \end{align*}
  The series $G_\mathrm{cc}(\xi,t)$ converges for $\abs{\xi}<\pi\abs t$ and extends meromorphically in $\xi$. Its simple poles lie at $\xi=\pi i(t+k)m$, where $k\in\mathbb Z$ and $m\in\mathbb Z\setminus\{0\}$, and have residues
  \[
    \Res\mleft(G_\mathrm{cc}(\xi,t),\pi i(t+k)m\mright)
    =-\frac{i(-1)^{k+m}}{4\pi m}.
  \]
\end{theorem}
\begin{proof}
  Taking the Borel transform coefficientwise gives
  \begin{align*}
    \mathcal B\mleft[i\Phi(-,t)-\frac{i}{2}\Phi(2(-),2t)\mright](\xi)
    &=iG(\xi,t)-iG(2\xi,2t)\\*
    &=G_\mathrm{cc}(\xi,t).
  \end{align*}
  The poles of $G(\xi,t)$ occur at $\pi i(t+2a)m$, whereas those of $G(2\xi,2t)$ occur at $\pi i(t+a)m$, with $a\in\mathbb Z$ and $m\neq0$. An equality $(t+k)m=(t+k')m'$ implies $(m-m')t=k'm'-km$; since $\rmIm(t)>0$, it follows that $m=m'$ and $k=k'$. The two terms can therefore meet only at poles with even index. At such a pole,
  \begin{align*}
    \Res\mleft(G(\xi,t),\pi i(t+2a)m\mright)
    &=-\frac{1}{2\pi}\frac{(-1)^m}{m},\\*
    \Res\mleft(G(2\xi,2t),\pi i(t+2a)m\mright)
    &=-\frac{1}{4\pi}\frac{(-1)^m}{m}.
  \end{align*}
  The combined residue is nonzero. At an odd-indexed pole, only the second term contributes. Hence, for every $k\in\mathbb Z$ and $m\neq0$,
  \[
    \Res\mleft(G_\mathrm{cc}(\xi,t),\pi i(t+k)m\mright)
    =-\frac{i}{4\pi}\frac{(-1)^{m+k}}{m}.
  \]
\end{proof}

The odd projection preserves Borel summability away from the Stokes directions. Thus, if $\rho$ avoids $\{\pm\ell_k\}_{k\in\mathbb Z}\cup\{\pm\ell_\infty\}$ and $\lambda\in\mathbb H_\rho$, then
\begin{equation}\label{crosscap-sectorial-sum}
  \begin{aligned}
    \mathcal S_\rho\widetilde F_\mathrm{cc}(\lambda,t)
    &\coloneqq\frac{i}{\lambda}\Li_2(Q^{1/2})-\frac{i}{4\lambda}\Li_2(Q)+\int_\rho e^{-\xi/\check\lambda}G_\mathrm{cc}(\xi,t)d\xi\\
    &=i\mathcal S_\rho\widetilde F(\lambda,t)-\frac{i}{2}\mathcal S_\rho\widetilde F(2\lambda,2t).
  \end{aligned}
\end{equation}
The second equality follows by setting $\eta=2\xi$, since dilation by $2$ preserves the oriented ray $\rho$.

\begin{proposition}
  Suppose that $0<\rmRe(t)<1$, $\rmIm(t)>0$, $\rmRe(\check\lambda)>0$, and $\rmRe(t)<\rmRe(\check\lambda+1)$. Then
  \begin{equation}\label{crosscap-positive-real-sectorial-sum}
    \mathcal S_{\mathbb R_{>0}}\widetilde F_\mathrm{cc}(\lambda,t)
    =\frac12\log F_2\vargs*{t+\check\lambda}{2\check\lambda,1}
    -\log F_2\vargs*{t+\check\lambda}{2\check\lambda,2}.
  \end{equation}
\end{proposition}
\begin{proof}
  Applying \zcref{primitive-positive-real-sum} to each primitive term gives
  \begin{align*}
    \mathcal S_{\mathbb R_{>0}}\widetilde F(\lambda,t)
    &=i\log F_2\vargs*{\frac{t+\tau}{2}}{\tau,1},\\*
    \mathcal S_{\mathbb R_{>0}}\widetilde F(2\lambda,2t)
    &=i\log F_2\vargs*{t+\tau}{2\tau,1}.
  \end{align*}
  By simultaneous rescaling \cite[Proposition~4.1(ii)]{bridgeland-2020-7a4fdc8b}, $F_2\vargs*{(t+\tau)/2}{\tau,1}=F_2\vargs{t+\tau}{2\tau,2}$. Substitution proves \eqref{crosscap-positive-real-sectorial-sum}.
\end{proof}

\begin{proposition}
  Suppose that $\rmIm(t)>0$. Then, for $\varepsilon\in\{+1,-1\}$,
  \begin{equation}\label{crosscap-stokes-jump}
    \disc_{\varepsilon\ell_k}\widetilde F_\mathrm{cc}(\lambda,t)
    =\varepsilon\frac{(-1)^k}{2}\log\mleft(1+e^{-\varepsilon\pi i(t+k)/\check\lambda}\mright).
  \end{equation}
\end{proposition}
\begin{proof}
  Apply \zcref{primitive-stokes-jump} to \eqref{crosscap-sectorial-sum}. On the ray indexed by $k$, the second primitive term contributes $-(\varepsilon/2)\log(1+e^{-\varepsilon\pi i(t+k)/\tau})$. When $k$ is even, the first primitive term contributes an additional $\varepsilon\log(1+e^{-\varepsilon\pi i(t+k)/\tau})$. Adding the two contributions gives the formula.
\end{proof}

\begin{proposition}\label{crosscap-accumulation-limit}
  Suppose that $0<\rmRe(t)<1$, $\rmIm(t)>0$, and $\rmIm(\check\lambda)>0$. For each $k$, let $\rho_k\subset\Sigma_k$ be a ray. For all sufficiently large $k$, the sectorial sums are defined and
  \begin{equation}\label{crosscap-accumulation-sector}
    \lim_{k\to\infty}\mathcal S_{\rho_k}\widetilde F_\mathrm{cc}(\lambda,t)
    =\frac12\log\frac{\mleft(-e^{\pi i(t+\check\lambda)};e^{2\pi i\check\lambda}\mright)_\infty}{\mleft(e^{\pi i(t+\check\lambda)};e^{2\pi i\check\lambda}\mright)_\infty}
    =i\sum_{r\textup{ odd}}\frac{e^{\pi irt}}{2r\sin(\pi r\check\lambda)}.
  \end{equation}
  Here the sum is over positive odd integers. Both products are nonzero, the logarithms are their convergent power-series branches, and the sine series converges normally on compact subsets of the stated domain.
\end{proposition}
\begin{proof}
  Put $a=e^{\pi i(t+\check\lambda)}$ and $q=e^{2\pi i\check\lambda}$. Applying \zcref{primitive-accumulation-limit} to both terms of \eqref{crosscap-sectorial-sum} gives
  \begin{align*}
    \lim_{k\to\infty}\mathcal S_{\rho_k}\widetilde F_\mathrm{cc}(\lambda,t)
    &=-\log(a;q)_\infty+\frac12\log(a^2;q^2)_\infty\\
    &=\frac12\log\frac{(-a;q)_\infty}{(a;q)_\infty}.
  \end{align*}
  Since $(a^2;q^2)_\infty=(a;q)_\infty(-a;q)_\infty$, expansion of the quotient gives $\sum_{r\textup{ odd}}a^r/(r(1-q^r))$. The identity $1-e^{2\pi ir\check\lambda}=-2ie^{\pi ir\check\lambda}\sin(\pi r\check\lambda)$ yields the sine series. Finally, $\abs a<1$ and $\abs q<1$ fix the logarithm branches and give normal convergence.
\end{proof}

Every Borel pole is of the form $\pi i(t+k)m$. When $\rmIm(t)>0$, its real part is $-\pi m\rmIm(t)\neq0$. Thus the imaginary accumulation rays $\pm\ell_\infty$ contain no Borel singularities: they belong to the closure of the Stokes directions but are not themselves Stokes rays of the scalar transform.

\begin{corollary}\label{crosscap-accumulation-ray-no-jump}
  Suppose that $\rmIm(t)>0$, and choose rays $\rho_k\subset\Sigma_k$. For $M\geq N$, set
  \[
    \Delta_{N,M}(\check\lambda,t)
    \coloneqq
    \mathcal S_{\rho_{M+1}}\widetilde F_\mathrm{cc}(2\pi\check\lambda,t)
    -\mathcal S_{\rho_N}\widetilde F_\mathrm{cc}(2\pi\check\lambda,t).
  \]
  For a compact set $K\subset\mathbb H_{\ell_\infty}$ and $f\colon K\to\mathbb C$, write $\norm{f}_K=\sup_{\check\lambda\in K}\abs{f(\check\lambda)}$. For all sufficiently large $N$, the functions $\Delta_{N,M}$ are defined on $K$ for every $M\geq N$, and
  \[
    \sup_{M\geq N}\norm*{\Delta_{N,M}(-,t)}_K
    \longrightarrow0
    \qquad(N\to\infty).
  \]
  Hence $\exp(\Delta_{N,M}(\check\lambda,t))\to1$ uniformly in $M\geq N$ and $\check\lambda\in K$. The reflected sectorial sums satisfy the same conclusion on compact subsets of $\mathbb H_{-\ell_\infty}$. Each limiting function is holomorphic on its half-plane and, under the hypotheses of \zcref{crosscap-accumulation-limit}, is given by \eqref{crosscap-accumulation-sector}.
\end{corollary}
\begin{proof}
  Fix a compact $K\subset\mathbb H_{\ell_\infty}$. For all sufficiently large $k$, it lies in the common domain of the two adjacent sectorial sums. Since neither $\Sigma_k$ nor $\Sigma_{k+1}$ contains a Borel singularity, deforming $\rho_k$ and $\rho_{k+1}$ to the lateral rays at $\ell_k$ identifies their difference, up to sign, with $\disc_{\ell_k}\widetilde F_\mathrm{cc}$. The inequality $\rmRe(\pi i/\check\lambda)>0$ on this half-plane gives constants $c_K,C_K>0$ such that
  \[
    \sup_{\check\lambda\in K}
    \abs*{\exp\mleft(-\frac{\pi i(t+k)}{\check\lambda}\mright)}
    \leq C_Ke^{-c_Kk}.
  \]
  For sufficiently small $z$, use $\abs{\log(1+z)}\leq2\abs z$ in \eqref{crosscap-stokes-jump}. Uniformly for $M\geq N$ and all sufficiently large $N$, this gives
  \begin{align*}
    \norm*{\Delta_{N,M}(-,t)}_K
    &\leq
    \sum_{k=N}^{M}
    \sup_{\check\lambda\in K}
    \abs*{
      \disc_{\ell_k}\widetilde F_\mathrm{cc}
      (2\pi\check\lambda,t)
    }\\
    &\leq
    C_K\sum_{k=N}^{M}e^{-c_Kk}
    \leq
    \frac{C_Ke^{-c_KN}}{1-e^{-c_K}}
    \longrightarrow0.
  \end{align*}
  Taking the supremum over $M\geq N$ proves the first assertion, and $\abs{e^z-1}\leq e^{\abs z}\abs z$ gives the multiplicative limit. Reflection proves the statement at $-\ell_\infty$.
\end{proof}

\subsection{Resurgence of the large-\texorpdfstring{$N$}{N}
\texorpdfstring{$\SO/\Sp$}{SO/Sp} orientifold free energy}
Combining the crosscap formulas with the resolved-conifold summation theorem \cite[Theorem~2.1]{alim-2023-de7fb1f0} recovers the large-$N$ $\SO/\Sp$ orientifold free energy.

For $z\in\mathbb C$ and $\omega_1,\omega_2\in\mathbb C^\ast$ with $\omega_1/\omega_2\notin\mathbb R_{<0}$, set \cite[Section~4.2]{bridgeland-2020-7a4fdc8b}
\[
  G_3\vargs*{z}{\omega_1, \omega_2}\coloneqq\exp \mleft(\frac{\pi
  i}{6} B_{3,3}\vargs*{z+\omega_1}{\omega_1, \omega_1, \omega_2}\mright)\sin_3\vargs*{z+\omega_1}{\omega_1, \omega_1, \omega_2},
\]
where $\sin_3$ is the triple sine function and $B_{3,3}\vargs{z}{\omega_1,\omega_2,\omega_3}$ the multiple Bernoulli polynomial. The two periods of $G_3\vargs{z}{\omega_1,\omega_2}$ are not symmetric.

The full formal free energy is
\[
  \widetilde F_{\SO/\Sp}(\lambda,t)
  =-\frac12\widetilde F_\mathrm{GV}(\lambda,t)+\epsilon\widetilde F_\mathrm{cc}(\lambda,t),
\]
where
\[
  \widetilde F_\mathrm{GV}(\lambda,t)
  \coloneqq\frac{1}{\lambda^2}\Li_3(Q)+\frac{B_2}{2}\Li_1(Q)
  +\sum_{g=2}^{\infty}\lambda^{2g-2}\frac{(-1)^{g-1}B_{2g}}{2g(2g-2)!}\Li_{3-2g}(Q)
\]
is the resolved-conifold series of \cite[eq.~(2.11)]{alim-2023-de7fb1f0}; here $\check\lambda=\lambda/(2\pi)$ and $Q=e^{2\pi i t}$.

Suppose that $t\in\mathbb C^\ast$ and $\abs{\rmRe(t)}<1/2$. The Borel transform of the positive-degree part of $\widetilde F_{\SO/\Sp}$ converges for $\abs{\xi}<\pi\abs t$ and extends meromorphically. At $\xi=\pi i(t+k)m$, it has a simple pole for odd $m$ and a double pole for even $m$. The resolved-conifold poles occur at $2\pi i(t+k)m=\pi i(t+k)(2m)$ and hence correspond exactly to even crosscap multipliers. Since $\rmIm(t)>0$, each point has at most one label $(k,m)$. For even $m$, adding the simple crosscap principal part leaves the nonzero double-pole coefficient unchanged; for odd $m$, only the crosscap pole occurs.

Decompose both Laurent series into their parts of degree less than $1$ and at least $1$. By \zcref{sectorial-sum-laurent-series} and linearity of the Laplace integral, every admissible ray $\rho$ and $\lambda\in\mathbb H_\rho$ then satisfy
\begin{equation}\label{orientifold-sectorial-sum}
  \mathcal S_\rho\widetilde F_{\SO/\Sp}(\lambda,t)
  \coloneqq-\frac12\mathcal S_\rho\widetilde F_\mathrm{GV}(\lambda,t)
  +\epsilon\mathcal S_\rho\widetilde F_\mathrm{cc}(\lambda,t).
\end{equation}
For $\rho=\mathbb R_{>0}$, assume in addition that $0<\rmRe(t)<1$, $\rmRe(\check\lambda)>0$, and $\rmRe(t)<\rmRe(\check\lambda+1)$. Combining \cite[Theorem~2.1(ii)]{alim-2023-de7fb1f0} with \eqref{crosscap-positive-real-sectorial-sum} gives
\[
  \mathcal S_{\mathbb R_{>0}}\widetilde F_{\SO/\Sp}(\lambda,t)
  =-\frac12\log G_3\vargs*{t}{\check\lambda,1}
  +\epsilon\left(\frac12\log F_2\vargs*{t+\check\lambda}{2\check\lambda,1}
  -\log F_2\vargs*{t+\check\lambda}{2\check\lambda,2}\right).
\]

\begin{corollary}
  Suppose that $\rmIm(t)>0$. Then
  \begin{equation*}
    \disc_{\ell_k}\widetilde F_{\SO/\Sp}(\lambda,t)
    =\frac{1}{4\pi i}\partial_{\check\lambda}\mleft(\check\lambda\Li_2\mleft(e^{-2\pi i(t+k)/\check\lambda}\mright)\mright)
    +\epsilon\frac{(-1)^k}{2}\log\mleft(1+e^{-\pi i(t+k)/\check\lambda}\mright).
  \end{equation*}
\end{corollary}
\begin{proof}
  Combine the resolved-conifold discontinuity \cite[Theorem~2.1(iii)]{alim-2023-de7fb1f0} with \eqref{crosscap-stokes-jump}, using the linear combination in \eqref{orientifold-sectorial-sum}.
\end{proof}

\begin{corollary}
  Suppose that $0<\rmRe(t)<1$, $\rmIm(t)>0$, and $\rmIm(\check\lambda)>0$. For each $k$, let $\rho_k\subset\Sigma_k$ be a ray. For all sufficiently large $k$, the sectorial sums are defined and
  \begin{equation}\label{orientifold-accumulation-sector}
    \lim_{k\to\infty}\mathcal S_{\rho_k}\widetilde F_{\SO/\Sp}(\lambda,t)
    =-\frac12\mathcal F_{\mathrm{GV}}(\lambda,t)
    +\epsilon i\sum_{r\textup{ odd}}\frac{e^{\pi irt}}{2r\sin(\lambda r/2)}
    =\mathcal F_{\SO/\Sp}(\lambda,t).
  \end{equation}
  The sum is over positive odd integers, and both multicover series converge normally.
\end{corollary}
\begin{proof}
  The rays in \cite[Theorem~2.1(iv), eq.~(2.16)]{alim-2023-de7fb1f0} have the opposite orientation. After reflection, that result identifies the ordinary limit with $\mathcal F_{\mathrm{GV}}(\lambda,t)$ on a nonempty open subset of $\rmIm(\check\lambda)>0$. Its jump estimates extend the limiting sectorial sum holomorphically throughout this half-plane. The normally convergent Gopakumar--Vafa series is holomorphic there as well, so the identity theorem gives equality on the full domain.

  The crosscap limit is \zcref{crosscap-accumulation-sector} with $\lambda=2\pi\check\lambda$. Adding the two contributions proves \eqref{orientifold-accumulation-sector}.
\end{proof}

\section{The crosscap BPS Riemann--Hilbert problem}\label{crosscap-bps-rh}

We construct the crosscap BPS structure and test the full Bridgeland problem at its two accumulation directions. After identifying the resulting obstructions, we formulate and solve a modified problem.

\subsection{BPS structures and the full Riemann--Hilbert problem}

The full Riemann--Hilbert problem imposes normalization on every non-active ray, including accumulation directions. We use the conventions for BPS structures, active and null classes, ray-finiteness, convergence, uncoupledness, and twisted tori in \cite[Section~2]{bridgeland-2020-7a4fdc8b} and \cite[Definitions~2.2--2.3]{bridgeland-2019-6bdcc39c}. Write $\mathbb T$ for the twisted torus and $x_\gamma$ for the character of charge $\gamma$. The rational Donaldson--Thomas invariant is
\[
  \DT(\gamma)
  =\sum_{\gamma=m\alpha}
  \frac{\Omega(\alpha)}{m^2}.
\]
For a ray $\ell\subset\mathbb C^\ast$, its Hamiltonian is \cite[eq.~(10)]{bridgeland-2020-7a4fdc8b}
\[
  \DT(\ell)
  =-\sum_{Z(\gamma)\in\ell}\DT(\gamma)x_\gamma.
\]
For an acute sector $\Delta\subset\mathbb C^\ast$ and $R>0$, let $\mathcal U_\Delta(R)\subset\mathbb T$ be the interior of
\[
  \setc*{g\in\mathbb T}{
    Z(\gamma)\in\Delta,\ \Omega(\gamma)\neq0
    \implies
    \abs*{x_\gamma(g)}<e^{-R\norm*{\gamma}}
  }.
\]
This open subset is nonempty. For a convergent BPS structure and sufficiently large $R$, the ray Hamiltonians converge absolutely on $\mathcal U_\Delta(R)$, their time-$1$ flows define holomorphic Poisson embeddings $\mathbf S(\ell)$, and clockwise products over active rays of bounded height converge to a sector map $\mathbf S(\Delta)$ \cite[Section~2.3 and Proposition~2.2]{bridgeland-2020-7a4fdc8b}. We write
\[
  \mathbf S(\ell)^*
  =\exp\mleft(\{\DT(\ell),-\}\mright)
\]
for the time-$1$ pullback on this domain.

The character formula of \cite[Proposition~2.3]{bridgeland-2020-7a4fdc8b} remains valid analytically under finiteness and pairwise orthogonality, without assuming integral invariants.

\begin{proposition}
  \label{analytic-ray-character-formula}
  Fix a ray $\ell\subset\mathbb C^\ast$ and suppose that
  \begin{enumerate}
    \item there are only finitely many active classes $\gamma\in\Gamma$ such that $Z(\gamma)\in\ell$;
    \item any two such active classes $\gamma_i,\gamma_j$ satisfy $\langle\gamma_i,\gamma_j\rangle=0$.
  \end{enumerate}
  Products and sums indexed by $Z(\gamma)\in\ell$ are taken over active classes. If an acute sector $\Delta$ contains $\ell$, then for all sufficiently large $R>0$ the partially defined BPS automorphism satisfies
  \begin{equation}\label{analytic-ray-pullback}
    \mathbf S(\ell)^*(x_\beta)
    =x_\beta
    \prod_{Z(\gamma)\in\ell}
    (1-x_\gamma)^{\Omega(\gamma)\langle\gamma,\beta\rangle},
    \qquad \beta\in\Gamma,
  \end{equation}
  on $\mathcal U_\Delta(R)$. Each power is defined by
  \[
    (1-x_\gamma)^q
    =\exp\mleft(q\log(1-x_\gamma)\mright),
    \qquad
    \log(1-x_\gamma)
    =-\sum_{m=1}^{\infty}\frac{x_\gamma^m}{m},
  \]
  where the logarithm is its single-valued holomorphic power-series branch on $\mathcal U_\Delta(R)$.
\end{proposition}

\begin{proof}
  Expanding the multicover transform and using $x_{m\gamma}=x_\gamma^m$ yields
  \[
    \DT(\ell)
    =-\sum_{Z(\gamma)\in\ell}
    \Omega(\gamma)\Li_2(x_\gamma).
  \]
  Hence
  \[
    \{\DT(\ell),x_\beta\}
    =\left(
      \sum_{Z(\gamma)\in\ell}
      \Omega(\gamma)\langle\gamma,\beta\rangle
      \log(1-x_\gamma)
    \right)x_\beta.
  \]
  By isotropy, each $x_\gamma$ occurring in the coefficient is constant along the Hamiltonian flow. Integration from time $0$ to $1$ proves \eqref{analytic-ray-pullback}.
\end{proof}

When the invariants on $\ell$ are integral, the exponents in \eqref{analytic-ray-pullback} are integers and $\mathbf S(\ell)$ extends birationally to $\mathbb T$ \cite[Proposition~2.3]{bridgeland-2020-7a4fdc8b}. For nonintegral invariants, the power-series branches still define a holomorphic map on $\mathcal U_\Delta(R)$, but fractional powers generally branch around $x_\gamma=1$ and are not rational on $\mathbb T$. The crosscap structure has such exponents, so its BPS automorphisms remain partially defined.

\begin{problem}[\cite[Problem~2.4]{bridgeland-2020-7a4fdc8b}]\label{bridgeland-rh-problem}
  Fix a convergent BPS structure $(\Gamma,Z,\Omega)$ and $\xi\in\mathbb T$. For each non-active ray $\rho$, find a holomorphic map
  \[
    \Phi_\rho\colon\mathbb H_\rho\longrightarrow\mathbb T
  \]
  and set $\Phi_{\rho,\gamma}=x_\gamma\circ\Phi_\rho$. Require the following.
  \begin{enumerate}[label=\textup{(RH\arabic*)}]
    \item Suppose that two non-active rays $\rho_1,\rho_2$ form the boundary rays of an acute sector $\Delta$, taken in clockwise order. Then
      \[
        \Phi_{\rho_2}(\lambda_{\mathrm B})=\mathbf S(\Delta)\mleft(\Phi_{\rho_1}(\lambda_{\mathrm B})\mright)
      \]
      for every $\lambda_{\mathrm B}\in\mathbb H_{\rho_1}\cap\mathbb H_{\rho_2}$ satisfying $0<\abs{\lambda_{\mathrm B}}\ll1$.
    \item For every non-active ray $\rho$ and every class $\gamma\in\Gamma$,
      \[
        \exp\mleft(\frac{Z(\gamma)}{\lambda_{\mathrm B}}\mright)\Phi_{\rho,\gamma}(\lambda_{\mathrm B})
        \longrightarrow\xi(\gamma)
      \]
      as $\lambda_{\mathrm B}\to0$ in the half-plane $\mathbb H_\rho$.
    \item For every non-active ray $\rho$ and every class $\gamma\in\Gamma$, there exists $k>0$ such that
      \[
        \abs*{\lambda_{\mathrm B}}^{-k}
        <\abs*{\Phi_{\rho,\gamma}(\lambda_{\mathrm B})}
        <\abs*{\lambda_{\mathrm B}}^k
      \]
      for every $\lambda_{\mathrm B}\in\mathbb H_\rho$ satisfying $\abs{\lambda_{\mathrm B}}\gg0$.
  \end{enumerate}
\end{problem}

\subsection{The crosscap BPS structure}

The crosscap BPS structures are parametrized by the open set
\[
  \mathcal M_+
  \coloneqq
  \setc*{(v,w)\in\mathbb C^2}{
    w\neq0, v+nw\neq0\text{ for all }n\in\mathbb Z,
    \rmIm(v/w)>0
  }.
\]
The deleted hyperplanes exclude vanishing active periods. Fix $(v,w)\in\mathcal M_+$ throughout this section. The upper-half-plane condition is precisely the chamber in which the $Q$-series converge and determines the clockwise ordering of the active rays.

Set
\begin{align*}
  \Gamma_{\mathrm{act}}
  &\coloneqq
  \mathbb Z\gamma_0\oplus\mathbb Z\gamma_1\oplus\mathbb Zd,\\*
  \Lambda
  &\coloneqq
  \operatorname{Ann}_{\Gamma_{\mathrm{act}}^\vee}(d)
  =
  \setc*{e\in\Gamma_{\mathrm{act}}^\vee}{e(d)=0}
  \cong
  (\Gamma_{\mathrm{act}}/\mathbb Zd)^\vee,\\*
  \Gamma
  &\coloneqq
  \Gamma_{\mathrm{act}}\oplus\Lambda.
\end{align*}
Choose the basis $\nu_0,\nu_1$ of $\Lambda$ given by
\[
  \begin{aligned}
    \nu_0(\gamma_0)&=1,&
    \nu_0(\gamma_1)&=1,\\
    \nu_1(\gamma_0)&=0,&
    \nu_1(\gamma_1)&=1.
  \end{aligned}
\]
Then $\Gamma$ is free of rank five. For $\alpha,\beta\in\Gamma_{\mathrm{act}}$ and $e,f\in\Lambda$, set
\begin{equation}\label{orientifold-skew-form}
  \langle(\alpha,e),(\beta,f)\rangle
  \coloneqq
  e(\beta)-f(\alpha).
\end{equation}
Both summands are isotropic, and in the chosen basis
\begin{equation}\label{orientifold-mixed-pairing}
  \begin{aligned}
    \langle \gamma_0,\nu_0\rangle&=-1,&
    \langle \gamma_0,\nu_1\rangle&=0,\\
    \langle \gamma_1,\nu_0\rangle&=-1,&
    \langle \gamma_1,\nu_1\rangle&=-1,\\
    \langle d,\nu_0\rangle&=0,&
    \langle d,\nu_1\rangle&=0.
  \end{aligned}
\end{equation}
The radical is $\mathbb Zd$. Modulo this radical, the images of $\gamma_0,\gamma_1,\nu_1-\nu_0,-\nu_1$ form a symplectic basis.

Define $\xi_0\colon\Gamma\to\{+1,-1\}$ by
\[
  \xi_0(\alpha+e)=(-1)^{\langle\alpha,e\rangle},
  \qquad
  \alpha\in\Gamma_{\mathrm{act}},\quad e\in\Lambda.
\]
Its restriction to either summand is $1$.

Define the central charge by
\begin{equation}
  \begin{aligned}
    Z(\gamma_0)&=\pi iv,&
    Z(\gamma_1)&=\pi i(v+w),&
    Z(d)&=2\pi iw,\\
    Z(\nu_0)&=0,&
    Z(\nu_1)&=0.
  \end{aligned}
\end{equation}
For $m\in\mathbb Z$, put
\begin{equation}\label{orientifold-active-classes}
  \gamma_{2m}=\gamma_0+md,
  \qquad
  \gamma_{2m+1}=\gamma_1+md.
\end{equation}
These classes satisfy
\begin{equation}
  Z(\gamma_k)=\pi i(v+kw),
  \qquad k\in\mathbb Z,
\end{equation}
and $\gamma_{k+2}=\gamma_k+d$. Thus $d$ carries the nonzero period of the two-step translation, while the framing covectors depend only on active charges modulo $\mathbb Zd$. Under $(Z(\gamma_0),Z(\gamma_1))/(\pi i)=(v,v+w)$, the covectors $\nu_0$ and $\nu_1$ correspond to the $v$- and $w$-coordinates. Replacing $\Lambda$ by the full dual $\Gamma_{\mathrm{act}}^\vee$ would add a covector paired with $d$ and produce a rank-six problem. Along the active tower,
\begin{equation}\label{crosscap-active-pairings}
  \langle \gamma_k,\nu_0\rangle=-1,
  \qquad
  \langle \gamma_k,\nu_1\rangle
  =\frac{-1+(-1)^k}{2}.
\end{equation}
Define the BPS invariants by
\begin{equation}\label{orientifold-bps-invariants}
  \Omega(\gamma)
  =
  \begin{cases}
    -(-1)^k,& \gamma=\pm\gamma_k,\\[0.6em]
    \dfrac{(-1)^k}{4},& \gamma=\pm2\gamma_k,\\[0.6em]
    0,& \text{otherwise.}
  \end{cases}
\end{equation}
For $r\geq1$ and $\varepsilon\in\{+1,-1\}$, the multicover relation gives
\begin{equation}\label{crosscap-dt-coefficients}
  \DT(r\varepsilon\gamma_k)
  =
  \begin{cases}
    -\dfrac{(-1)^k}{r^2},&r\textup{ odd},\\[0.6em]
    0,&r\textup{ even}.
  \end{cases}
\end{equation}
For odd $r$, the only contribution is $\Omega(\varepsilon\gamma_k)/r^2$; for even $r$, the terms from $\varepsilon\gamma_k$ and $2\varepsilon\gamma_k$ cancel.

Define
\[
  \ell_k\coloneqq\mathbb R_{>0}\cdot Z(\gamma_k),
  \qquad
  \ell_\infty\coloneqq\mathbb R_{>0}\cdot\pi iw.
\]
If $t=v/w$ and $\check\lambda=\lambda_{\mathrm B}/w$, multiplication by $w$ identifies the Borel-plane rays and sectors with their central-charge counterparts. In particular, $(t+k)/\check\lambda=(v+kw)/\lambda_{\mathrm B}$.

The positive rays $\ell_k$ are pairwise distinct and ordered clockwise by increasing $k$. They tend to $\ell_\infty$ as $k\to+\infty$ and to $-\ell_\infty$ as $k\to-\infty$; the remaining active rays are their negatives. Hence the closure of the active-ray set is $\{\pm\ell_k:k\in\mathbb Z\}\cup\{\pm\ell_\infty\}$.

By \eqref{orientifold-bps-invariants}, $\Omega(rd)=0$ for every $r\neq0$. Thus the rays $\pm\ell_\infty$ are non-active despite belonging to the active-ray closure, and both Riemann--Hilbert problems considered here include them among their indexing rays.

The resulting ray configuration is shown in the figure.

\begin{figure}[htbp]
  \centering
  \begin{tikzpicture}[
      scale=1,
      every node/.style={font=\small},
      ray/.style={->,thin},
      origin/.style={circle,draw,fill=white,inner sep=1.8pt}
    ]
    \coordinate (O) at (0,0);

    \draw[ray] (O) -- (4.35,0);
    \draw[ray] (O) -- (-4.35,0);
    \node[anchor=west] at (4.35+0.15,0) {$\ell_\infty$};
    \node[anchor=east] at (-4.35-0.15,0) {$-\ell_\infty$};

    \foreach \n in {-2,-1,0,1,2}{
      \coordinate (P\n) at ({\n*1.55},2.75);
      \coordinate (M\n) at ({-\n*1.55},-2.75);
      \coordinate (U\n) at ($(O)!0.55!(P\n)$);
      \coordinate (N\n) at ($(O)!0.55!(M\n)$);

      \draw[ray] (O) -- (P\n);
      \fill (U\n) circle (1.1pt);
      \fill (P\n) circle (1.1pt);

      \draw[ray] (O) -- (M\n);
      \fill (N\n) circle (1.1pt);
      \fill (M\n) circle (1.1pt);
    }

    \node[above] at (P-1) {$\ell_{-1}$};
    \node[above] at (P0) {$\ell_0$};
    \node[above] at (P1) {$\ell_1$};
    \node[anchor=east,inner sep=1pt,font=\scriptsize] at
    ($(U0)+(-0.05,0.04)$)
    {$\pi iv$};
    \node[anchor=west,inner sep=1pt,font=\scriptsize] at ($(U2)+(0.18,0.04)$)
    {$\pi i(v+2w)$};

    \node[below] at (M1) {$-\ell_1$};
    \node[below] at (M0) {$-\ell_0$};
    \node[below] at (M-1) {$-\ell_{-1}$};

    \node at (3.55,1.15) {$\cdots$};
    \node at (-3.55,1.15) {$\cdots$};
    \node at (3.55,-1.15) {$\cdots$};
    \node at (-3.55,-1.15) {$\cdots$};

    \node[origin] at (O) {};
  \end{tikzpicture}
  \caption{The ray diagram associated to a point $(v,w)\in\mathcal M_+$.}
\end{figure}

Introduce the exponential coordinates and their chosen half-powers
\[
  \begin{aligned}
    x&=\exp\mleft(-\frac{2\pi iv}{\lambda_{\mathrm B}}\mright),&
    x^{1/2}&\coloneqq\exp\mleft(-\frac{\pi iv}{\lambda_{\mathrm B}}\mright),\\
    q&=\exp\mleft(-\frac{2\pi iw}{\lambda_{\mathrm B}}\mright),&
    q^{1/2}&\coloneqq\exp\mleft(-\frac{\pi iw}{\lambda_{\mathrm B}}\mright).
  \end{aligned}
\]
For $k\in\mathbb Z$, write $q^{k/2}=(q^{1/2})^k$. Since $Z(\gamma_k)=\pi i(v+kw)$, the active character is $x^{1/2}q^{k/2}$ and the finite-ray multiplier becomes
\begin{equation}\label{crosscap-ratio}
  R_k(v,w,\lambda_{\mathrm B})
  =\frac{1-x^{1/2}q^{k/2}}{1+x^{1/2}q^{k/2}}.
\end{equation}

\begin{proposition}\label{orientifold-bps-structure}
  The triple $(\Gamma,Z,\Omega)$ is a convergent, ray-finite, uncoupled BPS structure with rational invariants.
\end{proposition}
\begin{proof}
  \zcref{orientifold-bps-invariants} is symmetric under $\gamma\mapsto-\gamma$. For any norm on $\Gamma\otimes\mathbb R$, \eqref{orientifold-active-classes} and $\rmIm(v/w)>0$ give constants $C_1,C_2>0$ such that
  \[
    \norm*{\gamma_k}\leq C_1(1+\abs*{k}),
    \qquad
    \abs*{Z(\gamma_k)}=\pi\abs*{v+kw}
    \geq C_2(1+\abs*{k})
  \]
  for every $k\in\mathbb Z$, after decreasing $C_2$ for finitely many indices. This is the support property. For all sufficiently large $R$, it also implies
  \[
    \sum_{\gamma\in\Gamma}
    \abs*{\Omega(\gamma)}e^{-R\abs*{Z(\gamma)}}
    =2\sum_{k\in\mathbb Z}
    \left(e^{-R\pi\abs*{v+kw}}
    +\frac14e^{-2R\pi\abs*{v+kw}}\right)<\infty.
  \]
  If $\ell_j=\ell_k$, then the upper-half-plane points $v/w+j$ and $v/w+k$ have the same argument, forcing $j=k$. Each active ray therefore contains only $\gamma_k,2\gamma_k$ or their negatives. Finally, every active class lies in the isotropic summand $\Gamma_{\mathrm{act}}$, so the structure is uncoupled.
\end{proof}

The function $\xi_0$ defines a point of the twisted torus. To see this, write $\gamma=\alpha+e$ and $\gamma'=\beta+f$, with $\alpha,\beta\in\Gamma_{\mathrm{act}}$ and $e,f\in\Lambda$. Since both summands are isotropic, reduction modulo $2$ gives
\[
  \langle\alpha+\beta,e+f\rangle
  \equiv
  \langle\gamma,\gamma'\rangle
  +\langle\alpha,e\rangle
  +\langle\beta,f\rangle.
\]
Thus
\[
  \xi_0(\gamma+\gamma')
  =(-1)^{\langle\gamma,\gamma'\rangle}
  \xi_0(\gamma)\xi_0(\gamma').
\]
Thus $\xi_0\in\mathbb T$. Division by $\xi_0$ identifies
\[
  \mathbb T\xrightarrow{\sim}
  \Hom_{\mathbb Z}(\Gamma,\mathbb C^\ast)
  \cong
  \Hom(\Gamma_{\mathrm{act}},\mathbb C^\ast)
  \times\Hom(\Lambda,\mathbb C^\ast),
  \qquad
  g\longmapsto\xi_0^{-1}g,
\]
since the quotient of two twisted characters is an ordinary character.

\begin{proposition}
  Fix $\varepsilon\in\{+1,-1\}$, $k\in\mathbb Z$, and an acute sector $\Delta$ containing $\varepsilon\ell_k$. For all sufficiently large $R$, the ray Hamiltonian is the holomorphic function
  \[
    \DT(\varepsilon\ell_k)
    =\frac{(-1)^k}{2}
    \mleft(\Li_2(x_{\varepsilon\gamma_k})
    -\Li_2(-x_{\varepsilon\gamma_k})\mright)
    \colon \mathcal U_\Delta(R)\longrightarrow\mathbb C.
  \]
  Its time-$1$ flow is a holomorphic Poisson embedding $\mathbf S(\varepsilon\ell_k)\colon\mathcal U_\Delta(R)\to\mathbb T$ with pullback
  \begin{equation}\label{orientifold-ray-pullback}
    \mathbf S(\varepsilon\ell_k)^*(x_{\gamma})
    =x_{\gamma}
    \mleft(\frac{1-x_{\varepsilon\gamma_k}}
    {1+x_{\varepsilon\gamma_k}}\mright)^{
    -\varepsilon(-1)^k\langle \gamma_k,\gamma\rangle/2},
    \qquad \gamma\in\Gamma.
  \end{equation}
  The fractional power is the single-valued power-series branch at $x_{\varepsilon\gamma_k}=0$.
\end{proposition}
\begin{proof}
  Substitution of \eqref{crosscap-dt-coefficients} gives the Hamiltonian. Absolute convergence and the holomorphic Poisson flow follow from \cite[Proposition~2.2]{bridgeland-2020-7a4fdc8b}. The active classes on $\varepsilon\ell_k$ are $\varepsilon\gamma_k$ and $2\varepsilon\gamma_k$; applying \zcref{analytic-ray-character-formula} and $x_{2\varepsilon\gamma_k}=x_{\varepsilon\gamma_k}^2$ gives
  \begin{align*}
    \mathbf S(\varepsilon\ell_k)^*(x_{\gamma})
    &=x_{\gamma}
    (1-x_{\varepsilon\gamma_k})^{-\varepsilon(-1)^k\langle \gamma_k,\gamma\rangle}
    (1-x_{2\varepsilon\gamma_k})^{\varepsilon(-1)^k\langle \gamma_k,\gamma\rangle/2}\\
    &=x_{\gamma}
    \mleft(\frac{1-x_{\varepsilon\gamma_k}}
    {1+x_{\varepsilon\gamma_k}}\mright)^{
    -\varepsilon(-1)^k\langle \gamma_k,\gamma\rangle/2}.
  \end{align*}
\end{proof}

For $\varepsilon\in\{+1,-1\}$ and $k\in\mathbb Z$, the multiplier of $x_{\nu_0}$ contains $((1-x_{\varepsilon\gamma_k})/(1+x_{\varepsilon\gamma_k}))^{1/2}$. Its divisor has odd valuation at both $x_{\varepsilon\gamma_k}=1$ and $x_{\varepsilon\gamma_k}=-1$, so the multiplier is not rational. Consequently, $\mathbf S(\varepsilon\ell_k)$ has no birational extension to $\mathbb T$.

If $\gamma=\alpha+e$, with $\alpha\in\Gamma_{\mathrm{act}}$ and $e\in\Lambda$, then isotropy gives $\langle\gamma_k,\gamma\rangle=\langle\gamma_k,e\rangle$. Thus $\mathbf S(\varepsilon\ell_k)^*$ fixes every active character. On the remaining basis elements, \eqref{orientifold-ray-pullback} becomes
\begin{align*}
  \mathbf S(\varepsilon\ell_k)^*(x_{\nu_0})
  &=x_{\nu_0}R_k(v,w,\varepsilon\lambda_{\mathrm B})^{\varepsilon(-1)^k/2},\\
  \mathbf S(\varepsilon\ell_k)^*(x_{\nu_1})
  &=
  \begin{cases}
    x_{\nu_1},&k\textup{ even},\\
    x_{\nu_1}R_k(v,w,\varepsilon\lambda_{\mathrm B})^{-\varepsilon/2},
    &k\textup{ odd}.
  \end{cases}
\end{align*}
Together with the fixed active characters, these formulas determine the ray map. To make the parameter dependence explicit, write $\ell_k(v,w)=\mathbb R_{>0}\cdot\pi i(v+kw)$ and let $\Sigma_n(v,w)$ be the open sector bounded by $\ell_{n-1}(v,w)$ and $\ell_n(v,w)$.

\subsection{Failure of the full Bridgeland problem}

\begin{lemma}\label{ordinary-ray-null-character}
  Let $\{\Phi_\rho\}$ be indexed by all non-active rays and satisfy \textup{(RH1)} of \zcref{bridgeland-rh-problem}. Assume that \textup{(RH2)} with constant term $\xi_0$ and \textup{(RH3)} hold away from $\pm\ell_\infty$. If $\gamma\in\Gamma$ is null, then
  \[
    \Phi_{\rho,\gamma}(\lambda_{\mathrm B})
    =\xi_0(\gamma)\exp\mleft(-\frac{Z(\gamma)}{\lambda_{\mathrm B}}\mright)
  \]
  on every half-plane in the family.
\end{lemma}
\begin{proof}
  For a null class $\gamma$, RH1 glues the normalized characters
  \[
    \xi_0(\gamma)^{-1}
    \exp\mleft(\frac{Z(\gamma)}{\lambda_{\mathrm B}}\mright)
    \Phi_{\rho,\gamma}(\lambda_{\mathrm B})
  \]
  across overlaps centered on rays outside the active-ray closure. Choose two nearby rays in this complement together with their negatives; the corresponding four half-planes cover $\mathbb C^\ast$. Conditions RH2 and RH3 allow the polynomial argument of \cite[Proposition~2.5(i)]{bridgeland-2020-7a4fdc8b}, which makes the normalized character identically $1$ there. RH1 gives the same identity on overlaps with $\mathbb H_{\ell_\infty}$ and $\mathbb H_{-\ell_\infty}$, and the identity theorem extends it over both accumulation half-planes.
\end{proof}

Every active class is null and $\xi_0$ is $1$ on $\Gamma_{\mathrm{act}}$. Thus \zcref{ordinary-ray-null-character} gives
\begin{equation}\label{active-character-values}
  x_{\alpha}(\Phi_\rho(\lambda_{\mathrm B}))
  =\exp\mleft(-\frac{Z(\alpha)}{\lambda_{\mathrm B}}\mright),
  \qquad \alpha\in\Gamma_{\mathrm{act}}.
\end{equation}

\begin{lemma}\label{specialized-sector-product-convergence}
  Let $\Delta$ be an acute sector with clockwise boundary rays $\rho_1,\rho_2$. Assume that both lie outside the active-ray closure, or that one does and the other is $\pm\ell_\infty$. After the specialization \eqref{active-character-values}, define for $e\in\Lambda$
  \begin{equation}\label{specialized-sector-product}
    \prod_{\varepsilon\ell_k\subset\Delta}^{\curvearrowright}
    \mleft(
      \frac{1-\exp(-\varepsilon Z(\gamma_k)/\lambda_{\mathrm B})}
      {1+\exp(-\varepsilon Z(\gamma_k)/\lambda_{\mathrm B})}
    \mright)^{-\varepsilon(-1)^k\langle\gamma_k,e\rangle/2}.
  \end{equation}
  The product converges normally on compact subsets of $\mathbb H_{\rho_1}\cap\mathbb H_{\rho_2}$ to a holomorphic nowhere-zero function.
\end{lemma}
\begin{proof}
  If both boundary rays lie outside the active-ray closure, angular separation gives uniform exponential decay of the specialized active characters on compact subsets of the half-plane intersection. Suppose instead that one boundary is $\ell_\infty$, denote the other by $\rho$, and take the positive tail. For every compact $K\subset\mathbb H_\rho\cap\mathbb H_{\ell_\infty}$,
  \[
    \inf_{\lambda_{\mathrm B}\in K}
    \rmRe\mleft(\frac{\pi iw}{\lambda_{\mathrm B}}\mright)>0.
  \]
  Since $Z(\gamma_k)=\pi i(v+kw)$, the positive tail decays exponentially and uniformly on $K$. Changing the signs of the active classes and reflecting gives the negative-tail and $-\ell_\infty$ cases. Together with $\abs{Z(\gamma_k)}\geq C(1+\abs{k})$, this yields constants $c_K,C_K>0$ such that
  \[
    \abs*{\exp(-\varepsilon Z(\gamma_k)/\lambda_{\mathrm B})}
    \leq C_Ke^{-c_K\abs*{k}}
  \]
  for all sufficiently large indices in the tail and every $\lambda_{\mathrm B}\in K$. On the power-series branches, $\abs{\log(1-Y)-\log(1+Y)}\leq C\abs Y$ for small $Y$, while \eqref{crosscap-active-pairings} bounds $\langle\gamma_k,e\rangle$ uniformly in $k$. Hence the logarithms of the tail factors are dominated by the summable series $C'_Ke^{-c_K\abs{k}}$.

  The Weierstrass test gives normal convergence, and exponentiating the sum of the power-series logarithms gives a holomorphic nowhere-zero function.
\end{proof}

\begin{proposition}\label{specialized-sector-product-identification}
  Under the hypotheses of \zcref{specialized-sector-product-convergence}, fix sufficiently large $R$. Wherever the specialization \eqref{active-character-values} lies in $\mathcal U_\Delta(R)$, the product \eqref{specialized-sector-product} equals $\mathbf S(\Delta)^*(x_e)/x_e$. It therefore continues this multiplier holomorphically to $\mathbb H_{\rho_1}\cap\mathbb H_{\rho_2}$.
\end{proposition}
\begin{proof}
  For $\lambda_{\mathrm B}\in\mathbb H_{\rho_1}\cap\mathbb H_{\rho_2}$, consider the condition
  \[
    Z(\alpha)\in\Delta,\ \Omega(\alpha)\neq0
    \Longrightarrow
    \abs*{\exp(-Z(\alpha)/\lambda_{\mathrm B})}
    <e^{-R\norm*{\alpha}}.
  \]
  On each closed subsector of the intersection, the support bound and angular separation give $\eta>0$ such that $\rmRe(Z(\alpha)/\lambda_{\mathrm B})\geq(R+\eta)\norm{\alpha}$ for all relevant active classes and sufficiently small $\lambda_{\mathrm B}$. The specialized point therefore lies in the interior of $\mathcal U_\Delta(R)$. At finite height, \eqref{orientifold-ray-pullback} identifies the product with the corresponding pullback. Normal convergence identifies the limits, and \zcref{specialized-sector-product-convergence} supplies the holomorphic continuation.

  Isotropy makes the factors fix one another's coefficients and commute. Absolute normal convergence then makes the limit independent of the exhaustion and permits regrouping under subdivision.
\end{proof}

\begin{proposition}\label{specialized-sector-map-continuation}
  Under the hypotheses of \zcref{specialized-sector-product-convergence}, the basis multipliers in \zcref{specialized-sector-product-identification} extend near every specialized point over $\mathbb H_{\rho_1}\cap\mathbb H_{\rho_2}$ to an ambient holomorphic sector map with holomorphic inverse.
\end{proposition}
\begin{proof}
  Use the independent torus coordinates $x_{\gamma_0},x_{\gamma_1},x_d,x_{\nu_0},x_{\nu_1}$. The lattice identities give
  \[
    x_{\gamma_{2m}}=x_{\gamma_0}x_d^m,
    \qquad
    x_{\gamma_{2m+1}}=x_{\gamma_1}x_d^m.
  \]
  At the specialized point, angular separation gives $\abs{x_d}<1$ for a positive infinite tail and $\abs{x_d}>1$ for a negative one. These strict inequalities persist on a torus neighborhood, where the tail monomials decay geometrically. After shrinking the neighborhood, the finitely many remaining inequalities also persist. The logarithms defining the two $\Lambda$-multipliers therefore converge normally, and their exponentials give a holomorphic torus map; reversing the logarithms gives its inverse. By \zcref{specialized-sector-product-identification}, this map agrees with $\mathbf S(\Delta)$ on $\mathcal U_\Delta(R)$. For a finite sector, the same argument involves only finitely many factors.
\end{proof}

\begin{theorem}\label{crosscap-full-rh-obstruction}
  Fix $(v,w)\in\mathcal M_+$ and let $\{\Phi_\rho\}$ be indexed by all non-active rays. Assume that it satisfies \textup{(RH1)} of \zcref{bridgeland-rh-problem}, together with \textup{(RH2)} of constant term $\xi_0$ and \textup{(RH3)} away from $\pm\ell_\infty$. Then \textup{(RH2)} fails at both accumulation rays. In particular, \zcref{bridgeland-rh-problem} has no solution with constant term $\xi_0$ for the BPS structure in \zcref{orientifold-bps-structure}.
\end{theorem}

\begin{proof}
  By \zcref{ordinary-ray-null-character}, the active characters have the values \eqref{active-character-values}. Consider first $\ell_\infty$ and write
  \[
    t=\alpha+ib,
    \qquad b>0,
    \qquad
    s=-\frac{\pi i\lambda_{\mathrm B}}{w}.
  \]
  Choose a sufficiently large odd integer $K$ with $\rmRe(t+K)>0$. On the counterclockwise side of $\ell_\infty$, choose a ray $\rho_+$ outside the active-ray closure so that the acute sector between them contains the odd tail $\ell_K,\ell_{K+2},\ldots$. The exponent of the $\nu_1$-multiplier on these rays is $-1/2$, so RH1 and \zcref{specialized-sector-product-identification} give
  \begin{equation}\label{accumulation-delta-tail}
    \frac{\Phi_{\ell_\infty,\nu_1}(\lambda_{\mathrm B})}
    {\Phi_{\rho_+,\nu_1}(\lambda_{\mathrm B})}
    =P_K(s),
  \end{equation}
  on the common half-plane. In terms of $s$,
  \[
    x^{1/2}q^{K/2+m}
    =\exp\mleft(-\frac{\pi^2(t+K+2m)}s\mright),
  \]
  and hence
  \[
    P_K(s)
    =\prod_{m=0}^{\infty}
    \mleft(
      \frac{1-x^{1/2}q^{K/2+m}}{1+x^{1/2}q^{K/2+m}}
    \mright)^{-1/2}.
  \]
  If $\abs{x^{1/2}q^{K/2}}<1$ and $\abs q<1$, expanding the logarithms and summing the geometric progression in $m$ gives
  \begin{equation}\label{accumulation-delta-tail-series}
    \log P_K(s)
    =\sum_{j=0}^{\infty}
    \frac{(x^{1/2}q^{K/2})^{2j+1}}
    {(2j+1)(1-q^{2j+1})}.
  \end{equation}
  Write $t+K=c+ib$ with $c>0$. Choose a subsequence of positive integers, still denoted by $n$, for which
  \[
    e^{i\pi cn}\longrightarrow\zeta,
    \qquad \abs*{\zeta}=1.
  \]
  Choose $\eta>0$ outside the discrete set defined by $\exp(\zeta/(2\pi^2\eta))=1$, and put $s_n=(-in/\pi+\eta e^{-\pi bn})^{-1}$. Then $\rmRe(s_n)>0$, $s_n\to0$, and $\arg(s_n)\to\pi/2$. For large $n$, the point $\lambda_{\mathrm B,n}=iws_n/\pi$ lies in $\mathbb H_{\ell_\infty}\cap\mathbb H_{\rho_+}$. Along this sequence,
  \[
    q=e^{-2\pi^2\eta e^{-\pi bn}},
    \qquad
    1-q\sim2\pi^2\eta e^{-\pi bn},
  \]
  while
  \[
    x^{1/2}q^{K/2}=e^{-\pi bn}e^{i\pi cn}\mleft(1+o(1)\mright).
  \]
  Therefore
  \[
    \frac{x^{1/2}q^{K/2}}{1-q}
    \longrightarrow\frac{\zeta}{2\pi^2\eta}.
  \]
  The terms with $j\geq1$ satisfy
  \[
    \sum_{j=1}^{\infty}
    \abs*{\frac{(x^{1/2}q^{K/2})^{2j+1}}
    {(2j+1)(1-q^{2j+1})}}
    =O(e^{-2\pi bn}).
  \]
  Hence
  \[
    \log P_K(s_n)
    \longrightarrow\frac{\zeta}{2\pi^2\eta}.
  \]
  Our choice of $\eta$ ensures that $P_K(s_n)$ does not tend to $1$.

  Since $Z(\nu_1)=0$ and $\xi_0(\nu_1)=1$, condition~\textup{(RH2)} on the ordinary ray $\rho_+$ gives
  \[
    \Phi_{\rho_+,\nu_1}(\lambda_{\mathrm B,n})
    \longrightarrow1,
  \]
  whereas \textup{(RH2)} on $\ell_\infty$ would give
  \[
    \Phi_{\ell_\infty,\nu_1}(\lambda_{\mathrm B,n})
    \longrightarrow1.
  \]
  This contradicts \eqref{accumulation-delta-tail}. Thus \textup{(RH2)} fails at $\ell_\infty$.

  For $-\ell_\infty$, choose a sufficiently negative odd $K$, set $\lambda_{\mathrm B}=-iws/\pi$, and use the negative tail. The same factors occur with $t+K+2m$ replaced by $-t-K+2m$. Writing $-t-K=c'-ib$ gives $c'>0$, and the sequence $s_n=(in/\pi+\eta e^{-\pi bn})^{-1}$ yields the same contradiction after choosing a corresponding phase subsequence and $\eta$. Hence \textup{(RH2)} also fails at $-\ell_\infty$, without assuming reflection symmetry of the family.
\end{proof}

\begin{proposition}\label{crosscap-full-rh-growth-obstruction}
  Under the hypotheses of \zcref{crosscap-full-rh-obstruction}, the family fails \textup{(RH3)} at each of $\ell_\infty$ and $-\ell_\infty$.
\end{proposition}
\begin{proof}
  Retain $K$, $\rho_+$, and $P_K$ from the preceding proof. Along the central ray $s>0$, put $h=2\pi^2/s$ and $a=(t+K)/2$, so $\rmRe(a)>0$. Writing $r$ for a positive odd integer, \eqref{accumulation-delta-tail-series} becomes
  \[
    \log P_K(s)
    =\sum_{r\textup{ odd}}
    \frac{e^{-har}}{r(1-e^{-hr})}.
  \]
  Since $(1-e^{-u})^{-1}-u^{-1}$ is bounded for $u>0$, as $h\to0^+$ we obtain
  \[
    \log P_K(s)
    =\frac1h\sum_{r\textup{ odd}}\frac{e^{-har}}{r^2}
    +O\mleft(\sum_{r=1}^{\infty}\frac{e^{-h\rmRe(a)r}}{r}\mright)
    =\frac1h\sum_{r\textup{ odd}}\frac{e^{-har}}{r^2}
    +O\mleft(\log\frac1h\mright).
  \]
  Splitting the sum at $r=1/h$ and using $\abs*{e^{-har}-1}\leq C_a\min(hr,1)$ gives
  \[
    \sum_{r\textup{ odd}}\frac{e^{-har}}{r^2}
    =\frac{\pi^2}{8}+O\mleft(h\log\frac1h\mright).
  \]
  Taking real parts and substituting $h=2\pi^2/s$ yields
  \begin{equation}\label{accumulation-delta-tail-growth}
    \log\abs*{P_K(s)}
    =\frac{s}{16}+O(\log s)
    \qquad(s\to+\infty).
  \end{equation}
  If \textup{(RH3)} held at $\ell_\infty$, its bounds there and on $\rho_+$ would impose two-sided polynomial bounds on the quotient in \eqref{accumulation-delta-tail}. \zcref{accumulation-delta-tail-growth} violates the upper bound, so \textup{(RH3)} fails at $\ell_\infty$.

  The negative-tail quotient at $-\ell_\infty$ replaces $t+K+2m$ by $-t-K+2m$, with $-t-K=c'+ib$ and $c'>0$. The same large-$s$ estimate applies and rules out \textup{(RH3)} there.
\end{proof}

The two obstructions occur in different regimes. In \zcref{crosscap-full-rh-obstruction}, $\arg(s_n)\to\pi/2$. This leaves open an origin limit on each closed subsector of $\mathbb H_{\ell_\infty}$. By contrast, \eqref{accumulation-delta-tail-growth} is a central-ray large-parameter obstruction. The modified problem must address both failures.

\subsection{Normalization at the accumulation rays}

For the crosscap structure, $\pm\ell_\infty$ are the only non-active rays in the active-ray closure. We retain the full conditions elsewhere and modify them only at these two rays.

\begin{problem}\label{crosscap-rh-subsector-normalization}
  For the BPS structure in \zcref{orientifold-bps-structure} with constant term $\xi_0$, find a holomorphic map
  \[
    \Phi_\rho\colon\mathbb H_\rho\longrightarrow\mathbb T
  \]
  for every non-active ray $\rho$, and set $\Phi_{\rho,\gamma}=x_\gamma\circ\Phi_\rho$. Require the following.
  \begin{enumerate}
    \item[\textup{(RH1)}] Condition~\textup{(RH1)} of \zcref{bridgeland-rh-problem} holds.
    \item[\textup{(RH2a)}] If $\rho\notin\{\ell_\infty,-\ell_\infty\}$, then, for every $\gamma\in\Gamma$,
      \[
        \exp\mleft(\frac{Z(\gamma)}{\lambda_{\mathrm B}}\mright)
        \Phi_{\rho,\gamma}(\lambda_{\mathrm B})
        \longrightarrow\xi_0(\gamma)
      \]
      as $\lambda_{\mathrm B}\to0$ in the full half-plane $\mathbb H_\rho$.
    \item[\textup{(RH2b)}] If $\rho\in\{\ell_\infty,-\ell_\infty\}$, then the same limit holds as $\lambda_{\mathrm B}\to0$ in every closed subsector of $\mathbb H_\rho$.
    \item[\textup{(RH3$'$)}] If $\rho\notin\{\ell_\infty,-\ell_\infty\}$, then, for every $\gamma\in\Gamma$, there exists $k>0$ such that
      \[
        \abs*{\lambda_{\mathrm B}}^{-k}
        <\abs*{\Phi_{\rho,\gamma}(\lambda_{\mathrm B})}
        <\abs*{\lambda_{\mathrm B}}^k
      \]
      for every $\lambda_{\mathrm B}\in\mathbb H_\rho$ satisfying $\abs{\lambda_{\mathrm B}}\gg0$.
  \end{enumerate}
  No growth condition is imposed on $\mathbb H_{\ell_\infty}$ or $\mathbb H_{-\ell_\infty}$.
\end{problem}

Every solution satisfies the hypotheses of \zcref{ordinary-ray-null-character}. Thus every null-class character is given by $\xi_0(\gamma)\exp(-Z(\gamma)/\lambda_{\mathrm B})$ on every half-plane. In particular, \eqref{active-character-values} holds there.

\begin{proposition}\label{orientifold-rh-uniqueness}
  The problem in \zcref{crosscap-rh-subsector-normalization} has at most one solution.
\end{proposition}
\begin{proof}
  Let $\Phi^{(1)}$ and $\Phi^{(2)}$ be solutions. Their active characters agree by \zcref{ordinary-ray-null-character}. Uncoupledness then makes the ratios of the remaining character coordinates agree on overlaps centered outside the active-ray closure. Conditions \textup{(RH2a)} and \textup{(RH3$'$)} force these ratios to equal $1$. RH1 and the identity theorem extend the equality to $\mathbb H_{\ell_\infty}$ and $\mathbb H_{-\ell_\infty}$.
\end{proof}

Let $\varsigma\colon\mathbb T\to\mathbb T$ be inversion, so $\varsigma^*x_\gamma=x_{-\gamma}$. The ray formula and continued power-series branches satisfy $\mathbf S(-\ell)\circ\varsigma=\varsigma\circ\mathbf S(\ell)$ \cite[Remarks~3.1--3.2]{bridgeland-2020-7a4fdc8b}. If $\{\Phi_\rho\}$ is a solution, then $\lambda_{\mathrm B}\mapsto\varsigma(\Phi_{-\rho}(-\lambda_{\mathrm B}))$ is another one with the same constant term, since $\xi_0(-\gamma)=\xi_0(\gamma)$. By \zcref{orientifold-rh-uniqueness},
\begin{equation}\label{rh-reflection}
  \Phi_{-\rho}(-\lambda_{\mathrm B})=\varsigma\mleft(\Phi_\rho(\lambda_{\mathrm B})\mright).
\end{equation}

For $n\in\mathbb Z$, write
\[
  \mathcal V_n(v,w)
  =\mathbb H_{\ell_{n-1}(v,w)}\cup\mathbb H_{\ell_n(v,w)}.
\]
This is the union of the half-planes centered on non-active rays in $\Sigma_n(v,w)$. We suppress $(v,w)$ when the parameters are fixed.

Fix a non-active ray $\rho_0\subset\Sigma_0$. The clockwise half-turn from $\rho_0$ to $-\rho_0$ crosses $\ell_k$ for $k\geq0$ and $-\ell_{-k}$ for $k\geq1$. Let $\mathcal C_0(v,w)$ be the component of $\mathcal V_0(v,w)\cap-\mathcal V_0(v,w)$ on which $\rmRe(\pi iw/\lambda_{\mathrm B})>0$.

On $\mathcal C_0(v,w)$, inverse half-turn transport of $x_{\nu_0}$ and $x_{\nu_1}$ is represented by the normally convergent products
\begin{align}
  P_B(v,w,\lambda_{\mathrm B})
  &=\prod_{k\geq0}R_k(v,w,\lambda_{\mathrm B})^{-(-1)^k/2}
  \prod_{k\geq1}R_{-k}(v,w,-\lambda_{\mathrm B})^{(-1)^k/2},
  \label{b-half-turn-product}\\
  P_D(v,w,\lambda_{\mathrm B})
  &=\prod_{m\geq0}R_{1+2m}(v,w,\lambda_{\mathrm B})^{1/2}
  \prod_{m\geq0}R_{-1-2m}(v,w,-\lambda_{\mathrm B})^{-1/2}.
  \label{d-half-turn-product}
\end{align}
The inverse appears because RH1 transports from $\rho_0$ to $-\rho_0$, while reflection identifies the terminal character with the inverse value at $-\lambda_{\mathrm B}$. The two tails contain $x^{1/2}q^{k/2}$ for $k\geq0$ and $x^{-1/2}q^{k/2}$ for $k\geq1$. On every compact $K\subset\mathcal C_0(v,w)$, one has $\sup_K\abs{q^{1/2}}<1$, while $x^{1/2}$ and $x^{-1/2}$ remain bounded. The estimate in \zcref{specialized-sector-product-convergence} therefore gives normal convergence.

\zcref{orientifold-ray-pullback} supplies the factors in \eqref{b-half-turn-product}--\eqref{d-half-turn-product} term by term. Here $R_j(v,w,-\lambda_{\mathrm B})$ denotes \eqref{crosscap-ratio} with $-\lambda_{\mathrm B}$ in place of $\lambda_{\mathrm B}$, and each square root is the specialization of the power-series germ at $x_{\gamma_j}=0$. The positive and reflected odd rays consequently carry opposite square-root powers.

\subsection{Sectorial and difference problems}

Evaluation on $\nu_0$ and $\nu_1$ reduces the torus-valued problem to two scalar coordinates on each domain $\mathcal V_n$.

\begin{problem}\label{sectorial-character-problem}
  Fix $(v,w)\in\mathcal M_+$. For each $n\in\mathbb Z$, find $B_n,D_n\in\mathcal O^*(\mathcal V_n)$ satisfying the following.
  \begin{enumerate}
    \item On every closed subsector of $\mathcal V_n$, as $\lambda_{\mathrm B}\to0$,
      \[
        B_n(\lambda_{\mathrm B})\longrightarrow1,
        \qquad
        D_n(\lambda_{\mathrm B})\longrightarrow1.
      \]
    \item For every $n\in\mathbb Z$, some $k>0$ has the following property: on each closed subsector of $\mathcal V_n$,
      \[
        \abs*{\lambda_{\mathrm B}}^{-k}<\abs*{B_n(\lambda_{\mathrm B})}<\abs*{\lambda_{\mathrm B}}^k,
        \qquad
        \abs*{\lambda_{\mathrm B}}^{-k}<\abs*{D_n(\lambda_{\mathrm B})}<\abs*{\lambda_{\mathrm B}}^k
      \]
      whenever $\abs{\lambda_{\mathrm B}}\gg0$.
    \item On $\mathbb H_{\ell_n}=\mathcal V_n\cap\mathcal V_{n+1}$,
      \begin{align}
        B_{n+1}&=B_nR_n(v,w,\lambda_{\mathrm B})^{(-1)^n/2},
        \label{coordinate-b-jump}\\
        D_{n+1}&=D_nR_n(v,w,\lambda_{\mathrm B})^{-(1-(-1)^n)/4}.
        \label{coordinate-d-jump}
      \end{align}
      Each fractional power is obtained by specializing the power-series branch at $x_{\gamma_n}=0$ fixed in \eqref{orientifold-ray-pullback}.
    \item\label{sectorial-half-turn-condition} On $\mathcal C_0(v,w)$,
      \[
        B_0(\lambda_{\mathrm B})B_0(-\lambda_{\mathrm B})=P_B(v,w,\lambda_{\mathrm B}),
        \qquad
        D_0(\lambda_{\mathrm B})D_0(-\lambda_{\mathrm B})=P_D(v,w,\lambda_{\mathrm B}).
      \]
  \end{enumerate}
\end{problem}

\begin{theorem}\label{sectorial-character-equivalence}
  The Riemann--Hilbert problem in \zcref{crosscap-rh-subsector-normalization} is equivalent to the sectorial problem in \zcref{sectorial-character-problem}.
\end{theorem}
\begin{proof}
  Evaluate a torus solution on $\nu_0,\nu_1$ over rays outside the active-ray closure and glue the resulting functions inside $\Sigma_n$. Finitely many such half-planes cover every closed subsector of $\mathcal V_n$, so \textup{(RH2a)} gives the origin limits.

  The sector $\Sigma_{n+1}$ lies clockwise from $\Sigma_n$ across $\ell_n$, hence RH1 gives $\Phi_{n+1}=\mathbf S(\ell_n)(\Phi_n)$. \zcref{orientifold-ray-pullback,crosscap-active-pairings} then yield \eqref{coordinate-b-jump}--\eqref{coordinate-d-jump}. Transport of the fixed branches gives the reflected equations.

  For fixed $n$, choose rays in $\Sigma_{n-1}$, $\Sigma_n$, and $\Sigma_{n+1}$ whose half-planes cover $\mathcal V_n$. Since $Z(\gamma_j)\neq0$,
  \[
    \abs*{R_j(v,w,\lambda_{\mathrm B})}
    \asymp\abs*{\lambda_{\mathrm B}}^{-1},
    \qquad j\in\{n-1,n\},\qquad
    \abs*{\lambda_{\mathrm B}}\gg0.
  \]
  Condition \textup{(RH3$'$)} and the two jump equations give polynomial bounds for $B_n,D_n$ on all three half-planes. Taking the largest exponent gives a single $k$ valid on every closed subsector of $\mathcal V_n$.

  The normally convergent limit of finite-height RH1 transport across the oriented half-turn gives \eqref{b-half-turn-product} and \eqref{d-half-turn-product}.

  Conversely, assign the values \eqref{active-character-values} to $\gamma_0,\gamma_1,d$ and the functions $B_n,D_n$ to $x_{\nu_0},x_{\nu_1}$. For $\rho\subset\Sigma_n$, these five coordinates determine a unique map $\Phi_\rho$ in the $\xi_0$-chart: after division by $\xi_0$ and the exponential $\gamma\mapsto\exp(-Z(\gamma)/\lambda_{\mathrm B})$, they extend multiplicatively to an ordinary character.

  For $\rho\subset-\Sigma_n$, define $\Phi_\rho(\lambda_{\mathrm B})=\varsigma(\Phi_{-\rho}(-\lambda_{\mathrm B}))$ as in \eqref{rh-reflection}. The two ray equations verify RH1 on $x_{\nu_0},x_{\nu_1}$, and every ray map fixes the active basis. On the overlap of the positive and reflected constructions, the half-turn products identify sector transport with this reflected definition. The relevant products are normally convergent and compatible with subdivision by \zcref{specialized-sector-product-convergence,specialized-sector-product-identification,specialized-sector-map-continuation}. Moving the boundary of an acute sector containing $\pm\ell_\infty$ to $\rho_0$ and $-\rho_0$ introduces only finitely many additional factors, while the half-turn identities provide the remaining transport. Thus RH1 holds on a basis and hence on the torus point.

  Each ordinary non-active ray lies in a unique sector $\Sigma_n$ or $-\Sigma_n$. After a small enlargement inside that sector, its half-plane lies in a closed subsector of the corresponding $\mathcal V_n$. The scalar origin limits therefore imply \textup{(RH2a)} on the five basis characters, while the polynomial bounds and the explicit active characters imply \textup{(RH3$'$)}. Multiplicativity extends both conditions to every character.

  To define the map at $\ell_\infty$, let ordinary non-active rays $\rho$ approach it from the side containing the positive tail, and let $\Delta_\rho$ be the intervening acute sector. On $\mathbb H_\rho\cap\mathbb H_{\ell_\infty}$, use RH1 and the normally convergent character product for $\mathbf S(\Delta_\rho)$ from \zcref{specialized-sector-product-convergence,specialized-sector-product-identification,specialized-sector-map-continuation}. These overlaps cover $\mathbb H_{\ell_\infty}$.

  Subdivision and the established ordinary-ray transitions make the local definitions agree, so they glue to a holomorphic map $\Phi_{\ell_\infty}\colon\mathbb H_{\ell_\infty}\to\mathbb T$. Set $\Phi_{-\ell_\infty}(\lambda_{\mathrm B})=\varsigma(\Phi_{\ell_\infty}(-\lambda_{\mathrm B}))$. Reflection and subdivision then verify RH1 for sectors bounded by either accumulation ray.

  For \textup{(RH2b)}, put $t=v/w$ and $s=-\pi i\lambda_{\mathrm B}/w$. Fix a closed subsector of $\mathbb H_{\ell_\infty}$ and choose $K$ so that the transport contains only the positive tail. For sufficiently small $s$ there are constants $c_0,c_1,C>0$ such that
  \[
    \abs*{\exp\mleft(-\frac{\pi^2(t+K+m)}s\mright)}
    \leq C\exp\mleft(-\frac{c_0+c_1m}{\abs*{s}}\mright),
    \qquad m\geq0.
  \]
  This summable majorant bounds the logarithms of every tail multiplier on the basis, so their product tends uniformly to $1$. The adjacent map satisfies \textup{(RH2a)}, and the active characters are exactly \eqref{active-character-values}; hence $\Phi_{\ell_\infty}$ satisfies \textup{(RH2b)}. Reflection gives the condition at $-\ell_\infty$.

  The constructions are inverse. Away from $\pm\ell_\infty$, a torus solution is recovered from its sectorial coordinates. At the accumulation rays, RH1 determines it on adjacent overlaps and the identity theorem determines it on the full half-planes.
\end{proof}

The parameter translation $v\mapsto v+w$ is induced by a lattice automorphism.

\begin{definition}
  Define $\Theta$ on $\Gamma_{\mathrm{act}}$ by
  \begin{align*}
    \Theta(\gamma_0)&=\gamma_1,&
    \Theta(\gamma_1)&=\gamma_0+d,&
    \Theta(d)&=d.
  \end{align*}
  These formulas define an automorphism of $\Gamma_{\mathrm{act}}$ fixing $d$. On $\Lambda$, set
  \[
    (\Theta e)(\alpha)
    =-e(\Theta^{-1}\alpha),
    \qquad
    \alpha\in\Gamma_{\mathrm{act}},\quad e\in\Lambda.
  \]
  In the chosen basis of $\Lambda$,
  \[
    \Theta(\nu_0)=-\nu_0,
    \qquad
    \Theta(\nu_1)=\nu_1-\nu_0.
  \]
\end{definition}

\begin{proposition}\label{crosscap-theta-properties}
  The endomorphism $\Theta$ preserves $\Gamma_{\mathrm{act}}$ and $\Lambda$, is a lattice automorphism, and satisfies
  \[
    \Theta(\gamma_k)=\gamma_{k+1},
    \qquad
    Z_{v+w,w}=Z_{v,w}\circ\Theta.
  \]
  For all $\gamma,\eta\in\Gamma$,
  \[
    \langle\Theta\gamma,\Theta\eta\rangle=-\langle\gamma,\eta\rangle,
    \qquad
    \Omega(\Theta\gamma)=-\Omega(\gamma),
    \qquad
    \xi_0(\Theta\gamma)=\xi_0(\gamma).
  \]
\end{proposition}
\begin{proof}
  The displayed images form a basis; the inverse is
  \begin{align*}
    \Theta^{-1}(\gamma_0)&=\gamma_1-d,&
    \Theta^{-1}(\gamma_1)&=\gamma_0,&
    \Theta^{-1}(d)&=d,\\
    \Theta^{-1}(\nu_0)&=-\nu_0,&
    \Theta^{-1}(\nu_1)&=\nu_1-\nu_0.
  \end{align*}
  \zcref{orientifold-active-classes} gives $\Theta(\gamma_k)=\gamma_{k+1}$. Evaluation on the five basis elements gives $Z_{v+w,w}=Z_{v,w}\circ\Theta$.

  For $\alpha,\beta\in\Gamma_{\mathrm{act}}$ and $e,f\in\Lambda$, the definition on $\Lambda$ gives
  \[
    \begin{aligned}
      \langle(\Theta\alpha,\Theta e),(\Theta\beta,\Theta f)\rangle
      &=-e(\beta)+f(\alpha)\\
      &=-\langle(\alpha,e),(\beta,f)\rangle.
    \end{aligned}
  \]
  The support and signs in \eqref{orientifold-bps-invariants} give $\Omega\circ\Theta=-\Omega$. Finally, for $\alpha\in\Gamma_{\mathrm{act}}$ and $e\in\Lambda$,
  \[
    \xi_0(\Theta\alpha+\Theta e)
    =(-1)^{\langle\Theta\alpha,\Theta e\rangle}
    =(-1)^{-\langle\alpha,e\rangle}
    =\xi_0(\alpha+e).
  \]
\end{proof}

Multiplication of the skew form by $-1$ does not change its reduction modulo $2$, so pullback by $\Theta$ preserves twisted multiplication. It defines a holomorphic automorphism $\Theta_{\mathbb T}$ by $(\Theta_{\mathbb T}g)(\gamma)=g(\Theta\gamma)$. For every acute sector $\Delta$,
\begin{equation}\label{crosscap-sector-map-covariance}
  \Theta_{\mathbb T}\circ\mathbf S_{v,w}(\Delta)
  =\mathbf S_{v+w,w}(\Delta)\circ\Theta_{\mathbb T}.
\end{equation}
At finite height, the sign changes of the skew form and BPS invariants cancel in each Hamiltonian exponent. Passing to the convergent sector product \cite[Proposition~2.2]{bridgeland-2020-7a4fdc8b} proves the covariance.

\begin{proposition}\label{crosscap-solution-translation-covariance}
  Every family $\{\Phi_\rho^{v,w}\}$ of solutions of \zcref{crosscap-rh-subsector-normalization} over $\mathcal M_+$ satisfies, for every non-active ray $\rho\subset\mathbb C^\ast$,
  \[
    \Phi_\rho^{v+w,w}(\lambda_{\mathrm B})
    =\Theta_{\mathbb T}\mleft(\Phi_\rho^{v,w}(\lambda_{\mathrm B})\mright).
  \]
  Its sectorial coordinates satisfy
  \begin{equation}\label{sectorial-translation-covariance}
    B_n^{v+w,w}=\mleft(B_{n+1}^{v,w}\mright)^{-1},
    \qquad
    D_n^{v+w,w}=\frac{D_{n+1}^{v,w}}{B_{n+1}^{v,w}}.
  \end{equation}
\end{proposition}
\begin{proof}
  Since $\ell_k(v+w,w)=\ell_{k+1}(v,w)$, the two problems have the same indexing rays and half-plane domains. \zcref{crosscap-sector-map-covariance} verifies RH1 for $\Theta_{\mathbb T}\circ\Phi_\rho^{v,w}$. For every $\gamma\in\Gamma$,
  \[
    \exp\mleft(\frac{Z_{v+w,w}(\gamma)}{\lambda_{\mathrm B}}\mright)
    x_\gamma\mleft(\Theta_{\mathbb T}(\Phi_\rho^{v,w})\mright)
    =
    \exp\mleft(\frac{Z_{v,w}(\Theta\gamma)}{\lambda_{\mathrm B}}\mright)
    x_{\Theta\gamma}(\Phi_\rho^{v,w})
    \longrightarrow\xi_0(\Theta\gamma)
    =\xi_0(\gamma).
  \]
  This proves preservation of \textup{(RH2a)} and \textup{(RH2b)}. Since $\Theta$ is a lattice automorphism, the transformed growth condition is precisely \textup{(RH3$'$)} for $\Theta\gamma$. Uniqueness now gives the torus-valued covariance. Since $\Sigma_n(v+w,w)=\Sigma_{n+1}(v,w)$, evaluation on $\nu_0,\nu_1$ gives \eqref{sectorial-translation-covariance}; \zcref{sectorial-character-equivalence} transfers it to every sectorial family.
\end{proof}

For $j\in\{0,1\}$, the square of $\Theta$ preserves $\mathbb Z\gamma_j\oplus\mathbb Zd$ and sends $\gamma_j$ to $\gamma_j+d$. Under $\beta\mapsto\gamma_j$ and $\delta\mapsto d$, this is the $m=-1$ lattice action of \cite[Section~3.4]{bridgeland-2020-7a4fdc8b}. The identity $Z_{v+2w,w}=Z_{v,w}\circ\Theta^2$ realizes the corresponding translation of the first central-charge coordinate. For the resolved conifold this action comes from tensoring with $\mathcal O(1)$ \cite[Appendix~A.2]{bridgeland-2020-7a4fdc8b}; a geometric origin for $\Theta^2$ in the present setting is not known.

This covariance reduces the sectorial family to two functions satisfying a difference problem.

\begin{problem}\label{character-difference-problem}
  For every $(v,w)\in\mathcal M_+$, find holomorphic functions
  \[
    B(v,w,-)\in\mathcal O^*(\mathcal V_0(v,w)),
    \qquad
    D(v,w,-)\in\mathcal O^*(\mathcal V_0(v,w)\cup\mathcal V_1(v,w)),
  \]
  satisfying the following. Holomorphicity is fiberwise in $\lambda_{\mathrm B}$, while the displayed shifts govern the parameter dependence.
  \begin{enumerate}
    \item As $\lambda_{\mathrm B}\to0$ in any closed subsector of their respective domains,
      \[
        B(v,w,\lambda_{\mathrm B})\longrightarrow1,
        \qquad
        D(v,w,\lambda_{\mathrm B})\longrightarrow1.
      \]
    \item For fixed $(v,w)$, there exists $k>0$ such that, in every closed subsector of the respective domains,
      \[
        \abs*{\lambda_{\mathrm B}}^{-k}<\abs*{B(v,w,\lambda_{\mathrm B})}<\abs*{\lambda_{\mathrm B}}^k,
        \qquad
        \abs*{\lambda_{\mathrm B}}^{-k}<\abs*{D(v,w,\lambda_{\mathrm B})}<\abs*{\lambda_{\mathrm B}}^k
      \]
      whenever $\abs{\lambda_{\mathrm B}}\gg0$.
    \item The difference equations
      \begin{align}
        B(v+w,w,\lambda_{\mathrm B})B(v,w,\lambda_{\mathrm B})&=R_0(v,w,\lambda_{\mathrm B})^{-1/2},
        \label{b-difference-equation}\\*
        D(v+w,w,\lambda_{\mathrm B})&=D(v,w,\lambda_{\mathrm B})B(v+w,w,\lambda_{\mathrm B}).
        \label{d-difference-equation}
      \end{align}
      The first relation holds on $\mathbb H_{\ell_0(v,w)}=\mathcal V_0(v,w)\cap\mathcal V_0(v+w,w)$, and the second holds on $\mathcal V_1(v,w)$. The half-power is obtained by specializing the power-series branch at $x_{\gamma_0}=0$.
    \item On $\mathcal C_0(v,w)$,
      \[
        B(v,w,\lambda_{\mathrm B})B(v,w,-\lambda_{\mathrm B})=P_B(v,w,\lambda_{\mathrm B}),
        \qquad
        D(v,w,\lambda_{\mathrm B})D(v,w,-\lambda_{\mathrm B})=P_D(v,w,\lambda_{\mathrm B}).
      \]
  \end{enumerate}
\end{problem}

\begin{proposition}\label{sectorial-difference-equivalence}
  Solutions of \zcref{character-difference-problem} are canonically equivalent to $\Theta$-covariant families of solutions of \zcref{sectorial-character-problem}, with covariance \eqref{sectorial-translation-covariance}. Explicitly,
  \begin{align*}
    B_{2m}&=B(v+2mw,w,\lambda_{\mathrm B}),\\
    B_{2m+1}&=B(v+(2m+1)w,w,\lambda_{\mathrm B})^{-1},\\
    D_{2m}&=D(v+2mw,w,\lambda_{\mathrm B}),\\
    D_{2m+1}&=
    \frac{D(v+(2m+1)w,w,\lambda_{\mathrm B})}{B(v+(2m+1)w,w,\lambda_{\mathrm B})}.
  \end{align*}
\end{proposition}
\begin{proof}
  Given a sectorial family, set $B(v,w,-)=B_0^{v,w}$. At $n=0$, \eqref{coordinate-d-jump} states that $D_1^{v,w}=D_0^{v,w}$ on $\mathbb H_{\ell_0(v,w)}$. Thus $D_0^{v,w}$ and $D_1^{v,w}$ glue to a unique function $D(v,w,-)$ on $\mathcal V_0(v,w)\cup\mathcal V_1(v,w)$. Covariance gives
  \[
    B_1^{v,w}=(B_0^{v+w,w})^{-1},
    \qquad
    D_1^{v,w}=\frac{D_0^{v+w,w}}{B_0^{v+w,w}}
  \]
  on $\mathcal V_1(v,w)$. Combining these identities with the $n=0$ jumps gives \eqref{b-difference-equation} and \eqref{d-difference-equation} on their stated domains. Normalization, growth, and half-turn compatibility pass to the restrictions and gluing.

  Conversely, define $B_n,D_n$ by the displayed formulas. The shifted functions lie on $\mathcal V_n$, and \eqref{b-difference-equation}--\eqref{d-difference-equation} give both parity cases of \eqref{coordinate-b-jump}--\eqref{coordinate-d-jump}. They also imply \eqref{sectorial-translation-covariance}. The remaining analytic and half-turn conditions pass to the restrictions, making the construction inverse to the first. Iterating the difference equations gives
  \[
    \frac{D(v+2w,w,\lambda_{\mathrm B})}{D(v,w,\lambda_{\mathrm B})}=R_1(v,w,\lambda_{\mathrm B})^{-1/2},
    \qquad
    \frac{B(v+2w,w,\lambda_{\mathrm B})}{B(v,w,\lambda_{\mathrm B})}
    =R_0(v,w,\lambda_{\mathrm B})^{1/2}R_1(v,w,\lambda_{\mathrm B})^{-1/2}.
  \]
\end{proof}

\subsection{Solution in terms of \texorpdfstring{$F^*$}{F*}}

Normalized double-sine functions solve the difference problem explicitly. We first record the shift and reflection identities needed for the construction.

For $\rmIm(z/\omega_1)>0$, define \cite[Section~5.1]{bridgeland-2020-7a4fdc8b}
\begin{equation}
  F^*\vargs*{z}{\omega_1,\omega_2}
  \coloneqq
  F_2\vargs*{z}{\omega_1,\omega_2}
  \exp\mleft(Q_F\vargs*{z}{\omega_1,\omega_2}\mright),
\end{equation}
where
\[
  Q_F\vargs*{z}{\omega_1,\omega_2}
  =-\frac{\omega_1}{2\pi i\omega_2}
  \Li_2\mleft(e^{2\pi iz/\omega_1}\mright)
  -\frac12\log\mleft(1-e^{2\pi iz/\omega_1}\mright)
  +\frac{\pi i}{12}\frac{\omega_2}{\omega_1}.
\]
The inequality $\rmIm(z/\omega_1)>0$ selects the convergent logarithm.

For $z\in\mathbb C$ and $\omega\in\mathbb C^\ast$, choose logarithms on a simply connected $\lambda_{\mathrm B}$-domain where both factors below are holomorphic and nonzero, and set
\[
  J\vargs*{z}{\omega,\lambda_{\mathrm B}}
  =F^*\vargs*{z}{\omega,-2\lambda_{\mathrm B}}
  F^*\vargs*{2z}{2\omega,-2\lambda_{\mathrm B}}^{-1/2}.
\]
For $(v,w)\in\mathcal M_+$ and $j\in\mathbb Z$, define
\[
  \mathcal D_j=\mathbb H_{\ell_{j-2}}\cup\mathbb H_{\ell_j},
  \qquad
  J_j(\lambda_{\mathrm B})
  =J\vargs*{v+jw}{2w,\lambda_{\mathrm B}}.
\]

The domain $\mathcal D_j$ is simply connected and satisfies $w/\lambda_{\mathrm B},2w/\lambda_{\mathrm B}\notin\mathbb R_{>0}$. Moreover,
$\rmIm((v+jw)/(2w))=\rmIm(v/w)/2>0$, with the same inequality after simultaneous doubling. Applying \zcref{f2-bridgeland-properties} as in \cite[Theorem~5.2]{bridgeland-2020-7a4fdc8b} shows that both factors of $J_j$ are holomorphic and zero-free on $\mathcal D_j$. That theorem also supplies compatible normalized logarithms, closed-subsector limits at the origin, and two-sided polynomial bounds. It follows that $J_j$ is holomorphic and nonvanishing, tends to $1$ on every closed subsector, and has two-sided polynomial growth.

Under $j\mapsto j+2$, the two $F^*$-factors have ratios $(1-x^{1/2}q^{j/2})^{-1}$ and $(1-xq^j)^{-1}$ \cite[Proposition~5.1]{bridgeland-2020-7a4fdc8b}. Combining them with exponent $-1/2$ and factoring $1-xq^j$ gives
\begin{equation}\label{modified-double-sine-difference}
  \frac{J_{j+2}(\lambda_{\mathrm B})}{J_j(\lambda_{\mathrm B})}
  =R_j(v,w,\lambda_{\mathrm B})^{-1/2}.
\end{equation}

On the component of $\mathcal D_j\cap(-\mathcal D_j)$ on which $\abs q<1$, one has
\begin{equation}\label{modified-double-sine-reflection}
  J_j(\lambda_{\mathrm B})J_j(-\lambda_{\mathrm B})
  =\prod_{m\geq0}R_{j+2m}(v,w,\lambda_{\mathrm B})^{1/2}
  \prod_{m\geq1}R_{j-2m}(v,w,-\lambda_{\mathrm B})^{-1/2}.
\end{equation}
The condition $\abs q<1$ is equivalent to $\rmIm(w/\lambda_{\mathrm B})<0$ and also implies $\rmIm(2w/\lambda_{\mathrm B})<0$. Thus the reflection formula of \cite[Proposition~5.1]{bridgeland-2020-7a4fdc8b} applies to both $F^*$-factors. Its products contain $(1-x^{1/2}q^{j/2+m})$ for $m\geq0$ and $(1-x^{-1/2}q^{-j/2+m})$ for $m\geq1$. Divide by the normalized square root of the products with squared arguments and use $1-Y^2=(1-Y)(1+Y)$ to obtain \eqref{modified-double-sine-reflection}.

The clockwise ray ordering gives
\[
  \mathcal V_{2m}\subset\mathcal D_{2m}\cap\mathcal D_{2m+1},
  \qquad
  \mathcal V_{2m+1}\subset\mathcal D_{2m+1}\cap\mathcal D_{2m+2}.
\]
Define
\begin{align}
  D(v,w,\lambda_{\mathrm B})
  &\coloneqq J_1(\lambda_{\mathrm B})
  =J\vargs*{v+w}{2w,\lambda_{\mathrm B}},
  \\
  B(v,w,\lambda_{\mathrm B})
  &\coloneqq\frac{J_1(\lambda_{\mathrm B})}
  {J_0(\lambda_{\mathrm B})}
  =\frac{J\vargs*{v+w}{2w,\lambda_{\mathrm B}}}
  {J\vargs*{v}{2w,\lambda_{\mathrm B}}}.
\end{align}

\begin{theorem}\label{crosscap-rh-solution}
  For every $(v,w)\in\mathcal M_+$, the problem in \zcref{crosscap-rh-subsector-normalization} has a unique solution. Its sectorial coordinates are
  \begin{gather}
    D_{2m}=D_{2m+1}=J_{2m+1},
    \label{crosscap-sectorial-d-solution}\\
    B_{2m}=\frac{J_{2m+1}}{J_{2m}},
    \qquad
    B_{2m+1}=\frac{J_{2m+1}}{J_{2m+2}}.
  \end{gather}
\end{theorem}
\begin{proof}
  The properties of $J_j$ and the domain containments give holomorphicity, nonvanishing, normalization at the origin, and two-sided growth. Translation by $v\mapsto v+w$ increases every index by one. Since \eqref{modified-double-sine-difference} gives $J_2/J_0=R_0(v,w,\lambda_{\mathrm B})^{-1/2}$,
  \begin{align*}
    B(v+w,w,\lambda_{\mathrm B})B(v,w,\lambda_{\mathrm B})
    &=\frac{J_2}{J_1}
    \frac{J_1}{J_0}
    =R_0(v,w,\lambda_{\mathrm B})^{-1/2},\\
    D(v,w,\lambda_{\mathrm B})B(v+w,w,\lambda_{\mathrm B})
    &=J_1\frac{J_2}{J_1}
    =J_2
    =D(v+w,w,\lambda_{\mathrm B}).
  \end{align*}
  These are exactly the two equations in \zcref{character-difference-problem}.

  For $J_1$, \eqref{modified-double-sine-reflection} agrees with \eqref{d-half-turn-product}: its two products run over positive odd rays and reflected negative odd rays with exponents $1/2$ and $-1/2$. Dividing the reflection products of $J_1$ and $J_0$ gives \eqref{b-half-turn-product}. A positive ray of index $k$ contributes exponent $-(-1)^k/2$, while its reflected negative ray contributes $(-1)^k/2$. Hence the functions satisfy the half-turn conditions with the prescribed branches and orientation.

  The displayed sectorial coordinates follow from \zcref{sectorial-difference-equivalence}, and \zcref{sectorial-character-equivalence} reconstructs the torus-valued solution on all half-planes. Uniqueness follows from \zcref{orientifold-rh-uniqueness}.
\end{proof}

The limiting character admits an exact product formula that sharpens the growth estimate in \zcref{crosscap-full-rh-growth-obstruction}.

\begin{proposition}\label{crosscap-accumulation-character-asymptotics}
  Fix $(v,w)\in\mathcal M_+$ and let $\{\Phi_\rho\}$ be the solution in \zcref{crosscap-rh-solution}. Set $D_\infty(\lambda_{\mathrm B})\coloneqq\Phi_{\ell_\infty,\nu_1}(\lambda_{\mathrm B})$. Along the central ray $\lambda_{\mathrm B}=iws/\pi$, with $s>0$, one has
  \begin{equation}\label{crosscap-accumulation-character-product}
    \abs*{D_\infty(iws/\pi)}
    =\frac{e^{s/16}}{\abs*{1+Q^{1/2}}^{1/4}}
    \frac{\abs*{(-Q^{1/2}e^{-s};e^{-s})_\infty}^{1/2}}
    {\abs*{(-Q^{1/2}e^{-2s};e^{-2s})_\infty}},
  \end{equation}
  and hence
  \begin{equation}\label{crosscap-accumulation-character-asymptotic}
    \abs*{D_\infty(iws/\pi)}
    =\abs*{1+Q^{1/2}}^{-1/4}
    e^{s/16}\mleft(1+O(e^{-s})\mright)
    \qquad(s\to+\infty).
  \end{equation}
\end{proposition}
\begin{proof}
  Let $K\subset\mathbb H_{\ell_\infty}$ be compact. For large $m$, it is contained in $\mathcal V_{2m}$. By \zcref{sectorial-character-equivalence,specialized-sector-product-identification}, the quotient $D_\infty/D_{2m}$ is a positive odd Stokes tail, inverted when required by the boundary orientation. This tail tends uniformly to $1$ on $K$ by \zcref{specialized-sector-product-convergence}. \zcref{crosscap-sectorial-d-solution} therefore gives the locally uniform limit
  \[
    D_\infty(\lambda_{\mathrm B})
    =\lim_{m\to\infty}J_{2m+1}(\lambda_{\mathrm B}).
  \]

  Put $a=(t+1)/2$. Simultaneous rescaling of $F_2$ \cite[Proposition~4.1(ii)]{bridgeland-2020-7a4fdc8b}, together with the formula for $Q_F$, shows that $F^*$ is homogeneous. Hence
  \[
    J_{2m+1}
    =F^*\vargs*{a+m}{1,-\check\lambda}
    F^*\vargs*{a+m}{1,-\check\lambda/2}^{-1/2}.
  \]
  Since $\rmIm(a)>0$, all logarithms use their convergent branches and $\abs*{Q^{1/2}}<1$. On the central ray, $\check\lambda=is/\pi$ and $e^{2\pi ia}=-Q^{1/2}$. The product formula \cite[Corollary~6]{narukawa-2004-00c94aa0} and the double-sine identities \cite[Proposition~4.1(ii)--(iv)]{bridgeland-2020-7a4fdc8b} imply, for fixed $s>0$,
  \begin{align*}
    F_2\vargs*{a+m}{1,-\check\lambda}
    &\longrightarrow(-Q^{1/2}e^{-2s};e^{-2s})_\infty^{-1},\\*
    F_2\vargs*{a+m}{1,-\check\lambda/2}
    &\longrightarrow(-Q^{1/2}e^{-s};e^{-s})_\infty^{-1}
  \end{align*}
  as $m\to\infty$. After exchanging the periods, the first arguments in the remaining products tend exponentially to zero.

  The exponential normalizations satisfy
  \[
    Q_F\vargs*{a+m}{1,-\check\lambda}
    -\frac12Q_F\vargs*{a+m}{1,-\check\lambda/2}
    =-\frac14\log(1+Q^{1/2})-\frac{\pi i\check\lambda}{16}.
  \]
  Combining the limits with the exponential normalization and taking absolute values proves \eqref{crosscap-accumulation-character-product}. Here $(u;r)_\infty=\prod_{m\geq0}(1-ur^m)$; the inherited square-root branch does not affect the modulus. As $s\to0$, the expansion
  \[
    \frac12\sum_{n=1}^{\infty}\log(1+Q^{1/2}e^{-ns})
    -\sum_{n=1}^{\infty}\log(1+Q^{1/2}e^{-2ns})
    =\frac14\log(1+Q^{1/2})+O(s)
  \]
  cancels $-\log\abs*{1+Q^{1/2}}/4$ after taking real parts. Since $e^{s/16}\to1$, \eqref{crosscap-accumulation-character-product} tends to $1$, as required by RH2b. As $s\to+\infty$, both products are $1+O(e^{-s})$, proving \eqref{crosscap-accumulation-character-asymptotic}.
\end{proof}

\section{\texorpdfstring{$\tau$}{tau}-functions and resurgent comparison}\label{crosscap-tau-resurgence}

We first construct a differential $\tau$-function for the crosscap solution. We then compare its BPS character multipliers with finite differences of the resurgent partition function.

\subsection{The crosscap \texorpdfstring{$\tau$}{tau}-function}

The physical family has only two parameters, whereas the lattice has rank five. We therefore embed it into an auxiliary miniversal variation, whose basis periods form local coordinates \cite[Definition~3.3 and Section~3.4]{bridgeland-2019-6bdcc39c}, and construct the potential locally over $\mathbb C^\ast_{\lambda_{\mathrm B}}$.

For a piecewise-holomorphic family of Riemann--Hilbert solutions with constant term $\xi_0$, define the normalized characters
\[
  \Psi_{\rho,\eta}(\lambda_{\mathrm B})
  =\xi_0(\eta)^{-1}
  \exp\mleft(\frac{Z(\eta)}{\lambda_{\mathrm B}}\mright)
  \Phi_{\rho,\eta}(\lambda_{\mathrm B}).
\]
A \emph{$\tau$-function} \cite[eq.~(16)]{bridgeland-2019-6bdcc39c} is a piecewise-holomorphic nowhere-zero function, homogeneous under simultaneous rescaling of the periods and $\lambda_{\mathrm B}$, and satisfying
\begin{equation}\label{bps-tau-equations}
  \frac{1}{2\pi i}
  \frac{\partial}{\partial\lambda_{\mathrm B}}
  \log\Psi_{\rho,\eta}
  =\sum_{\alpha}
  \langle\alpha,\eta\rangle
  \frac{\partial}{\partial Z(\alpha)}\log\tau_\rho
\end{equation}
for every $\eta\in\Gamma$, where $\alpha$ ranges over the basis $\gamma_0,\gamma_1,d,\nu_0,\nu_1$.

The radical $\mathbb Zd$ leaves a normalization ambiguity in \eqref{bps-tau-equations}, so we construct the potential locally on $\mathcal M_{\mathrm{aux}}\times\mathbb C^\ast_{\lambda_{\mathrm B}}$. Joyce structures provide another construction of $\tau$-functions \cite{bridgeland-2024-368163f4}; whether this potential has such an intrinsic interpretation remains open.

Let $\mathcal M_{\mathrm{aux}}\subset\mathbb C^5$ consist of tuples $(v_0,v_1,w,a,b)$ with $w\neq0$ and $v_j+2mw\neq0$ for $j=0,1$ and every $m\in\mathbb Z$. Keep the lattice, skew form, and BPS invariants fixed, and define
\[
  \begin{aligned}
    Z(\gamma_0)&=\pi iv_0,&
    Z(\gamma_1)&=\pi iv_1,&
    Z(d)&=2\pi iw,\\
    Z(\nu_0)&=a,&
    Z(\nu_1)&=b.
  \end{aligned}
\]
The deleted hyperplanes are locally finite, so $\mathcal M_{\mathrm{aux}}$ is open, and the displayed linear period map is a local biholomorphism onto the space of central charges. For every compact $K\subset\mathcal M_{\mathrm{aux}}$, there is also $C_K>0$ such that
\[
  \abs*{Z(\gamma_{2m})},\abs*{Z(\gamma_{2m+1})}
  \geq C_K(1+\abs*{m}),
  \qquad m\in\mathbb Z.
\]
For large $\abs m$, the term $2mw$ gives the estimate; the remaining finite set has a positive minimum on $K$. Hence the framed miniversal variation is locally uniformly convergent and remains uncoupled.

No open miniversal neighborhood fixes the complete angular ordering of both infinite towers: varying $v_0$ and $v_1$ independently can exchange rays arbitrarily far toward the accumulation directions. We therefore use the piecewise-holomorphic formulation of \cite[Section~2.5]{bridgeland-2020-7a4fdc8b}. Fix a simply connected open subset of $\mathcal M_{\mathrm{aux}}\times\mathbb C^\ast_{\lambda_{\mathrm B}}$ on which a chosen boundary ray remains non-active and all multiple-sine factors are holomorphic and nonzero. Continue their logarithms from a normalized branch on the physical slice. The boundary ray then determines cutoff indices $r$ and $s$ for the towers $v_0+2mw$ and $v_1+2mw$.

When active rays exchange order, their Hamiltonians Poisson commute by uncoupledness, so the sector product is unchanged. We continue the cutoff data and logarithm branches consistently across overlaps.

The physical family is recovered on the slice
\[
  v_0=v,
  \qquad
  v_1=v+w,
  \qquad
  a=b=0.
\]

Define
\[
  H^\dagger\vargs*{z}{\omega_2,\omega_1}
  \coloneqq
  \frac{G_3\vargs*{z}{\omega_2,\omega_1}}
  {G_3\vargs*{0}{\omega_2,\omega_1}}
  \exp\mleft(
    \mleft(\frac{\omega_1}{2\pi i\omega_2}\mright)^2
    \mleft(\Li_3\mleft(e^{2\pi iz/\omega_1}\mright)-\zeta(3)\mright)
    +\frac{\pi i}{12}\frac{z}{\omega_1}
  \mright).
\]
On the branches chosen for $F^*$, the differential identity \cite[eq.~(70)]{bridgeland-2020-7a4fdc8b} reads
\begin{equation}\label{modified-sine-differential-identity}
  \frac{\partial}{\partial\omega_2}
  \log F^*\vargs*{z}{\omega_1,\omega_2}
  =\frac{\partial}{\partial z}
  \log H^\dagger\vargs*{z}{\omega_2,\omega_1}.
\end{equation}
Set
\begin{equation}\label{crosscap-triple-sine-primitive}
  I\vargs*{z}{\omega,\lambda_{\mathrm B}}
  \coloneqq
  H^\dagger\vargs*{z}{-2\lambda_{\mathrm B},\omega}^{-2}
  H^\dagger\vargs*{2z}{-2\lambda_{\mathrm B},2\omega}^{1/2},
\end{equation}
using the same sectorial logarithm branches as for $J$. Applying \eqref{modified-sine-differential-identity} to both factors and using the chain rule gives
\begin{equation}\label{crosscap-tau-primitive-identity}
  \frac{\partial}{\partial\lambda_{\mathrm B}}
  \log J\vargs*{z}{\omega,\lambda_{\mathrm B}}
  =\frac{\partial}{\partial z}
  \log I\vargs*{z}{\omega,\lambda_{\mathrm B}}.
\end{equation}

On a fixed open subset, write
\begin{align*}
  J_{0,m}&=J\vargs*{v_0+2mw}{2w,\lambda_{\mathrm B}},&
  J_{1,m}&=J\vargs*{v_1+2mw}{2w,\lambda_{\mathrm B}},\\*
  I_{0,m}&=I\vargs*{v_0+2mw}{2w,\lambda_{\mathrm B}},&
  I_{1,m}&=I\vargs*{v_1+2mw}{2w,\lambda_{\mathrm B}}.
\end{align*}
For its cutoff indices $(r,s)$, define
\begin{equation}\label{crosscap-tau-functions}
  \begin{aligned}
    D_{r,s}&=J_{1,s},&
    B_{r,s}&=\frac{J_{1,s}}{J_{0,r}},&
    \tau_{r,s}
    &=\mleft(\frac{I_{0,r}}{I_{1,s}}\mright)^{1/2}.
  \end{aligned}
\end{equation}
The continued logarithms determine the local square-root branch. On the physical slice, the chambers are recovered from
\[
  \begin{aligned}
    (B_{2m},D_{2m},\tau_{2m})
    &=(B_{m,m},D_{m,m},\tau_{m,m}),\\
    (B_{2m+1},D_{2m+1},\tau_{2m+1})
    &=(B_{m+1,m},D_{m+1,m},\tau_{m+1,m}).
  \end{aligned}
\]
The cutoff ratios are
\[
  \frac{B_{r,s}}{B_{r-1,s}}
  =\frac{J_{0,r-1}}{J_{0,r}},
  \qquad
  \frac{B_{r,s}}{B_{r,s-1}}
  =\frac{D_{r,s}}{D_{r,s-1}}
  =\frac{J_{1,s}}{J_{1,s-1}}.
\]
By \eqref{modified-double-sine-difference}, the two ratios are the multipliers for $\gamma_0+(r-1)d$ and $\gamma_1+(s-1)d$. They agree with the specialized ray actions in \eqref{orientifold-ray-pullback} and \eqref{orientifold-mixed-pairing}. Origin normalization and normal convergence of the tails then identify $B_{r,s},D_{r,s}$ with the normalized $x_{\nu_0},x_{\nu_1}$ coordinates.

\begin{proposition}
  Fix a simply connected open subset of $\mathcal M_{\mathrm{aux}}\times\mathbb C^\ast_{\lambda_{\mathrm B}}$ on which the boundary ray remains non-active and the factors in \eqref{crosscap-tau-functions} are holomorphic and nonzero. Continue their logarithms from the normalized physical branch. Then $\tau_{r,s}$ is a $\tau$-function for $B_{r,s},D_{r,s}$. With $z_0=\pi iv_0$ and $z_1=\pi iv_1$, its nontrivial equations are
  \begin{align}
    \frac{1}{2\pi i}
    \frac{\partial}{\partial\lambda_{\mathrm B}}\log B_{r,s}
    &=-\mleft(\frac{\partial}{\partial z_0}
    +\frac{\partial}{\partial z_1}\mright)
    \log\tau_{r,s},
    \label{crosscap-tau-b-equation}\\
    \frac{1}{2\pi i}
    \frac{\partial}{\partial\lambda_{\mathrm B}}\log D_{r,s}
    &=-\frac{\partial}{\partial z_1}
    \log\tau_{r,s}.
    \label{crosscap-tau-d-equation}
  \end{align}
\end{proposition}
\begin{proof}
  Applying \eqref{crosscap-tau-primitive-identity} to $B_{r,s},D_{r,s}$ gives \eqref{crosscap-tau-b-equation}--\eqref{crosscap-tau-d-equation}. Their coefficients are the pairings of $\nu_0,\nu_1$ with $\gamma_0,\gamma_1$ in \eqref{orientifold-mixed-pairing}. The normalized active characters equal $1$, and $\tau_{r,s}$ is independent of $a,b$, so \eqref{bps-tau-equations} also holds on the active basis. For $d$ it reduces to $0=0$, since $d$ spans the radical. Homogeneity follows from that of the modified multiple-sine functions, while holomorphicity gives the mixed-derivative compatibility conditions.
\end{proof}

On an overlap, two local $\tau$-branches for the same normalized sectorial solution differ by a nonzero homogeneous holomorphic function of $Z(d)/\lambda_{\mathrm B}$. The quotient has vanishing logarithmic derivative in every period direction paired nontrivially by \eqref{orientifold-skew-form}, while $\mathbb Zd$ is the radical; conversely, any such factor preserves \eqref{bps-tau-equations}. The ratio in \eqref{crosscap-tau-functions}, with its square root continued from the physical branch, fixes this freedom locally.

On the physical slice, the equation for $D$ uses the transverse variation separating $v_0$ and $v_1$ and cannot be recovered from tangent derivatives alone. For $n=0$, the function is
\begin{equation}\label{physical-crosscap-tau-function}
  \tau(v,w,\lambda_{\mathrm B})
  \coloneqq
  \tau_0
  =\mleft(
    \frac{I\vargs*{v}{2w,\lambda_{\mathrm B}}}
    {I\vargs*{v+w}{2w,\lambda_{\mathrm B}}}
  \mright)^{1/2}.
\end{equation}
The tangential equation \eqref{crosscap-tau-b-equation} reduces to
\[
  \frac{\partial}{\partial\lambda_{\mathrm B}}
  \log B(v,w,\lambda_{\mathrm B})
  =-2\frac{\partial}{\partial v}
  \log\tau(v,w,\lambda_{\mathrm B}).
\]

\subsection{Comparison with the resurgent partition function}

The scalar Stokes transition can now be compared with the BPS character multipliers. Put $t=v/w$, $Q=e^{2\pi it}$, and $\check\lambda=\lambda_{\mathrm B}/w$, so the topological-string coupling is $2\pi\check\lambda$. For any ray $\rho$ avoiding $\{\pm\ell_k\}_{k\in\mathbb Z}\cup\{\pm\ell_\infty\}$, define
\[
  \mathcal F_\rho(\lambda_{\mathrm B},v,w)
  =\mathcal S_{w^{-1}\rho}\widetilde F_\mathrm{cc}
  \mleft(2\pi\check\lambda,t\mright).
\]
The branch convention in \eqref{crosscap-sectorial-sum} fixes the logarithms. Removing the complete $\check\lambda^{-1}$ term gives
\begin{align}
  \widehat{\mathcal F}_\rho(\lambda_{\mathrm B},v,w)
  &\coloneqq
  \mathcal F_\rho(\lambda_{\mathrm B},v,w)
  -\frac{i}{2\pi\check\lambda}\Li_2(e^{\pi it})
  +\frac{i}{8\pi\check\lambda}\Li_2(e^{2\pi it})\notag\\
  &=\int_{w^{-1}\rho}e^{-\xi/\check\lambda}G_\mathrm{cc}(\xi,t)d\xi,
  \label{normalized-crosscap-potential}\\
  \widehat Z_\rho(\lambda_{\mathrm B},v,w)
  &\coloneqq
  \exp\widehat{\mathcal F}_\rho(\lambda_{\mathrm B},v,w).
\end{align}
Since $G_\mathrm{cc}$ is holomorphic at the origin, the integral is holomorphic on its Borel--Laplace domain and has a Poincar\'e expansion in positive powers of $\check\lambda$. Its exponential is therefore holomorphic and nowhere zero.

Choose compatible rays $\rho_n\subset\Sigma_n$ and continuation paths. Along the shifts $v\mapsto v-\varepsilon\lambda_{\mathrm B}$ and $w\mapsto w-\varepsilon\lambda_{\mathrm B}$, transport the active rays, sectors, and summation data continuously. Fix a closed subsector $S\subset\mathcal V_n$ on which the unshifted and both shifted Borel--Laplace integrals are defined. For sufficiently small $\lambda_{\mathrm B}\in S$, the parameter paths remain in $\mathcal M_+$ and all three integrals admit locally uniform asymptotic estimates. The induced transformation $\check\lambda\mapsto\check\lambda/(1-\varepsilon\check\lambda)$ preserves the sign of the imaginary part.

Orient the transition across $\varepsilon\ell_k$ from $\varepsilon\rho_k$ to $\varepsilon\rho_{k+1}$ and set
\begin{equation}\label{crosscap-partition-transition}
  \Xi_{\varepsilon\ell_k}(v,w,\lambda_{\mathrm B})
  \coloneqq
  \frac{\widehat Z_{\varepsilon\rho_{k+1}}(\lambda_{\mathrm B},v,w)}
  {\widehat Z_{\varepsilon\rho_k}(\lambda_{\mathrm B},v,w)},
  \qquad \varepsilon\in\{+1,-1\}.
\end{equation}
At $\varepsilon\ell_\infty$, define $\Xi_{\varepsilon\ell_\infty}$ by the locally uniform limit of the cumulative finite-ray transition. The subtracted terms in \eqref{normalized-crosscap-potential} do not depend on the summation ray, and hence \zcref{crosscap-accumulation-ray-no-jump} gives $\Xi_{\varepsilon\ell_\infty}=1$.

On the common continuation domain, \eqref{crosscap-stokes-jump} gives
\begin{align*}
  \log\Xi_{\varepsilon\ell_k}(v,w,\lambda_{\mathrm B})
  &=
  \widehat{\mathcal F}_{\varepsilon\rho_{k+1}}
  (\lambda_{\mathrm B},v,w)
  -\widehat{\mathcal F}_{\varepsilon\rho_k}
  (\lambda_{\mathrm B},v,w)\\*
  &=
  -\frac{\varepsilon(-1)^k}{2}
  \log\mleft(1+(x^{1/2}q^{k/2})^\varepsilon\mright).
\end{align*}
Therefore
\begin{equation}\label{crosscap-transition-factor}
  \Xi_{\varepsilon\ell_k}(v,w,\lambda_{\mathrm B})
  =
  \mleft(1+(x^{1/2}q^{k/2})^\varepsilon\mright)^{-\varepsilon(-1)^k/2}.
\end{equation}
The two sectorial continuations in \eqref{crosscap-partition-transition} select the power-series branch at $x_{\varepsilon\gamma_k}=0$.

\begin{proposition}
  Fix $\varepsilon\in\{+1,-1\}$ and $k\in\mathbb Z$, and specialize $x_{\varepsilon\gamma_k}=(x^{1/2}q^{k/2})^\varepsilon$. As power-series germs at $x_{\varepsilon\gamma_k}=0$, the $v$-difference of the scalar transition equals the $\nu_0$-character multiplier:
  \begin{equation}\label{crosscap-transition-v-shift}
    \frac{\Xi_{\varepsilon\ell_k}(v,w,\lambda_{\mathrm B})}
    {\Xi_{\varepsilon\ell_k}
    (v-\varepsilon\lambda_{\mathrm B},w,\lambda_{\mathrm B})}
    =
    \frac{
      \mathbf S(\varepsilon\ell_k)^*(x_{\nu_0})
    }{x_{\nu_0}}.
  \end{equation}
  The complementary $w$-difference equals the $\nu_1$-character multiplier:
  \begin{equation}\label{crosscap-transition-w-shift}
    \frac{\Xi_{\varepsilon\ell_k}(v,w,\lambda_{\mathrm B})}
    {\Xi_{\varepsilon\ell_k}(v,w-\varepsilon\lambda_{\mathrm B},\lambda_{\mathrm B})}
    =
    \frac{
      \mathbf S(\varepsilon\ell_k)^*(x_{\nu_1})
    }{x_{\nu_1}}.
  \end{equation}
\end{proposition}

\begin{proof}
  The shift $v\mapsto v-\varepsilon\lambda_{\mathrm B}$ sends $x_{\varepsilon\gamma_k}$ to its negative, whereas $w\mapsto w-\varepsilon\lambda_{\mathrm B}$ multiplies it by $(-1)^k$. At the chosen specialization, $R_k(v,w,\varepsilon\lambda_{\mathrm B})=(1-x_{\varepsilon\gamma_k})/(1+x_{\varepsilon\gamma_k})$. \zcref{crosscap-transition-factor} therefore gives
  \begin{align*}
    \frac{\Xi_{\varepsilon\ell_k}(v,w,\lambda_{\mathrm B})}
    {\Xi_{\varepsilon\ell_k}
    (v-\varepsilon\lambda_{\mathrm B},w,\lambda_{\mathrm B})}
    &=
    R_k(v,w,\varepsilon\lambda_{\mathrm B})^{
    \varepsilon(-1)^k/2},\\
    \frac{\Xi_{\varepsilon\ell_k}(v,w,\lambda_{\mathrm B})}
    {\Xi_{\varepsilon\ell_k}
    (v,w-\varepsilon\lambda_{\mathrm B},\lambda_{\mathrm B})}
    &=
    \begin{cases}
      1,&k\textup{ even},\\
      R_k(v,w,\varepsilon\lambda_{\mathrm B})^{-\varepsilon/2},&k\textup{ odd}.
    \end{cases}
  \end{align*}
  \zcref{orientifold-ray-pullback,crosscap-active-pairings} identify these expressions with the required multipliers.
\end{proof}

At an accumulation ray, \zcref{crosscap-accumulation-ray-no-jump} gives $\Xi_{\varepsilon\ell_\infty}=1$. The collapse of the finite resurgent jumps explains why this agrees with the choice $\Omega(rd)=0$. Let $t\to0$ through the upper half-plane along a path with $\ell_0\to\ell_\infty$. On compatible logarithm branches in $\mathbb H_{\ell_\infty}$, \eqref{crosscap-stokes-jump} gives, for $j\geq1$,
\[
  \lim_{t\to0}\disc_{\ell_j}\widetilde F_\mathrm{cc}
  =-\lim_{t\to0}\disc_{-\ell_{-j}}\widetilde F_\mathrm{cc}
  =\frac{(-1)^j}{2}\log\mleft(1+e^{-\pi ij/\check\lambda}\mright).
\]
The paired terms cancel. The unpaired term with $k=0$ tends to $(\log2)/2$, whose $v$- and $w$-differences vanish; reflection gives the same result at $-\ell_\infty$. Thus the resurgent transition induces no accumulation-ray character multiplier. Since $\Omega(rd)=0$ implies $\DT(\varepsilon\ell_\infty)=0$ and $\mathbf S(\varepsilon\ell_\infty)=\id$, both character identities reduce to $1=1$.

\begin{theorem}\label{crosscap-bps-reconstruction}
  Let $\Gamma_{\mathrm{act}}$ be a finite-rank lattice containing classes $\gamma_k$, and set
  \[
    \Lambda
    =\mathbb Z\nu_0\oplus\mathbb Z\nu_1,
    \qquad
    \Gamma=\Gamma_{\mathrm{act}}\oplus\Lambda.
  \]
  Equip $\Gamma$ with an integral skew form for which both summands are isotropic. Let $Z_{v,w}\colon\Gamma\to\mathbb C$ vanish on $\Lambda$ and satisfy
  \[
    Z_{v,w}(\gamma_k)=\pi i(v+kw),
    \qquad k\in\mathbb Z.
  \]
  For $\varepsilon=+1$, assume that the differences in \eqref{crosscap-transition-v-shift} and \eqref{crosscap-transition-w-shift} are induced by holomorphic ray Hamiltonians supported on positive multiples of $\gamma_k$:
  \[
    H_k=-\sum_{r=1}^{\infty}\DT(r\gamma_k)x_{r\gamma_k}.
  \]
  Suppose also that a lattice automorphism $\Theta$ preserves both summands and satisfies the translation and anti-symplectic properties of \zcref{crosscap-theta-properties}, the pullback covariance \eqref{crosscap-sector-map-covariance}, and the coordinate covariance of \zcref{crosscap-solution-translation-covariance}. Finally, assume that $\operatorname{rank}\Gamma_{\mathrm{act}}$ is minimal under these conditions.

  Then the following statements hold.
  \begin{enumerate}
    \item The lattice $\Gamma_{\mathrm{act}}$ has rank three and $\gamma_{k+2}-\gamma_k$ is independent of $k$. Denote the common difference by $d$. The classes $\gamma_0,\gamma_1,d$ form a basis of the active span, with
      \[
        \gamma_{2m}=\gamma_0+md,
        \qquad
        \gamma_{2m+1}=\gamma_1+md,
        \qquad m\in\mathbb Z.
      \]
      Thus $\Gamma$ has rank five.

    \item There is $p\in\mathbb Z\setminus\{0\}$ such that
      \begin{equation}\label{crosscap-reconstructed-mixed-pairing}
        \begin{aligned}
          \langle \gamma_0,\nu_0\rangle&=p,&
          \langle \gamma_0,\nu_1\rangle&=0,\\
          \langle \gamma_1,\nu_0\rangle&=p,&
          \langle \gamma_1,\nu_1\rangle&=p,\\
          \langle d,\nu_0\rangle&=0,&
          \langle d,\nu_1\rangle&=0.
        \end{aligned}
      \end{equation}
      Consequently,
      \[
        \langle\gamma_j,\nu_0\rangle=p,
        \qquad
        \langle\gamma_j,\nu_1\rangle
        =\frac{1-(-1)^j}{2}p,
        \qquad j\in\mathbb Z.
      \]

    \item Covariance uniquely determines
      \begin{align*}
        \Theta(\gamma_0)&=\gamma_1,&
        \Theta(\gamma_1)&=\gamma_0+d,&
        \Theta(d)&=d,\\
        \Theta(\nu_0)&=-\nu_0,&
        \Theta(\nu_1)&=\nu_1-\nu_0.
      \end{align*}

    \item The $v$-difference uniquely determines
      \[
        \DT(r\gamma_k)=
        \begin{cases}
          \dfrac{(-1)^k}{pr^2},&r\textup{ odd},\\[0.6em]
          0,&r\textup{ even},
        \end{cases}
      \]
      while the $w$-difference gives the induced action on $x_{\nu_1}$. The ray Hamiltonian is
      \[
        H_k
        =-\frac{(-1)^k}{2p}\mleft(
          \Li_2(x_{\gamma_k})-\Li_2(-x_{\gamma_k})
        \mright),
      \]
      and Möbius inversion gives
      \[
        \Omega(\gamma_k)=\frac{(-1)^k}{p},
        \qquad
        \Omega(2\gamma_k)=-\frac{(-1)^k}{4p},
        \qquad
        \Omega(r\gamma_k)=0\quad(r\geq3).
      \]
  \end{enumerate}
\end{theorem}

\begin{proof}
  For $k\in\mathbb Z$, write
  \[
    p_k=\langle\gamma_k,\nu_0\rangle,
    \qquad
    q_k=\langle\gamma_k,\nu_1\rangle.
  \]
  On the specialization \eqref{active-character-values}, the $v$-difference gives
  \[
    -p_k\sum_{r=1}^{\infty}r\DT(r\gamma_k)x_{\gamma_k}^r
    =\frac{(-1)^k}{2}\log R_k(v,w,\lambda_{\mathrm B})
    =-(-1)^k\sum_{r\textup{ odd}}\frac{x_{\gamma_k}^r}{r}.
  \]
  Hence $p_k\neq0$, and comparison of coefficients gives
  \[
    p_k\DT(r\gamma_k)
    =
    \begin{cases}
      \dfrac{(-1)^k}{r^2},&r\textup{ odd},\\[0.6em]
      0,&r\textup{ even}.
    \end{cases}
  \]
  The $w$-difference gives
  \[
    -q_k\sum_{r=1}^{\infty}r\DT(r\gamma_k)x_{\gamma_k}^r
    =-\frac{1-(-1)^k}{4}\log R_k(v,w,\lambda_{\mathrm B}).
  \]
  Comparing the two identities yields
  \[
    \frac{q_k}{p_k}=\frac{1-(-1)^k}{2}.
  \]

  Since $\Theta$ preserves $\Lambda$, write $\Theta(\nu_0)=a\nu_0+b\nu_1$ and $\Theta(\nu_1)=c\nu_0+e\nu_1$. Pullback and coordinate covariance identify their characters with $B^{-1}$ and $D/B$. On a positive even ray, the multiplier pair for $(B,D)$ is $(R_k^{1/2},1)$; on a positive odd ray, it is $(R_k^{-1/2},R_k^{-1/2})$. Comparing exponents gives $(a,b,c,e)=(-1,0,-1,1)$, and hence
  \[
    \Theta(\nu_0)=-\nu_0,
    \qquad
    \Theta(\nu_1)=\nu_1-\nu_0.
  \]
  Translation and anti-symplecticity now give
  \[
    -p_{k+1}
    =\langle\Theta\gamma_k,\Theta\nu_0\rangle
    =-\langle\gamma_k,\nu_0\rangle
    =-p_k.
  \]
  Thus all $p_k$ have the same nonzero integral value $p$, and
  \[
    q_k=\frac{1-(-1)^k}{2}p.
  \]

  The classes $\gamma_0,\gamma_1,\gamma_2$ are linearly independent. If $u_0\gamma_0+u_1\gamma_1+u_2\gamma_2=0$, pairing first with $\nu_1$ gives $u_1p=0$, and pairing with $\nu_0$ then gives $p(u_0+u_2)=0$. Applying $Z_{v,w}$ gives $2u_2w=0$. Thus all coefficients vanish and $\operatorname{rank}\Gamma_{\mathrm{act}}\geq3$. In particular, the relation $\gamma_2-2\gamma_1+\gamma_0=0$ is impossible, since its pairing with $\nu_1$ is $-2p$.

  Let $N\subset\Gamma_{\mathrm{act}}$ be the saturation of the subgroup generated by
  \[
    \gamma_{k+3}-\gamma_{k+2}-\gamma_{k+1}+\gamma_k,
    \qquad k\in\mathbb Z.
  \]
  Each generator has zero central charge and pairs trivially with $\nu_0,\nu_1$. Since $\Gamma_{\mathrm{act}}$ is isotropic, $N$ lies in the radical; it is also $\Theta$-stable because $\Theta(\gamma_k)=\gamma_{k+1}$. The skew form, central charge, Hamiltonians, and $\Theta$-action therefore descend to $(\Gamma_{\mathrm{act}}/N)\oplus\Lambda$ without changing the two finite differences or covariance properties. If $N\neq0$, the quotient would satisfy all hypotheses with smaller rank. Minimality forces $N=0$, and hence
  \[
    \gamma_{k+3}-\gamma_{k+1}
    =\gamma_{k+2}-\gamma_k.
  \]

  Thus $\gamma_{k+2}-\gamma_k$ is independent of $k$. With $d=\gamma_2-\gamma_0$, the active span is generated by $\gamma_0,\gamma_1,d$; independence of $\gamma_0,\gamma_1,\gamma_2$ makes these a basis. Its saturation is $\Theta$-stable and has rank three, and restriction to this saturation together with $\Lambda$ preserves every hypothesis. Minimality therefore gives $\operatorname{rank}\Gamma_{\mathrm{act}}=3$.

  The translation property of \zcref{crosscap-theta-properties} now gives
  \[
    \Theta(\gamma_0)=\gamma_1,
    \qquad
    \Theta(\gamma_1)=\gamma_0+d,
    \qquad
    \Theta(d)=d.
  \]
  Since $\gamma_0,\gamma_1,d$ span $\Gamma_{\mathrm{act}}\otimes_{\mathbb Z}\mathbb Q$, these formulas determine $\Theta$ on the active lattice, and the formulas for $\nu_0,\nu_1$ determine it on $\Gamma$.

  Since $d=\gamma_2-\gamma_0$, the formulas for $p_k,q_k$ give $\langle d,\nu_0\rangle=\langle d,\nu_1\rangle=0$. This proves \eqref{crosscap-reconstructed-mixed-pairing} and its extension to all $\gamma_j$.

  Substitution of $p_k=p$ gives the formula for $\DT(r\gamma_k)$, and summing the odd powers gives $H_k$.

  Möbius inversion of the multicover relation gives
  \[
    \Omega(r\gamma_k)
    =\sum_{s\mid r}\frac{\mu(s)}{s^2}\DT\mleft(\frac rs\gamma_k\mright)
    =\frac{(-1)^k}{pr^2}
    \sum_{\substack{s\mid r\\ r/s\textup{ odd}}}\mu(s).
  \]
  The final divisor sum equals $1$ for $r=1$, $-1$ for $r=2$, and $0$ for $r\geq3$. This gives the stated invariants and completes the proof.
\end{proof}

The integer $p$ measures an overall scaling of the Poisson structure. For any $p\in\mathbb Z\setminus\{0\}$, keep the central charge and ray support fixed and use the skew form and invariants of \zcref{crosscap-bps-reconstruction}. In the ordinary torus chart determined by $\xi_0(\alpha+e)=(-1)^{\langle\alpha,e\rangle}$,
\[
  \{-,-\}_p=p\{-,-\}_1,
  \qquad
  H_{k,p}=\frac1pH_{k,1}.
\]
Therefore $\{H_{k,p},-\}_p=\{H_{k,1},-\}_1$. The Stokes maps, character multipliers, sectorial coordinates, difference functions, and scalar Riemann--Hilbert solution are all independent of $p$.

The value $p=-1$ recovers \eqref{orientifold-mixed-pairing}. Other nonzero values rescale the integral skew form and inversely rescale the rational invariants. They preserve the scalar solution but change the normalization of its differential $\tau$-function.

The physical $\tau$-function has the following asymptotic relation to the formal crosscap free energy.

\begin{proposition}\label{crosscap-tau-free-energy-relation}
  Let $K$ be a compact subset of the region in $\mathcal M_+$ where $0<\rmRe(v/w)<1$, and let $S$ be a closed $\lambda_{\mathrm B}$-subsector such that $\rmRe(\lambda_{\mathrm B}/w)<0$ on $K\times S$. Uniformly for $(v,w)\in K$ and $\lambda_{\mathrm B}\to0$ through $S$,
  \begin{equation}
    \log\tau\mleft(v+\lambda_{\mathrm B},w,\lambda_{\mathrm B}\mright)+\log\tau\mleft(v,w,\lambda_{\mathrm B}\mright)\sim\frac{\pi i}{16}+\frac12\left(\widetilde F_\mathrm{cc}\mleft(\frac{2\pi\lambda_{\mathrm B}}{w},\frac{v+\lambda_{\mathrm B}}{w}\mright)-\widetilde F_\mathrm{cc}\mleft(\frac{2\pi\lambda_{\mathrm B}}{w},\frac{v}{w}\mright)\right).
  \end{equation}
  The shifted occurrence of $\widetilde F_\mathrm{cc}$ is expanded coefficientwise in $s=\lambda_{\mathrm B}/w$ at the origin. The difference contains no negative powers, and $\sim$ denotes the resulting Poincar\'e expansion: truncation before degree $N$ leaves $O(\abs{\lambda_{\mathrm B}}^N)$ uniformly on the stated sets.
\end{proposition}
\begin{proof}
  Put $s=\lambda_{\mathrm B}/w$ and rescale all periods by $w^{-1}$. The four $H^\dagger$-factors in \eqref{crosscap-triple-sine-primitive} and \eqref{physical-crosscap-tau-function} then lie in the domain of the asymptotic expansion in \cite[Propositions~4.6--4.7, Theorem~1.2, and eqs.~(43), (49), and~(69)]{bridgeland-2020-7a4fdc8b}; its contour estimate is uniform on compact subsets. The logarithmic and constant-map terms cancel among the four factors, while the linear terms contribute $\pi i/32$. Expanding the remaining polylogarithms, using $\abs*{Q^{1/2}}<1$, gives
  \[
    \log\tau(v,w,\lambda_{\mathrm B})
    -\frac{\pi i}{32}
    \sim-\frac18\sum_{r\textup{ odd}}
    \frac{Q^{r/2}}{r}\sec^2\mleft(\frac{\pi rs}{2}\mright),
  \]
  where $r$ ranges over positive odd integers. Compactness of $K$ gives the same expansion after replacing $v$ by $v+\lambda_{\mathrm B}$, after reducing the radial bound of $S$ if needed. This shift sends $Q^{1/2}$ to $e^{\pi is}Q^{1/2}$, and hence
  \[
    \log\tau(v+\lambda_{\mathrm B},w,\lambda_{\mathrm B})+\log\tau(v,w,\lambda_{\mathrm B})-\frac{\pi i}{16}\sim-\frac18\sum_{r\textup{ odd}}\frac{Q^{r/2}}{r}\left(1+e^{\pi irs}\right)\sec^2\mleft(\frac{\pi rs}{2}\mright).
  \]

  The formal odd-projected crosscap free energy satisfies
  \[
    \widetilde F_\mathrm{cc}\mleft(2\pi s,\frac vw\mright)
    \sim i\sum_{r\textup{ odd}}
    \frac{Q^{r/2}}{2r\sin(\pi rs)}.
  \]
  Therefore
  \[
    \frac12\left(\widetilde F_\mathrm{cc}\mleft(2\pi s,\frac vw+s\mright)-\widetilde F_\mathrm{cc}\mleft(2\pi s,\frac vw\mright)\right)=\frac{i}{4}\sum_{r\textup{ odd}}\frac{Q^{r/2}\left(e^{\pi irs}-1\right)}{r\sin(\pi rs)},
  \]
  as a formal series in $s$. The apparent terms of degree $-1$ cancel because $e^{\pi irs}-1$ vanishes to first order.

  The elementary identity
  \[
    -\frac{1+e^{ia}}{8}
    \sec^2\mleft(\frac a2\mright)
    =
    \frac{i(e^{ia}-1)}{4\sin a}
  \]
  with $a=\pi rs$ identifies the two multicover expansions term by term. For every power of $s$, finite sums of $\sum_{r=1}^\infty r^m\abs*{Q^{1/2}}^r$ bound the coefficients on compact parameter sets. Thus all coefficientwise operations and remainder estimates are locally uniform.
\end{proof}

For the chosen compatible rays $\rho_n\subset\Sigma_n$, set
\begin{equation}\label{raw-crosscap-sectorial-tau}
  \tau_n^{\mathrm{raw}}(v,w,\lambda_{\mathrm B})
  \coloneqq
  \widehat Z_{\rho_n}(\lambda_{\mathrm B},v,w).
\end{equation}
Along the same transported parameter paths, define
\begin{equation}\label{resurgent-crosscap-coordinate-ratios}
  \begin{aligned}
    B_n^{\mathrm{raw}}
    &\coloneqq
    \frac{\tau_n^{\mathrm{raw}}(v,w,\lambda_{\mathrm B})}
    {\tau_n^{\mathrm{raw}}(v-\lambda_{\mathrm B},w,\lambda_{\mathrm B})},\\
    D_n^{\mathrm{raw}}
    &\coloneqq
    \frac{\tau_n^{\mathrm{raw}}(v,w,\lambda_{\mathrm B})}
    {\tau_n^{\mathrm{raw}}(v,w-\lambda_{\mathrm B},\lambda_{\mathrm B})}.
  \end{aligned}
\end{equation}

\begin{proposition}\label{resurgent-crosscap-local-coordinate-data}
  On their common continuation domains, the coordinate germs in \eqref{resurgent-crosscap-coordinate-ratios} satisfy the finite-ray jumps
  \begin{equation}\label{raw-crosscap-coordinate-jumps}
    \begin{aligned}
      \frac{B_{n+1}^{\mathrm{raw}}}{B_n^{\mathrm{raw}}}
      &=R_n(v,w,\lambda_{\mathrm B})^{(-1)^n/2},\\
      \frac{D_{n+1}^{\mathrm{raw}}}{D_n^{\mathrm{raw}}}
      &=
      \begin{cases}
        1,&n\textup{ even},\\
        R_n(v,w,\lambda_{\mathrm B})^{-1/2},&n\textup{ odd}.
      \end{cases}
    \end{aligned}
  \end{equation}
  On every closed subsector of a common Borel--Laplace domain, they also satisfy, locally uniformly,
  \[
    B_n^{\mathrm{raw}}\longrightarrow1,
    \qquad
    D_n^{\mathrm{raw}}\longrightarrow1.
  \]
\end{proposition}
\begin{proof}
  Put $s=\check\lambda$.
  \zcref{crosscap-partition-transition,resurgent-crosscap-coordinate-ratios,crosscap-transition-v-shift,crosscap-transition-w-shift} give
  \begin{align*}
    \frac{B_{n+1}^{\mathrm{raw}}}{B_n^{\mathrm{raw}}}
    &=\frac{\Xi_{\ell_n}(v,w,\lambda_{\mathrm B})}{\Xi_{\ell_n}(v-\lambda_{\mathrm B},w,\lambda_{\mathrm B})}
    =R_n(v,w,\lambda_{\mathrm B})^{(-1)^n/2},\\
    \frac{D_{n+1}^{\mathrm{raw}}}{D_n^{\mathrm{raw}}}
    &=\frac{\Xi_{\ell_n}(v,w,\lambda_{\mathrm B})}{\Xi_{\ell_n}(v,w-\lambda_{\mathrm B},\lambda_{\mathrm B})}
    =
    \begin{cases}
      1,&n\textup{ even},\\
      R_n(v,w,\lambda_{\mathrm B})^{-1/2},&n\textup{ odd}.
    \end{cases}
  \end{align*}
  This proves \eqref{raw-crosscap-coordinate-jumps}. The regular part of the formal crosscap expansion begins with
  \begin{gather}
    \widetilde F_\mathrm{cc}(2\pi s,t)
    -\frac{i}{2\pi s}\Li_2(Q^{1/2})
    +\frac{i}{8\pi s}\Li_2(Q)
    \sim\sum_{g=1}^\infty a_g(t)s^{2g-1},
    \label{regular-crosscap-series}\\
    a_1(t)
    =\frac{i\pi}{12}\mleft(\Li_0(Q^{1/2})-\Li_0(Q)\mright)
    =\frac{i\pi}{12}\frac{Q^{1/2}}{1-Q}.
    \label{regular-crosscap-leading-coefficient}
  \end{gather}
  By \eqref{normalized-crosscap-potential}, this is the Poincar\'e expansion of $\log\tau_n^{\mathrm{raw}}$. The shift $v\mapsto v-\lambda_{\mathrm B}$ sends $t$ to $t-s$, so $\log B_n^{\mathrm{raw}}=a_1'(t)s^2+O(s^3)$.

  For $D_n^{\mathrm{raw}}$, the shift $w\mapsto w-\lambda_{\mathrm B}$ sends
  \[
    (s,t)\longmapsto
    \mleft(\frac{s}{1-s},\frac{t}{1-s}\mright).
  \]
  Substitution in \eqref{regular-crosscap-series} gives $\log D_n^{\mathrm{raw}}=O(s^2)$. Both coordinates therefore tend to $1$.
\end{proof}

\begin{proposition}\label{resurgent-crosscap-half-turn-obstruction}
  Suppose that $B_0^{\mathrm{raw}}$ and $D_0^{\mathrm{raw}}$ extend holomorphically to $\mathcal C_0(v,w)$. Both half-turn identities in \zcref[noname]{sectorial-half-turn-condition} then fail, so \eqref{raw-crosscap-sectorial-tau} does not determine a Riemann--Hilbert solution.
\end{proposition}
\begin{proof}
  Choose logarithms vanishing at the origin. The series \eqref{regular-crosscap-series} is odd in $s$, so reflection in $\lambda_{\mathrm B}$ gives
  \begin{align*}
    \log\mleft(
      B_0^{\mathrm{raw}}(v,w,\lambda_{\mathrm B})
      B_0^{\mathrm{raw}}(v,w,-\lambda_{\mathrm B})
    \mright)
    &\sim
    \sum_{g=1}^\infty
    \mleft(a_g(t+s)-a_g(t-s)\mright)s^{2g-1}\\*
    &=2a_1'(t)s^2+O(s^4).
  \end{align*}
  Differentiating \eqref{regular-crosscap-leading-coefficient} gives
  \[
    a_1'(t)
    =-\frac{\pi^2}{12}
    \frac{Q^{1/2}(1+Q)}{(1-Q)^2}.
  \]
  Since $\abs*{Q^{1/2}}<1$ on $\mathcal M_+$, this derivative never vanishes. Thus the logarithm of the $B_0^{\mathrm{raw}}$ half-turn product has a nonzero quadratic term.

  The same expansion gives $\log\tau_0^{\mathrm{raw}}=a_1(t)s+O(s^3)$. Applying $(s,t)\mapsto(s/(1-s),t/(1-s))$ and its reflection yields
  \begin{align*}
    \log\mleft(
      D_0^{\mathrm{raw}}(v,w,\lambda_{\mathrm B})
      D_0^{\mathrm{raw}}(v,w,-\lambda_{\mathrm B})
    \mright)
    &=-2\mleft(a_1(t)+ta_1'(t)\mright)s^2+O(s^4)\\
    &=\frac{\pi}{12\sin(\pi t)}
    \mleft(1-\pi t\cot(\pi t)\mright)s^2+O(s^4).
  \end{align*}

  The coefficient for $D_0^{\mathrm{raw}}$ could vanish only if $z=\pi t$ were a nonzero root of $\tan z=z$. Then $u(r)=\sin(zr)$ would solve the self-adjoint problem $-u''=z^2u$, $u(0)=0$, and $u'(1)=u(1)$. Equality in $\abs*{u(1)}^2\leq\int_0^1\abs*{u'(r)}^2dr$ would make $u$ linear, forcing $z=0$. For a nonzero root, integration by parts therefore gives
  \[
    z^2\int_0^1\abs*{u(r)}^2dr
    =\int_0^1\abs*{u'(r)}^2dr-\abs*{u(1)}^2>0.
  \]
  Hence every nonzero root is real, contrary to $\rmIm(t)>0$.

  Compare these algebraic terms with the Stokes products. On every closed subsector compactly contained in $\mathcal C_0(v,w)$, the tails $x^{1/2}q^{k/2}$ for $k\geq0$ and $x^{-1/2}q^{k/2}$ for $k\geq1$ are bounded by $C\exp(-(c_0+c_1k)/\abs s)$. The estimate in \zcref{specialized-sector-product-convergence} and normal convergence of \eqref{b-half-turn-product}--\eqref{d-half-turn-product} give, for some $c>0$,
  \[
    \log P_B(v,w,\lambda_{\mathrm B})=O\mleft(e^{-c/\abs*{s}}\mright),
    \qquad
    \log P_D(v,w,\lambda_{\mathrm B})=O\mleft(e^{-c/\abs*{s}}\mright).
  \]
  Both Stokes products are therefore exponentially flat and have zero Poincar\'e expansion. They cannot equal the algebraic half-turn products computed here.
\end{proof}

For compatible sectorial branches, \eqref{sosp-free-energy-large-n} factorizes the scalar transition into the inverse square root of the resolved-conifold transition and the $\epsilon$-th power of the crosscap transition. Their finite differences combine the ordinary identities \cite[Corollary~4.13 and Remark~4.14]{alim-2023-de7fb1f0} with \eqref{crosscap-transition-v-shift}--\eqref{crosscap-transition-w-shift}. A torus-valued factorization additionally requires the mixed pairings.

\section{Conclusion and discussion}
Starting from the exact crosscap product, we determined its Borel transform, sectorial sums, Stokes jumps, and accumulation limits. The associated rank-five BPS structure has Stokes tails that violate both unrestricted normalization and polynomial growth at the non-active accumulation rays. Replacing normalization there by closed-subsector limits and omitting growth yields a unique solution in normalized double-sine functions, together with an exact central-ray asymptotic. Its scalar finite differences recover the active lattice, rational invariants, and mixed pairing within the two-coordinate ansatz, while the modified triple sine supplies a local differential $\tau$-function and its asymptotic relation to the crosscap free energy.

An orientifold Donaldson--Thomas construction of the rank-five lattice and its mixed ordinary--crosscap pairing remains open. Such a construction would decide whether accumulation-ray classes and a further covector paired with $d$ occur. Refined Chern--Simons multiple-sine formulas distinguish the three non-orientable sectors by their single-crosscap factors \cite[Sections~2--3]{avetisyan-2022-e99b6793}; refined dual partition functions exhibit BPS-ray Stokes automorphisms \cite[Section~4.2]{alexandrov-2024-52f5f751}, and the corresponding nonperturbative free energy reproduces both the refined transseries and these automorphisms \cite[Section~4]{chuang-2025-741da84d}.

\section*{Declarations}

\subsection*{Funding}
The authors were partially supported by Taiwan NSTC grant 114-2115-M-002-010 and NTU Core Consortiums grants 113L893603, 114L891803, and 115L890503 (TIMS).

\subsection*{Author contributions}
Both authors contributed to and approved the manuscript.

\subsection*{Data availability}
This study generated and analysed no datasets.

\subsection*{Competing interests}
The authors have no relevant financial or non-financial interests to disclose.

\printbibliography
\end{document}